%% file: main.tex
\newif\ifllncs\llncsfalse
\newif\ifanon\anonfalse

\ifllncs
\documentclass{llncs}
\else
\documentclass[11pt,letter]{article}
\fi 

\usepackage{breakcites}
\usepackage[utf8]{inputenc}
\usepackage[T1]{fontenc}
\ifllncs
\else
\usepackage{times}
\usepackage{fullpage}
\fi 

\usepackage{mdframed}
\usepackage{amsmath}
\usepackage{relsize}
\usepackage{amsfonts}
\usepackage{amssymb}
\usepackage[dvipsnames]{xcolor}
\usepackage{graphicx}
\usepackage{braket}
\usepackage{url}
\usepackage{bm}
\usepackage{tikz}
\usetikzlibrary{shapes,arrows,positioning,calc}
\usepackage[mathscr]{euscript}
\allowdisplaybreaks

\usepackage{soul}
\usepackage{multirow}
\definecolor{DarkBlue}{RGB}{0,0,150}
\definecolor{NotSoDarkBlue}{RGB}{15,15,210}
\definecolor{DarkRed}{RGB}{150,0,0}
\definecolor{DarkGreen}{RGB}{0,100,0}
\usepackage[pdfstartview=FitH,pdfpagemode=UseNone,colorlinks,linkcolor=DarkRed,filecolor=blue,citecolor=DarkRed,urlcolor=DarkRed,pagebackref,breaklinks]{hyperref}
\usepackage[ruled, algosection]{algorithm2e}
\usepackage{cleveref}

\usepackage{orcidlink}

\usetikzlibrary{quantikz}

\newcommand{\CC}{\mathbb{C}}

\newcommand{\NN}{\mathbb{N}}

\newcommand{\EE}{\mathbb{E}}

\newcommand{\poly}{\mathsf{poly}}

\newcommand{\fullcirc}{\boldsymbol{\cdot}}

\def\reconstruct{\mathsf{Reconstruct}}

\newcommand{\norm}[1]{\left\| {#1} \right\|}

\renewcommand{\vec}[1]{\mathbf{#1}} 
\newcommand{\kett}[1]{\ket{#1}\hspace{-0.8mm}\rangle}

\usepackage{amsthm}
\newtheorem{theorem}{Theorem}

\newtheorem{lemma}[theorem]{Lemma}
\newtheorem{corollary}[theorem]{Corollary}
\newtheorem{definition}[theorem]{Definition}
\newtheorem{remark}{Remark}

\newtheorem*{example}{Example}
\newtheorem*{theorem*}{Theorem}
\numberwithin{theorem}{section}
\numberwithin{conjecture}{section}
\numberwithin{problem}{section}
 \newtheorem{construction}{Construction}

\newmdtheoremenv[backgroundcolor=gray!10,
                 linewidth=0pt,
                 innerleftmargin=16pt,
                 innerrightmargin=16pt,
                 innertopmargin=6pt,
                 innerbottommargin=6pt,
            splitbottomskip=4pt]{protocol}[prot]{Game}

\newcommand{\N}{\mathbb{N}}
\newcommand{\bit}{\{0,1\}}
\newcommand\algo{\mathcal}
\newcommand{\Uf}{\mathsf{U}_f}
\newcommand{\reg}[1]{{\color{gray}\mathsf{#1}}}
\newcommand{\linear}{\mathrm{L}}
\newcommand{\state}{\mathrm{S}}
\newcommand{\negl}{\mathsf{negl}}
\newcommand{\ketbra}[2]{\left|#1\right\rangle\!\!\left\langle #2\right|}
\newcommand{\Tr}[1]{\mathrm{Tr}\left[#1 \right]}  
\newcommand{\ot}{\otimes}
\newcommand\id{\mathbb{I}}

\newcommand{\KeyGen}{\mathsf{KeyGen}}
\newcommand{\Enc}{\mathsf{Enc}}
\newcommand{\Dec}{\mathsf{Dec}}
\newcommand{\secp}{\lambda}

\newcommand{\Type}{\mathsf{Type}}
\newcommand{\BinType}{\mathsf{BinType}}

\newcommand{\sUE}{\mathsf{UE}}

\newcommand{\sfU}{\mathsf{U}}

\newcommand{\blambda}{\boldsymbol{\lambda}}
\newcommand{\bmu}{\boldsymbol{\mu}}

\newif\ifnotes
\notesfalse
\newcommand{\authnote}[3]{\textcolor{#3}{[{\footnotesize {\bf #1:} { {#2}}}]}}
\newcommand{\vinod}[1]{\ifnotes \authnote{V}{#1}{blue} \fi}

\newcommand{\seyoon}[1]{\ifnotes \authnote{S}{#1}{ForestGreen} \fi}

\newif\ifedits
\editsfalse
\newcommand{\itcsedit}[1]{\ifedits\textcolor{blue}{#1}\else{#1}\fi}

\title{Cloning Games, Black Holes and Cryptography \vspace{0.3cm}}
\ifanon
    \author{Anonymous ITCS 2026 Submission}
\else
    \author{Alexander Poremba\footnote{\url{poremba@mit.edu}}\\MIT 
    \and Seyoon Ragavan\footnote{\url{sragavan@mit.edu}}\\MIT
    \and Vinod Vaikuntanathan\footnote{\url{vinodv@mit.edu}}\\MIT}
\fi
\date{ }

\def\vecy{\mathbf{y}}

\def\vecx{\mathbf{x}}

\def\matA{\mathbf{A}}

\def\vecz{\mathbf{z}}

\def\vecc{\mathbf{c}}
\def\vecv{\mathbf{v}}
\def\vecw{\mathbf{w}}
\def\veca{\mathbf{a}}
\def\vecb{\mathbf{b}}
\def\vecd{\mathbf{d}}

\begin{document}
\maketitle

\begin{abstract}

In this work, we introduce a new toolkit for analyzing \emph{cloning games}, a notion that captures stronger and more quantitative versions of the celebrated quantum no-cloning theorem.
This framework allows us to analyze a new cloning game based on \emph{binary phase states}.
Our results provide evidence that these games may be able to overcome important limitations of previous candidates based on BB84 states and subspace coset states:
in a model where the adversaries are restricted to making a single oracle query, we show that the binary phase variant is $t$-copy secure when $t=o(n/\log n)$.
Moreover, for constant $t$, we obtain the \emph{first} optimal bounds of $O(2^{-n})$, asymptotically matching the value attained by a trivial adversarial strategy.
We also show a worst-case to average-case reduction which allows us to show the same quantitative results for the new and natural notion of \emph{Haar cloning games}.

Our analytic toolkit, which we believe will find further applications, is based on binary subtypes and uses novel bounds on the operator norms of block-wise tensor products of matrices. To illustrate the effectiveness of these new techniques, we present two applications: first, in black-hole physics, where our asymptotically optimal bound offers quantitative insights into information scrambling in idealized models of black holes; and second, in unclonable cryptography, where we (a) construct succinct unclonable encryption schemes from the existence of pseudorandom unitaries, and (b) propose and provide evidence for the security of multi-copy unclonable encryption schemes.

\end{abstract}

\newpage

\thispagestyle{empty}

\newpage

\pagenumbering{roman}
\tableofcontents

\newpage
\pagenumbering{arabic}

\input{intro-overview}

\input{prelims}

\input{monogamy-and-cloning-games-defns}

\input{binary-phase-technical}

\input{tfkw-salted-phase}

\input{worsttoav}

\input{black-hole-games}

\input{unclonable-enc}

\bibliographystyle{alpha}
\bibliography{main}

\appendix

\section{Relating Cloning Games to Monogamy-Like Games}\label{sec:moetocloning}

In this section, we show that $t \mapsto t+1$ cloning games can be recast as a particular variant of a monogamy of entanglement game, as defined in~\Cref{sec:MOE}. We first begin by presenting this argument when $t = 1$, which essentially follows by an argument laid out by Broadbent and Lord~\cite{broadbent_et_al:LIPIcs.TQC.2020.4}. The additional structure we impose on the MOE game is as follows:
\begin{itemize}
        \item The tripartite state $\rho \in \algo D(\algo H_{\reg{A}} \otimes \algo H_{\reg{B}} \otimes \algo H_{\reg{C}})$ which is shared between Alice, Bob and Charlie is the result of applying a cloning channel $\Phi_{\reg{A' \rightarrow BC}}$ to one half of an EPR pair, i.e.,
    \begin{equation*}
    \rho_{\reg{ABC}} = (\id_{\reg A} \otimes \Phi_{\reg{A' \rightarrow BC}})(\proj{\mathsf{EPR}}_{\reg{AA'}}).
    \end{equation*}
    In other words, $\rho_{\reg{ABC}}$ is the normalized Choi state of some channel $\Phi_{\reg{A' \rightarrow BC}}$.
    
    \item Alice's measurement $\big\{\vec{A}_x^{\theta}\big\}_{\theta \in \Theta, x \in \algo X}$ on register $\reg A$ is a projective measurement of the form
    $$
    \vec{A}_x^{\theta} = \bar{U}_\theta \proj{x} 
    \bar{U}_\theta^\dag \,,
    $$
    for some family of unitary operators $\{U_\theta\}_{\theta \in\Theta}$ acting on $\algo H_{\reg A}$.

    \item (If we are in the oracular setting) Bob and Charlie's measurements can only depend on oracle queries to $U_\theta$ and $U_\theta^\dag$, rather than directly on $\theta$.
    \end{itemize}
\noindent
We now prove a formal equivalence between the two notions for $t = 1$.

\begin{lemma}\label{lem:equiv}
Let $\mathsf{G}_{1 \mapsto 2} = (1,\algo H_{\reg{A}},\Theta, \algo X, \{U_{\theta}\}_{\theta \in \Theta})$ be a $1 \mapsto 2$ cloning game, for some family of unitary operators $\{U_\theta\}_{\theta \in\Theta}$ acting on the Hilbert space $\algo H_{\reg A}$. Then,
the winning probability of a particular strategy $\mathsf{S} = (\algo H_{\reg{B}} \otimes \algo H_{\reg{C}}, \Phi_{\reg{A \rightarrow BC}},\{\vec{P}_{1,x}^{\theta}\}_{\theta \in \Theta, x \in \algo X},\{\vec{P}_{2,x}^{\theta}\}_{\theta \in \Theta, x \in \algo X})$ (possibly in the oracular setting) with
$$
\omega_{\mathsf{S}}(\mathsf{G}_{1 \mapsto 2}) = 
\underset{\theta \sim \Theta}{\mathbb{E}}  \,\underset{x \sim \algo X}{\mathbb{E}}  \, \Tr{\left( \vec{P}_{1,x}^{\theta}\ot \vec{P}_{2,x}^{\theta}\right) \Phi_{\reg{A \rightarrow BC}}(U_\theta \proj{x}_{\reg{A}}U_\theta^\dag)}.
$$
is exactly equal to the winning probability 
$$
\omega_{\tilde{\mathsf{S}}}(\mathsf{G}) = 
\underset{\theta \sim \Theta}{\mathbb{E}}  \sum_{x \in \algo X}  \, \Tr{\left( \vec{P}_{1,x}^{\theta}\ot \vec{P}_{2,x}^{\theta}\otimes (\bar{U}_\theta \proj{x}_{\reg{A}} \bar{U}_\theta^\dag)\right) \rho_{\reg{BCA}}}.
$$
of a quantum strategy $\tilde{\mathsf{S}} = (\algo H_{\reg{B}},\algo H_{\reg{C}}, \rho_{\reg{BCA}},\{\vec{P}_{1,x}^{\theta}\}_{\theta \in \Theta, x \in \algo X},\{\vec{P}_{2,x}^{\theta}\}_{\theta \in \Theta, x \in \algo X})$ (possibly in the oracular setting) of a monogamy entanglement game $\mathsf{G} = (\algo H_{\reg{A}},\Theta, \algo X, \{\bar{U}_\theta \proj{x} 
    \bar{U}_\theta^\dag\}_{\theta \in \Theta, x \in \algo X})$, where
$$
 \rho_{\reg{BCA}} = (\Phi_{\reg{A' \rightarrow BC}} \otimes \id_{\reg A})\proj{\mathsf{EPR}}_{\reg{A'A}}.
$$
Here, $\ket{\mathsf{EPR}}_{\reg{AA'}}$ denotes $\frac{1}{\sqrt{|\mathcal{X}|}} \sum_{x \in \mathcal{X}} \ket{x}_{\reg A} \otimes \ket{x}_{\reg A'}$. (We swap the $\reg{A}$ and $\reg{BC}$ registers in the formulation of the relevant monogamy game for the purposes of syntactic compliance with Corollary~\ref{cor:CIprojector} in our proof.)
\end{lemma}
\begin{proof}
Recall that Corollary \ref{cor:CIprojector} implies that, for any projector $\vec{P} \in L(\algo H_{\reg{BC}})$,
$$\Tr{\vec{P}_{\reg{BC}} \Phi_{\reg{A' \rightarrow BC}}(U_\theta \proj{x}_{\reg{A}}U_\theta^\dag)} = \Tr{\left(\vec{P} \otimes {\bar{U}_\theta\proj{x}\bar{U}_\theta^\dag}\right)J(\Phi)},$$
where $J(\Phi) \in \linear(\algo H_{\reg B} \otimes \algo H_{\reg C} \otimes \algo H_{\reg A})$ is the Choi-Jamiołkowski isomorphism of $\Phi$. Therefore:

\begin{align*}
\omega_{\mathsf{S}}(\mathsf{G}_{1 \mapsto 2}) &= 
\underset{\theta \sim \Theta}{\mathbb{E}}  \,\underset{x \sim \algo X}{\mathbb{E}}  \, \Tr{\left( \vec{P}_{1,x}^{\theta}\ot \vec{P}_{2,x}^{\theta}\right) \Phi_{\reg{A \rightarrow BC}}(U_\theta \proj{x}_{\reg{A}}U_\theta^\dag)}\\
 &= 
\underset{\theta \sim \Theta}{\mathbb{E}} \, \underset{x \sim \algo X}{\mathbb{E}}   \, \Tr{\left( \vec{P}_{1,x}^{\theta}\ot \vec{P}_{2,x}^{\theta} \ot (\bar{U}_\theta \proj{x}_{\reg{A}} \bar{U}_\theta^\dag)\right) J(\Phi)}   \\
&=\underset{\theta \sim \Theta}{\mathbb{E}} \, \sum_{x \sim \algo X} \Tr{\left( \vec{P}_{1,x}^{\theta}\ot \vec{P}_{2,x}^{\theta} \ot (\bar{U}_\theta \proj{x}_{\reg{A}} \bar{U}_\theta^\dag)\right) \rho_{\reg{BCA}}} \, = \, \omega_{\tilde{\mathsf{S}}}(\mathsf{G}) \, ,
\end{align*}
where we define $\rho_{\reg{BCA}} = (\Phi_{\reg{A' \rightarrow BC}} \ot \id_{\reg A})\proj{\mathsf{EPR}}_{\reg{A'A}}$. The final step holds because of the identity $J(\Phi) = |\mathcal{X}| \cdot \rho_{\reg{BCA}}$ (see the definitions preceding Lemma~\ref{fact:CI}). This proves the claim.
\end{proof}
\noindent
We can generalize this to arbitrary $t$ as follows:

\begin{lemma}\label{lemma:CImultiplayer}
    For any $t \geq 1$, let $\mathsf{G}_{t \mapsto t+1} = (t,\algo H_{\reg{A^t}},\Theta, \algo X, \{U_{\theta}\}_{\theta \in \Theta})$ be a $t \mapsto t+1$ cloning game, for some family of unitary operators $\{U_\theta\}_{\theta \in\Theta}$ acting on the Hilbert space $\algo H_{\reg A}$. Moreover, define the shared state $$\rho_{\reg{B_{1:t+1}A'_{1:t}}} := (\Phi_{\reg{A_1\dots A_t \rightarrow B_1 \dots B_{t+1}}} \otimes \id_{\reg{A'_{1:t}}}) \left(\proj{\mathsf{EPR}^n}^{\otimes t}\right).$$
    Then,
    the winning probability of a particular strategy $\mathsf{S} = (\algo H_{\reg{B^{t+1}}}, \Phi_{\reg{A^t \rightarrow B^{t+1}}},\{\vec{P}_{1,x}^{\theta}\}_{\theta \in \Theta, x \in \algo X},\ldots,\{\vec{P}_{t+1,x}^{\theta}\}_{\theta \in \Theta, x \in \algo X})$ (possibly in the oracular setting) with
    $$
    \omega_{\mathsf{S}}(\mathsf{G}_{t \mapsto t+1}) := 
    \underset{\theta \sim \Theta}{\mathbb{E}}  \,\underset{x \sim \algo X}{\mathbb{E}}  \, \Tr{\left( \vec{P}_{1,x}^{U_\theta,U_{\theta}^\dag}\ot \dots \ot \vec{P}_{t+1,x}^{U_\theta,U_{\theta}^\dag}\right) \Phi_{\reg{A^t \rightarrow B^{t+1}}}\left( (U_\theta \proj{x}U_\theta^\dag)_{\reg{A^t}}^{\otimes t} \right)} .
    $$
    is exactly equal to the quantity
    $$
    |\mathcal{X}|^{t-1} \cdot \underset{\theta \sim \Theta}{\mathbb{E}}  \sum_{x \in \algo X}  \, \Tr{\left( \bigotimes_{i = 1}^{t+1} \vec{P}_{i, x}^{\theta} \otimes (\bar{U}_\theta \proj{x}_{\reg{A}} \bar{U}_\theta^\dag)^{\otimes t}\right) \rho_{\reg{B_{1:t+1}A'_{1:t}}}}.
    $$
\end{lemma}
\begin{proof}
    Let $J(\Phi) \in \linear(\algo H_{\reg{B_{1:t+1}}} \otimes \algo H_{\reg{A_{1:t}'}})$ denote the Choi-Jamiołkowski isomorphism of the cloning map $\Phi_{\reg{A_{1:t}} \rightarrow \reg{B_{1:t+1}}}$. Recall that $$J(\Phi) = |\mathcal{X}|^t \cdot (\Phi_{\reg{A_1\dots A_t \rightarrow B_1 \dots B_{t+1}}} \otimes \id_{\reg{A'_{1:t}}}) \left(\proj{\mathsf{EPR}^n}^{\otimes t}\right) = |\mathcal{X}|^t \rho.$$
    Now using Lemma~\ref{fact:CI}, we have:
    \begin{align*}
        \omega_\mathsf{S}(\mathsf{G}) = &\underset{\theta}{\mathbb{E}} \,\underset{x \sim \mathcal{X}}{\mathbb{E}}  \, \Tr{\left( \bigotimes_{i = 1}^{t+1} \vec{P}_{i, x}^\theta \right) \Phi_{\reg{A_1\dots A_t \rightarrow B_1 \dots B_{t+1}}}\left( (U_\theta \proj{x} U_\theta^\dag)_{\reg{A_1}\dots \reg{A_t}}^{\otimes t} \right)} \\
        = &\underset{\theta}{\mathbb{E}} \,\underset{x \sim \mathcal{X}}{\mathbb{E}}  \, \Tr{ \left( \bigotimes_{i = 1}^{t+1} \vec{P}_{i, x}^\theta \otimes \left(\bar{U}_\theta \proj{x}\bar{U}_\theta^\dag\right)_{\reg{A'_{1:t}}}^{\otimes t}\right) J(\Phi)_{\reg{B_{1:t+1} A_{1:t}'}} } \\
        = & |\mathcal{X}|^t \cdot \underset{\theta}{\mathbb{E}} \,\underset{x \sim \mathcal{X}}{\mathbb{E}}  \, \Tr{ \left( \bigotimes_{i = 1}^{t+1} \vec{P}_{i, x}^\theta \otimes \left(\bar{U}_\theta \proj{x}\bar{U}_\theta^\dag\right)_{\reg{A'_{1:t}}}^{\otimes t}\right) \rho_{\reg{B_{1:t+1} A_{1:t}'}} } \\
        = & |\mathcal{X}|^{t-1} \cdot \underset{\theta}{\mathbb{E}} \,\underset{x \sim \mathcal{X}}{\sum}  \, \Tr{ \left( \bigotimes_{i = 1}^{t+1} \vec{P}_{i, x}^\theta \otimes \left(\bar{U}_\theta \proj{x}\bar{U}_\theta^\dag\right)_{\reg{A'_{1:t}}}^{\otimes t}\right) \rho_{\reg{B_{1:t+1} A_{1:t}'}} }.
    \end{align*}
\end{proof}

\end{document}

%% file: intro-overview.tex
\section{Introduction}

\emph{Quantum no-cloning}~\cite{1982Natur.299..802W} is one of the most fundamental properties of quantum information. Roughly speaking, it states that no quantum procedure can create an exact copy of an arbitrary unknown quantum state. The principle of no-cloning has profound implications in quantum information processing~\cite{BB84,Bennett1993quantumteleportation,Roncaglia2019conservationofinformation,RevModPhys.89.015002} and has even inspired entirely 
new 
cryptographic primitives, starting with Wiesner's remarkable quantum money scheme~\cite{DBLP:journals/sigact/Wiesner83} and many subsequent primitives which are collectively known as
\emph{unclonable cryptography}~\cite{sattath2022uncloneablecryptography}; these include
unclonable quantum encryption~\cite{broadbent_et_al:LIPIcs.TQC.2020.4, 10.1007/978-3-031-38554-4_3, kundu2023deviceindependentuncloneableencryption, aky24}, encryption with unclonable decryption keys~\cite{cryptoeprint:2020/877,APV23}, quantum copy-protection~\cite{10.1007/978-3-031-15979-4_8, coladangelo2022quantum,10.1007/978-3-030-84242-0_20}, unclonable commitments and proofs~\cite{cryptoeprint:2023/1538}, and many more.

These cryptographic applications require one to prove unclonability guarantees that are much stronger than those implied by the no-cloning theorem, or even stronger variants stating that an unknown quantum state cannot be approximately copied to high fidelity~\cite{Bu_ek_1996, 10.1007/978-3-642-35656-8_4}.
Take the example of \emph{unclonable encryption}, where a classical message is encrypted into a quantum state, which an adversarial cloner $\Phi$ then operates on arbitrarily and forwards to two isolated adversaries $\mathcal{B}$ and $\mathcal{C}$.
Later, the decryption key is revealed and $\mathcal{B}$ and $\mathcal{C}$ attempt to recover the original message; they win if they both succeed.
The adversaries will certainly succeed if $\Phi$ can clone a ciphertext state, but $\Phi$ could also conceivably succeed by generating two completely different states that merely reveal enough information to later decrypt.
The no-cloning theorem and even its approximate variants can only rule out the first type of attack.

\vinod{excellent! I like this text very much, but cloning games are not directly relevant to quantum money, right? maybe give one more example above to connect it more directly to cloning games?} 
\seyoon{added, and I also commented out the quantum money example. what do you think?}


\paragraph{Cloning Games.}
{Following a long-standing tradition of studying quantum mechanical phenomena through the lens of interactive games~\cite{PhysRevLett.65.3373,Greenberger1989,reichardt2012classicalleashquantumsystem,Tomamichel_2013,10.1145/3485628}, the field of unclonable cryptography today relies on abstract \emph{cloning games}~\cite{10.1007/978-3-031-38554-4_3} as a means of capturing the desired strong unclonability guarantees.
}
These types of games first emerged in the context of unclonable encryption schemes~\cite{broadbent_et_al:LIPIcs.TQC.2020.4} but the general framework\footnote{We note that our notion of a basic cloning game is slightly more specific than the general framework studied in~\cite{10.1007/978-3-031-38554-4_3}; we discuss these differences in more detail in Remark~\ref{remark:aklcomparison}.} also applies to many other fundamental unclonable primitives, such as copy-protection, single-decryptor encryption, quantum money, and more~\cite{10.1007/978-3-031-38554-4_3}.

A basic $1 \mapsto 2$ cloning game $\mathsf{G}_{1 \mapsto 2}$ with respect to the question set $\Theta$, answer set $\mathcal{X}$, and ensemble of unitaries $\left\{U_\theta\right\}_{\theta \in \Theta}$ of dimension $|\mathcal{X}|$ is the following interactive game played by a trusted challenger, say Alice, as well as an adversary consisting of a cloner $\Phi$ and two additional players, say Bob and Charlie.
\begin{enumerate}
          \item (\textbf{Setup phase}) Alice samples random $x \sim \mathcal{X}$ and $\theta \sim \Theta$, and sends $U_\theta\ket{x}_{\reg{A}}$ to the cloner $\Phi$.

          The cloner $\Phi$ splits the state into two registers $\reg{B}$ and $\reg{C}$, which he then forwards to Bob and Charlie, respectively.
          Afterwards, the players may no longer communicate for the rest of the game.
        
          \item (\textbf{Question phase}) Bob and Charlie both receive the string $\theta$.
        \item (\textbf{Answer phase}) Bob and Charlie independently output a guess for the element $x$.
        
        \item (\textbf{Outcome phase}) Bob and Charlie win if they both guess $x$ correctly.
    \end{enumerate}


\begin{figure}[t]
\begin{center}
{\small
\begin{tikzpicture}
\draw (2,-0.5) rectangle (3.3,0.5) node [pos=.5]{$\algo A$}; 
  \draw (5,-0.5) rectangle (6.3,0.5) node [pos=.5]{${\Phi}$}; 
  \draw[->] (3.3,0) -- node[above]{{\small $U_\theta \ket{x}$}} (5,0);
  \draw[->,dashed] (5.65,-0.5) -- node[left]{} (5,-1.5);
  \draw[->,dashed] (5.65,-0.5) -- node[right]{} (6.3,-1.5);
  \draw (4.5,-1.5) rectangle (5.65,-2.5) node [pos=.5]{$\algo B$};
  \draw[->] (4,-2) node[left]{$\theta$} -- (4.5,-2);
  \draw[->] (5.05,-2.5) -- (5.05,-3) node[below]{$x_{\algo B}$};
  \draw[->] (7.65,-2) node[right]{$\theta$} -- (7.15,-2);
  \draw[->] (6.55,-2.5) -- (6.55,-3) node[below]{$x_{\algo C}$};
  \draw[->,dashed] (5.65,-0.5) -- (6.3,-1.5);
  \draw (6,-1.5) rectangle (7.15,-2.5) node [pos=.5]{${\algo C}$};

 \end{tikzpicture}
\caption{A basic $1\mapsto 2$ cloning game.} 
\label{fig:nocloning-intro}
}
\end{center}
\end{figure}
We illustrate the cloning game $\mathsf{G}_{1 \mapsto 2}$ in \Cref{fig:nocloning-intro}.
Formally, a strategy $\mathsf{S}$ for the game $\mathsf{G}_{1 \mapsto 2}$ consists of a cloning map $\Phi$ and positive operator-valued measurements $\algo B=\{\vec{B}_x^{\theta}\}_{\theta \in \Theta, x \in\algo X}$ and $\algo C = \{\vec{C}_x^{\theta}\}_{\theta \in \Theta, x \in \algo X}$.
The \emph{value} of a particular strategy $\mathsf{S}$ for the cloning game $\mathsf{G}_{1 \mapsto 2}$ is defined as the average winning probability
$$
\omega_{\mathsf{S}}(\mathsf{G}_{1 \mapsto 2}) = 
\underset{\theta \sim \Theta}{\mathbb{E}}  \,\underset{x \sim \algo X}{\mathbb{E}}  \, \Tr{\left( \vec{B}_{x}^{\theta}\ot \vec{C}_{x}^{\theta}\right) \Phi_{\reg{A \rightarrow BC}}(U_\theta \proj{x}_{\reg{A}}U_\theta^\dag)}.
$$
Here, $\omega(\mathsf{G}_{1 \mapsto 2})$ denotes the optimal winning probability over all strategies specified by $\Phi$, $\algo B$ and $\algo C$. Note that there exists a trivial strategy that succeeds with probability $1/|\mathcal{X}|$: the cloner $\Phi$ simply forwards $U_\theta \ket{x}$ to Bob, who can easily recover $x$ once $\theta$ becomes available, whereas Charlie simply guesses at random.


\begin{remark}[Cloning games capture much stronger forms of no-cloning]\label{remark:strongnocloning}
At first glance, it appears that establishing an upper bound on the winning probability of Bob and Charlie just boils down to the no-cloning theorem~\cite{1982Natur.299..802W} or perhaps its approximate variant~\cite{Bu_ek_1996}. Indeed, if the cloner $\Phi$ can copy the state $U_\theta \ket{x}$, then $\Phi$ can certainly also send the two copies to Bob and Charlie and ensure that they win the game. However, there could be other strategies which do not involve direct cloning but may nevertheless provide the players with enough information to win the game.\footnote{We provide an example of such a game and strategy at the beginning of~\Cref{sec:prevtechniques}.} In fact, the only property of $U_\theta \ket{x}$ that $\Phi$ needs to clone are the measurement statistics with respect to the unknown basis specified by $\theta$ (or more weakly, just $x$ itself). As it turns out, this stronger notion of unclonability is required for most applications in unclonable cryptography~\cite{sattath2022uncloneablecryptography}, and is significantly more challenging to prove.
\end{remark}

To this day, the majority of unclonable cryptography is rooted in either $n$-qubit BB84 states where $U_\theta= \mathsf{H^{\theta}}$~\cite{Tomamichel_2013,broadbent_et_al:LIPIcs.TQC.2020.4} or subspace coset states over $\mathbb{F}_2^n$, where $U_\theta$ encodes a shift of a random $n/2$-dimensional subspace $A \subset \mathbb{F}_2^n$~\cite{10.1007/978-3-030-84242-0_20,Culf_2022,schleppy2025winning}. In both cases, the optimal winning probability for the corresponding cloning game decays exponentially in the number of qubits~\cite{broadbent_et_al:LIPIcs.TQC.2020.4,Culf_2022,schleppy2025winning}.


\subsection{Our Contributions}\label{sec:ourcontributions}

Despite extensive study and multiple successful applications in unclonable cryptography, several important gaps in our understanding of cloning games remain. Our contributions to this effect are several fold:
\begin{enumerate}
    \item 
    We show that existing techniques for analyzing cloning games are severely limited; in particular, they prevent us from making progress on many fundamental open questions in the field.
    We formally expose these limitations with counterexamples and concrete, quantitative proofs.
    \item We study new cloning games and develop a suite of techniques for analyzing them; these techniques allow us to circumvent some of the limitations of previous approaches (in some cases, at the expense of restricting Bob's and Charlie's access to $\theta$ down to only a single query to $U_\theta$ or $U_\theta^\dag$).
    \item Finally, we present two applications of our results which have previously been out of reach; one in the area of \emph{black hole physics} and one in the field of \emph{unclonable cryptography}. Both of these applications provably require us to overcome several technical barriers which are inherent in prior work.
\end{enumerate}
We now discuss each of these contributions in more detail.



\paragraph{Exposing Limitations on Cloning Games.} 
Our first contribution is to expose several important gaps in our understanding of cloning games; more importantly, we also show that existing techniques for analyzing cloning games appear fundamentally insufficient at addressing them. We list some of these gaps below:

\begin{enumerate}
\item\label{item:optimalgames} \textbf{Optimal games:} Prior work on cloning games over $\algo X = \bit^n$ has shown the upper bounds of $\cos^2\left(\frac{\pi}{8}\right)^n$ and $2^{-n/4}$ in the case of BB84 states~\cite{Tomamichel_2013} and subspace coset states~\cite{Culf_2022, schleppy2025winning}, respectively. In contrast, a trivial strategy always succeeds with probability $2^{-n}$, and this holds for any cloning game. Are there \emph{especially hard} cloning games which admit no non-trivial strategies and have asymptotically optimal bounds of the form $O(2^{-n})$?
\itcsedit{Closing this gap is not merely an intellectual and aesthetic curiosity; it has important consequences for an application of cloning games to black hole physics which we introduce and study in our work.}
\itcsedit{We outline this application in \Cref{sec:applicationbh}.}

Not only are all known cloning games far from optimal, we prove that existing techniques can at best only produce upper bounds of the form $2^{-n/2}$. We discuss this \itcsedit{limitation} in detail in \Cref{sec:prevtechniques}.

\item\label{item:microcryptapp} \textbf{Unclonable encryption in MicroCrypt:} A number of recent works~\cite{DBLP:conf/tqc/Kretschmer21,10.1007/978-3-031-15802-5_8, DBLP:conf/tcc/AnanthGQY22,brakerski_et_al:LIPIcs.ITCS.2023.24} showed how to build quantum cryptography from \emph{pseudorandom states and unitaries}, which exist in ``MicroCrypt'' and are potentially even weaker than one-way functions~\cite{DBLP:conf/tqc/Kretschmer21}. To this day, however, the worlds of unclonable cryptography and MicroCrypt have been somewhat disconnected\footnote{This is with the notable exception of private-key quantum money, which is implied by pseudorandom states~\cite{cryptoeprint:2018/544}.}, as was recently observed in~\cite{DBLP:journals/corr/abs-2404-12647,ananth2024revocableencryptionprogramsmore}. Do pseudorandom unitaries, which have so far eluded major cryptographic application, give rise to interesting unclonable cryptography? 

The analysis of \emph{Haar cloning games}, where $U_\theta$ is a \emph{Haar} unitary (or, a unitary sampled from a unitary design), seems far beyond the scope of existing techniques, as we explain in Sections~\ref{sec:prevtechniques} and~\ref{sec:techoverviewhaar}.

\item\label{item:multicopygames} \textbf{Multi-copy games:} 
Can we extend $1 \mapsto 2$ cloning games to $t \mapsto t+1$ cloning games, where the cloner $\Phi$ receives $t$ many copies $(U_\theta \ket{x})^{\otimes t}$ and where $t+1$ players $\algo P_1, \ldots, \algo P_{t+1}$ simultaneously seek to recover $x$?
Although multi-copy variants of no-cloning have been studied before~\cite{werneroptimalcloning}, these are far from sufficient to understand multi-copy cloning games, as we previously noted in Remark~\ref{remark:strongnocloning}.
The natural notion of multi-copy games was raised as an open problem in~\cite{DBLP:journals/corr/abs-2404-12647,ananth2024revocableencryptionprogramsmore}, where the latter initiated the study of multi-copy security in the context of revocable cryptography.

Not only is prior work limited to $1 \mapsto 2$ cloning games, all existing unclonable cryptography is based on \emph{highly learnable} classes of states and becomes \textbf{completely insecure} if $t$ is allowed to grow polynomially in the number of qubits; this is in stark contrast with most quantum states which require exponentially many copies to be learned, as the literature on quantum tomography suggests~\cite{BCG13}.
We discuss the limitations of current approaches in \Cref{sec:prevtechniques} and provide strong evidence that existing techniques seem fundamentally insufficient for analyzing multi-copy games more generally.

We remark that prior works~\cite{ DBLP:conf/tcc/LiuLQZ22,10.1007/978-3-031-78020-2_8} studied a variant of multi-copy security in the context of quantum copy-protection and unclonable decryption, where each ``copy'' of the program or key is an i.i.d. \emph{mixed} state, and thus effectively an independent sample rather than an identical copy of a \emph{pure state}. 
While this setting allows one to generically reduce notions of multi-copy security to $1 \mapsto 2$ security (using a ``quantum pigeonhole argument''), it fails to capture the natural---and significantly more \emph{quantum}---notion of identical pure state copies; in particular, the same techniques do not carry over in this setting.
The pure state notion of unclonability is clearly more desirable in practice; not only does it allow one to send copies of the \emph{same} exact state to multiple recipients without compromising security, many identical copies also allow the recipients to approximately \emph{reflect} around the state~\cite{schoute2024quantumprogrammablereflections}, thereby offering an additional ``public verification feature'' which is unavailable for mixed states, and which has found many use-cases in unclonable cryptography~\cite{behera2024publicquantumcoins,ananth2024revocableencryptionprogramsmore}.

\item\label{item:beyondcryptoapp} \textbf{Applications beyond cryptography:}
While cloning games appear quite fundamental, their use case has so far been limited to cryptography.
Can cloning games offer new insights in other 
scenarios where no-cloning and monogamy of entanglement play an important role, such as in black hole physics? Recent works studied idealized models of black holes which rely on Haar random or pseudorandom unitary dynamics~\cite{Hayden_Preskill_2007,Kim_2023,engelhardt2024cryptographiccensorship}, which raises the question: can cloning games with Haar random unitaries help us understand how information gets scrambled inside of a black hole?

 An application to black hole physics once again seems to require new insights into \emph{Haar cloning games} which, as mentioned before, are currently out of reach. We refer to Sections~\ref{sec:prevtechniques} and~\ref{sec:techoverviewhaar}.

\end{enumerate}
\noindent
Given these inherent limitations on cloning games, it seems that fundamentally new techniques are needed in order to advance the field. This is where our next contribution comes in.

\paragraph{A New Suite of Techniques for Analyzing Cloning Games.} Our approach towards overcoming both limitations~\ref{item:optimalgames} and~\ref{item:multicopygames} is to focus on entirely new cloning games altogether.  Inspired by the recent literature on pseudorandom quantum states~\cite{cryptoeprint:2018/544, brakerski2019pseudo,cryptoeprint:2023/282}, we study a cloning game based on \emph{binary phase states}. The pseudorandomness of these states makes them excellent candidates for multi-copy unclonability~\cite{werneroptimalcloning}, in the sense of a traditional no-cloning theorem. In order to extend this to a stronger cloning game bound as discussed in Remark~\ref{remark:strongnocloning}, we take the existing formalism of binary types~\cite{DBLP:conf/tcc/AnanthGQY22} and extend it to a new notion of binary \emph{subtypes}, proving new standalone spectral bounds along the way. For technical reasons, our results only apply to a restricted model: rather than receiving the string $\theta$ in the clear, each player receives oracle access and is allowed to make a \itcsedit{single} query to either $U_\theta$ or $U_\theta^\dag$. While this constitutes a weaker model, it already implies something much stronger than a conventional $t \mapsto t+1$ no-cloning bound.\footnote{As we explain in Remarks~\ref{remark:strongnocloning} and~\ref{remark:veryrestricted}, approximate $t \mapsto t+1$ no-cloning emerges as a special case of our one-query cloning game, whereby each player makes a single query to $U_\theta^\dag$ and immediately measures in the computational basis (with no post-processing whatsoever). In this case, the value of the cloning game is precisely equal to the maximum \emph{average} cloning fidelity for $t \mapsto t+1$.}

Ultimately, we prove the following theorem (see \Cref{sec:techoverviewbinaryphase} for more details):

\begin{theorem}[Informal, see Theorem~\ref{thm:tcopycloning} for a formal statement]\label{thm:informaltcopycloning}
Let $n,t \in \N$. Then, the one-query $t \mapsto t+1$ binary phase cloning game $\mathsf{G}_{t \, \mapsto \, t+1}$ over $\mathcal{X} = \bit^n$, where each of the players is allowed to make one oracle query, has a value of $\omega(\mathsf{G}_{t \,\mapsto\, t+1}) \leq \exp(O(t \log t)) \cdot 2^{-n}$.
\end{theorem}
For constant $t$, this is asymptotically optimal and thus overcomes limitation~\ref{item:optimalgames}. For $t=o(n/\log n)$, this is still negligible in $n$ and thus makes significant progress towards overcoming limitation~\ref{item:multicopygames}. However, we believe that this construction is plausibly secure when $t$ is \emph{any} polynomial in $n$ (unlike previous constructions based on BB84~\cite{Tomamichel_2013, broadbent_et_al:LIPIcs.TQC.2020.4} and coset states~\cite{10.1007/978-3-030-84242-0_20, Culf_2022, schleppy2025winning}), and view our results as providing evidence towards this conjecture. Our justification for the plausible security of this construction is the fact that binary phase states are pseudorandom~\cite{cryptoeprint:2018/544, brakerski2019pseudo} and hence multi-copy unclonable~\cite{werneroptimalcloning}. We discuss our binary phase state construction more in~\Cref{sec:techoverviewbinaryphase}.

Secondly, we study 
the new and natural notion of a \emph{Haar cloning game}. Here, the unitary $U_\theta$ is sampled according to the \emph{Haar measure} and the players receive oracle access to $U_\theta$ and $ U_\theta^\dag$.
We show that the Haar cloning game is the \emph{hardest} cloning game by exhibiting a \emph{worst-case to average-case reduction}; this allows us to use an upper bound on the value of \emph{any} cloning game, including our binary phase state game, in order to bound the value of the Haar cloning game.  As a consequence, we additionally obtain the following:

\begin{corollary}[Informal]\label{cor:worsttoavinformal}
Let $n,t \in \N$.
As a consequence of our worst-case to average-case reduction (\Cref{thm:worst-case-to-average-case}), we can show the following bounds on the Haar cloning game:
    \begin{itemize}
        \item In the single-copy setting, the Haar game $\mathsf{G}_{1 \,\mapsto\, 2}$ has a value of $\omega(\mathsf{G}_{1\, \mapsto \,2})\leq \left(\cos^2(\pi/8)\right)^n \approx 2^{-0.228n}$.\\ (Here, the players are free to make \textbf{arbitrarily many adaptive queries} to $U_\theta$ or $U_\theta^\dag$.)

        \item In the multi-copy setting, the Haar game $\mathsf{G}_{t \,\mapsto\, t+1}$ has a value of $\omega(\mathsf{G}_{t\, \mapsto \,t+1})\leq \exp(O(t \log t)) \cdot 2^{-n}$. (Here, the players are restricted to making \textbf{only a single query} to $U_\theta$ or $U_\theta^\dag$.)
    \end{itemize}
\end{corollary}
\noindent
We will see next that Haar cloning games are central to the applications previously listed in items~\ref{item:microcryptapp} and~\ref{item:beyondcryptoapp}. We will discuss Haar cloning games and our worst-case to average-case reduction more in~\Cref{sec:techoverviewhaar}.


\paragraph{Opening Up New Applications.}
To demonstrate the full potential of our new insights into cloning games, we give two applications of our techniques which help resolve fundamental open questions in the field. We will see that these applications crucially require us to overcome the aforementioned limitations~\ref{item:optimalgames} and~\ref{item:multicopygames} in the existing constructions and analyses of cloning games.
\begin{itemize}
    \item \textbf{Black Hole Cloning Games.} In \Cref{sec:black-hole-games}, we analyze a new three-player game which is designed to capture cloning and entanglement monogamy in the context of evaporating black holes (see~\Cref{fig:black-hole-scrambling}). Our results offer new quantitative insights into the \emph{black hole information paradox}~\cite{PhysRevD.14.2460,Preskill:1992tc,Hayden_Preskill_2007} and suggest that, in an idealized model of a black hole which features Haar random (or pseudorandom) scrambling dynamics, the information from infalling qubits can only be recovered from either the interior or the exterior of the black hole, but never from both places at once---even in the presence of powerful observers which can make a single query to the scrambling unitary or its inverse.

    At a technical level, this requires us to essentially show a bound of $O(2^{-n})$ for the $1 \mapsto 2$ Haar cloning game; \textbf{even an exponentially small bound of $O(2^{-cn})$ for $c < 1$ will not suffice}. We thus crucially need to overcome the aforementioned limitation~\ref{item:optimalgames}, and we also need to make use of the aforementioned worst-case to average-case reduction. We will discuss this application in more detail in~\Cref{sec:applicationbh}, taking care to provide context on relevant prior work in black hole physics.

    \item \textbf{Unclonable Encryption: ``MicroCrypt'' and the Multi-Copy Setting.} In \Cref{sec:unclonableenc}, we give an affirmative answer to an open question which was recently posed in \cite{DBLP:journals/corr/abs-2404-12647}; namely: do interesting unclonable cryptographic primitives---other than private-key quantum money which is implied by pseudorandom states~\cite{cryptoeprint:2018/544}---exist, even in a world in which $\mathsf{P} = \mathsf{NP}$? We construct succinct unclonable encryption schemes from the existence of pseudorandom unitaries; thereby continuing to close the gap between the worlds of MicroCrypt and unclonable cryptography. A crucial ingredient for this result is the aforementioned worst-case to average-case reduction.

    Secondly, we propose a candidate multi-copy unclonable encryption scheme based on the aforementioned binary phase cloning game. We view~\Cref{thm:informaltcopycloning} as evidence and a first step towards proving its security in the stronger setting where $t$ can be an a priori unbounded polynomial in $n$ and the players are free to make polynomially many queries to the encryption and decryption functionality (or even more strongly, are given the secret key $\theta$ in the clear). Considering the pseudorandomness of the binary phase states that we use as ciphertexts, we believe this stronger security guarantee to be plausible. Obtaining multi-copy security clearly requires us to overcome limitation~\ref{item:multicopygames}.

    We will discuss these applications to unclonable cryptography in some more detail in~\Cref{sec:introue}.
\end{itemize}
We now turn to a detailed overview of our techniques.

\subsection{Technical Overview}\label{sec:techoverview}

This technical overview is organized as follows:
\begin{itemize}
    \item In~\Cref{sec:prevtechniques}, we discuss previous constructions~\cite{broadbent_et_al:LIPIcs.TQC.2020.4, 10.1007/978-3-030-84242-0_20} and explain the limitations of these constructions and the underlying techniques~\cite{Tomamichel_2013, Culf_2022, schleppy2025winning} in both the multi-copy and single-copy settings.

    \item In~\Cref{sec:techoverviewhaar}, we discuss a natural construction using Haar random unitaries, the obstacles to directly analyzing this construction, and our worst-case to average-case reduction as a means of indirectly studying this construction.

    \item In~\Cref{sec:techoverviewbinaryphase}, we present our construction based on binary phase states (previously put forth by~\cite{cryptoeprint:2018/544, brakerski2019pseudo} as a quantum pseudorandom state), and our new tools for analyzing this. Our tools significantly extend the previous notion of \emph{binary types} introduced by~\cite{DBLP:conf/tcc/AnanthGQY22} for analyzing binary phase states.
\end{itemize}
We remark at the outset that, when considering cloning games, we will consider a few different models for how the players (Bob and Charlie) access $\theta$:
\begin{itemize}
    \item \emph{Strong cloning games:} Bob and Charlie are given $\theta$ in the clear; 
    \item \emph{Oracular cloning games:} Bob and Charlie are given oracle access to $U_\theta$ and $U_\theta^\dag$, and are free to make an a priori unbounded polynomial number of queries.
    \item \emph{Bounded-query-oracular cloning games:} Bob and Charlie are given some a priori bound $q$ on the number of queries they can adaptively make.

    We note that, even when $q = 1$, this model is still quite expressive; as noted in~\Cref{remark:veryrestricted}, it captures conventional and approximate $t \mapsto t+1$ no-cloning bounds as a special case.
\end{itemize}
We will also sometimes consider Bob and Charlie who are required to run in quantum polynomial time.

\subsubsection{Previous Techniques and Their Limitations}\label{sec:prevtechniques}

Let us now dive a little deeper into the shortcomings of our current understanding of cloning games. Here, it is instructive to start with limitation \ref{item:multicopygames}; namely, that of multi-copy games.

\paragraph{Multi-Copy Insecurity of BB84 and Coset States.} One can extend cloning games as defined in~\Cref{fig:nocloning-intro} to \emph{multi-copy} cloning games: in this variant, the cloner $\Phi$ receives $t$ copies i.e. the state $(U_\theta \ket{x})^{\otimes t}$. In the place of Bob and Charlie, there are now $t+1$ players $\algo P_1, \ldots, \algo P_{t+1}$ who wish to all simultaneously guess $x$. It is easy to see that this game becomes provably easier for the adversaries as $t$ increases.

In the case of BB84 states, where $U_\theta \ket{x} = \mathsf{H}^\theta \ket{x}$, we claim that the adversaries can win with probability 1 for any $t \geq 2$. The attack is simple: $\Phi$ will measure one copy in the standard basis and one copy in the Hadamard basis, and forward these results to all players. Upon receiving $\theta$, the players can mix-and-match these results to recover $x$. The intuitive issue here is that the measurement basis is entirely disentangled across qubits; in fact,~\cite{DBLP:conf/tcc/AnanthK21} describes a generic attack on cloning games with this disentangled structure.

The case of coset states~\cite{10.1007/978-3-030-84242-0_20, Culf_2022, schleppy2025winning} is similar, albeit only after $t$ becomes larger than $n$. Here, we will interpret the basis $\theta$ as a subspace $A \subseteq \mathbb{F}_2^n$ of dimension $n/2$, and the input $x \in \bit^n$ as a pair of cosets $s+A, s'+A^\perp$. Then the cloner will receive copies of the state
$$\ket{A_{s, s'}} = \frac{1}{2^{n/4}}\sum_{a \in A} (-1)^{\langle a, s' \rangle} \ket{a+s}.$$
The cloner can measure half the copies in the standard basis and half the copies in the Hadamard basis, and forward these results to all players. Upon receiving $A$, they can all identify the cosets $s+A$ and $s' + A^\perp$ with high probability.

As stated, the above attacks only apply in the strong cloning setting. However, the situation is more grim for a very simple reason: \emph{these families of states are both learnable given $\mathsf{poly}(n)$ copies}, whereas most states require exponentially many copies to become learnable, as the state of the art in quantum tomography~\cite{BCG13} suggests. Hence, the cloner $\Phi$ can simply learn a classical description of the state and subsequently forward both the basis $\theta$ and the message $x$ to the $t+1$ players; the players do not even require the challenger to send them $\theta$. In other words, these attacks even hold in the bounded-query-oracular model with $q = 0$ or $q=1$. 
Nevertheless, our technical starting point, which we discuss next, is the toolkit used by previous works to analyze cloning games in the single copy ($t = 1$) case. We will pin down where exactly it fails to work in the multi-copy case and remedy these problems. To provide the necessary background, we begin by introducing the notion of a monogamy of entanglement game.

\paragraph{Monogamy of Entanglement Games.} We now take a brief detour and discuss the concept of \emph{monogamy of entanglement games}; we will see shortly that they are closely related to cloning games. Monogamy of entanglement games were introduced by Tomamichel, Fehr, Kaniewski and Wehner~\cite{Tomamichel_2013} in order to characterize entanglement monogamy~\cite{5388928} using the language of non-local games. Informally, quantum correlations are ``monogamous'', and thus cannot be shared freely among multiple parties.

A monogamy of entanglement (MOE) game $\mathsf{G}$ with respect to the question set $\Theta$, answer set $\algo X$ and measurement set $\{\vec{A}_x^{\theta}\}_{\theta \in \Theta, x \in \algo X}$ is an interactive game played by three players: a trusted referee called Alice, as well as two colluding and adversarial parties called Bob and Charlie.
\begin{enumerate}
  \item (\textbf{Setup phase}) Bob and Charlie prepare a tripartite quantum state $\rho \in \algo D(\algo H_{\reg{A}} \otimes \algo H_{\reg{B}} \otimes \algo H_{\reg{C}})$. They send register $\reg A$ to Alice, and hold onto registers $\reg B$ and $\reg C$, respectively. Afterwards, they are no longer allowed to communicate for the remainder of the game.

  \item (\textbf{Question phase}) Alice samples a random question $\theta \sim \Theta$, and then applies the corresponding measurement $\{\vec{A}_x^{\theta}\}_{x \in \algo X}$ to her register $\reg A$. Afterwards, Alice announces the question $\theta$ to both Bob and Charlie, and keeps the measurement outcome in $\algo X$ to herself.

\item (\textbf{Answer phase}) Bob and Charlie independently output a guess for Alice's outcome by applying the measurements $\{\vec{B}_x^{\theta}\}_{x \in\algo X}$ and $\{\vec{C}_x^{\theta}\}_{x \in \algo X}$ to their registers $\reg B$ and $\reg C$, respectively.

\item (\textbf{Outcome phase}) Bob and Charlie win if they both guess Alice's outcome correctly.
\end{enumerate}
Here, we associate a particular \emph{strategy} $\mathsf{S}$ employed by Bob and Charlie with the tuple consisting of the initial shared state $\rho$ and the positive operator-valued measurements $\{\vec{B}_x^{\theta}\}_{\theta \in \Theta, x \in\algo X}$ and $\{\vec{C}_x^{\theta}\}_{\theta \in \Theta, x \in \algo X}$.
The \emph{value} of a particular strategy $\mathsf{S}$ for the monogamy game $\mathsf{G}$ is defined as the average winning probability
$$
\omega_{\mathsf{S}}(\mathsf{G}) := 
\underset{\theta \sim \Theta}{\mathbb{E}} \sum_{x \in \algo X} \mathrm{Tr}\left[ (\vec{A}_x^{\theta} \otimes \vec{B}_x^{\theta} \otimes \vec{C}_x^{\theta}) \rho_{\reg{ABC}}  \right].
$$
We let $\omega(\mathsf{G})$ denote the maximal value of the game, i.e., the optimal winning probability over all strategies.
An upper bound on the value of a monogamy game therefore limits the extent to which Bob and Charlie can simultaneously maintain a quantum correlation with Alice who holds a register outside of their view. \emph{We emphasize that, in general monogamy of entanglement games, the shared state $\rho$ is completely arbitrary and adversarially chosen by Bob and Charlie; as we will see, this is the main way in which cloning games deviate from the monogamy of entanglement setting.}

\paragraph{Appendix~\ref{sec:moetocloning}: Connecting Cloning and Monogamy Games.} The standard tool for analyzing cloning games is to recast them as a special type of monogamy game. Together with the techniques laid out by~\cite{Tomamichel_2013} for analyzing monogamy games, this has found numerous applications in unclonable cryptography, including unclonable encryption~\cite{broadbent_et_al:LIPIcs.TQC.2020.4}, quantum copy-protection~\cite{Aaronson_2009,coladangelo2022quantum,10.1007/978-3-030-84242-0_20,10.1007/978-3-031-15979-4_8}, unclonable decryption keys~\cite{cryptoeprint:2020/877}, and unclonable proofs~\cite{cryptoeprint:2023/1538}.

We now explain this connection between cloning and MOE games. In a cloning game, Alice sends the cloner $\Phi$ the state $U_\theta \ket{x}$. Instead, we could imagine that Alice and $\Phi$ share several EPR pairs, and later on in the game (even after the cloning phase) Alice can apply a measurement $\left\{\matA_x^\theta := \bar{U}_\theta \proj{x} \bar{U}_\theta^\dag\right\}_{x \in \bit^n}$ on her side to induce the state $U_\theta \ket{x}$ on the cloner's side, where $\bar{U}$ denotes the complex conjugate of $U$. This yields a monogamy of entanglement game with the following two restrictions:
\begin{itemize}
    \item As already mentioned, Alice's measurements $\matA_x^\theta$ must take the form $\bar{U}_\theta \proj{x} \bar{U}_\theta^\dag$.
    \item The tripartite state $\rho$ shared by Alice, Bob, and Charlie is the \emph{Choi state} of the cloning channel $\Phi$. Concretely, we must have the special form \begin{equation}\label{eq:choidefinitiontechoverview}
        \rho_{\reg{ABC}} = (\id_{\reg A} \otimes \Phi_{\reg{A' \rightarrow BC}}) (\proj{\mathsf{EPR}}_{\reg{AA'}}).
    \end{equation}
    In words, \emph{$\rho$ can be adversarially chosen subject to the constraint that its marginal state on Alice's system is maximally mixed.}
\end{itemize}
\noindent
This equivalence, which we formally show in Lemma~\ref{lem:equiv}, was first observed in the context of BB84 states by Broadbent and Lord~\cite{broadbent_et_al:LIPIcs.TQC.2020.4}. On a high level, the statement is a consequence of the \emph{ricochet property} of EPR pairs, which we formally state in~\Cref{sec:quantumprelims}. The technical benefit of doing this is that it enables us to get a handle on the cloning channel $\Phi$ by absorbing it into the state shared by the players in the equivalent monogamy game. Now we can focus on Bob and Charlie's measurements which, as we will see, can be handled using spectral bounds as first observed by~\cite{Tomamichel_2013}.

Although the equivalence observed by~\cite{broadbent_et_al:LIPIcs.TQC.2020.4} is in the single-copy setting, it turns out that this idea readily generalizes to the multi-copy setting, as we show in Lemma~\ref{lemma:CImultiplayer}. The differences are as follows:
\begin{itemize}
    \item In the cloning game, Alice and the cloner $\Phi$ now share $nt$ EPR pairs, and Alice will measure each of the $t$ copies in the basis specified by $\left\{\bar{U}_\theta \proj{x} \bar{U}_\theta^\dag\right\}_{x \in \bit^n}$.

    \item In the equivalent monogamy-like\footnote{We say ``monogamy-like'' because monogamy-of-entanglement games traditionally involve just three parties~\cite{Tomamichel_2013}.} game, we only say that the adversaries win if Alice's $t$ measurements and the $t+1$ players' outputs are all equal to the same string $x$. In other words, we need to essentially post-select on Alice's measured string $x \in \bit^n$ being the same for each of the $t$ copies, which means the value of this monogamy-like game is immediately upper bounded by $2^{-n(t-1)}$.
\end{itemize}
Because of this post-selection, what we end up with is an equality of the following form:
\begin{equation}\label{eq:cloningtomonogamy}
    \omega(\mathsf{G}_{\mathrm{cloning}}) = 2^{n(t-1)} \cdot \omega(\mathsf{G}_{\mathrm{monogamy-like}}).
\end{equation}
Note that when $t = 1$, the $2^{n(t-1)}$ term is 1 and we recover the equivalence used by~\cite{broadbent_et_al:LIPIcs.TQC.2020.4}. Our goal is hence to upper bound the value of $\mathsf{G}_{\mathrm{monogamy-like}}$ by $2^{-n(t-1)} \cdot \negl(n)$, ideally even $O(2^{-nt})$.
\vinod{why monogamy-like?}
\seyoon{added a footnote addressing this; I don't think people considered this generalized monogamy game involving $> 3$ parties before.}


\paragraph{Section~\ref{sec:monogamyexisting}:~\cite{Tomamichel_2013} and its Limitations.} The work by~\cite{Tomamichel_2013} analyzes MOE games and later focuses on the BB84 case (i.e. $U_\theta = \mathsf{H}^\theta$) and uses two beautiful ideas, which for simplicity we state in the single-copy setting:

\begin{enumerate}
    \item\label{item:tfkwstate} The value of a particular monogamy game can be bounded \emph{independently} of the state $\rho_{\reg {ABC}}$ shared by the 3 players, noting that $\rho_{\reg {ABC}}$ is PSD and has trace 1. Concretely, one can show that
    \begin{align*}
         \underset{\theta \sim \Theta}{\EE} \sum_{x \in \mathcal{X}} \Tr{(\matA_x^\theta \otimes \mathbf{B}_x^\theta \otimes \mathbf{C}_x^\theta)\rho_{\reg{ABC}}} 
        & \, \leq \, \norm{\underset{\theta \sim \Theta}{\EE} \sum_{x \in \mathcal{X}} \matA_x^\theta \otimes \mathbf{B}_x^\theta \otimes \mathbf{C}_x^\theta}_\infty.
    \end{align*}

    This reduces the task of bounding the value of a monogamy game to bounding an operator norm. In general monogamy games as formulated by~\cite{Tomamichel_2013}, the shared state $\rho_{\reg ABC}$ is adversarially chosen so this step is tight. However, we will see soon that this step is too lossy when restricting attention to the special monogamy-like games that are equivalent to cloning games.

    \item\label{item:tfkwoverlap} This operator norm can in turn be bounded just in terms of \emph{pairwise overlaps} between the $\matA_x^\theta$'s, which the designer of the game is free to choose. As we restate in~\Cref{thm:tfkwmain}, the authors of~\cite{Tomamichel_2013} show that $$\norm{\underset{\theta \sim \Theta}{\EE} \sum_{x \in \mathcal{X}} \matA_x^\theta \otimes \mathbf{B}_x^\theta \otimes \mathbf{C}_x^\theta}_\infty \leq \frac{1}{|\Theta|} + \frac{|\Theta|-1}{|\Theta|} \cdot \underset{\substack{\theta, \theta' \in \Theta\\\theta \neq \theta'}}{\max}\text{ } \underset{x, x' \in \mathcal{X}}{\max} \norm{\matA_x^{\theta} \matA_{x'}^{\theta'}}_\infty.$$
    We refer the reader to~\Cref{thm:tfkwmain} for a formal statement.

\end{enumerate}
\noindent
In the BB84 monogamy game where $\Theta = \mathcal{X} = \left\{0, 1\right\}$ and $\matA_x^\theta = \mathsf{H}^\theta \proj{x} \mathsf{H}^\theta$, it is straightforward to see that $\norm{\matA_x^{\theta} \matA_{x'}^{\theta'}}_\infty = \frac{1}{\sqrt{2}}$, and hence $\omega(\mathsf{G}_{\text{\tiny{BB84}}}) \leq \frac{1}{2} + \frac{1}{2\sqrt{2}}$. The work by~\cite{Tomamichel_2013} also extends this to ``parallel-repeated'' BB84 games with $|\Theta| = |\mathcal{X}| = \bit^n$ (see~\Cref{def:parallel-rep} for a formal definition), and show that $$\omega(\mathsf{G}_{\text{\tiny{BB84}}}^{\otimes n}) \leq \cos^2\left(\frac{\pi}{8}\right)^n \approx 2^{-0.228n}.$$
Hence the BB84 monogamy game has value $\leq 2^{-0.228n}$, and in fact this is tight;~\cite{Tomamichel_2013} exhibits a simple strategy achieving this bound. Similar techniques were used by~\cite{Culf_2022} and improved upon by~\cite{schleppy2025winning} to analyze subspace coset states, ultimately proving an upper bound of $O(2^{-n/4})$; it is not known whether or not this is tight.\footnote{Previous work~\cite{schleppy2024optimal} proved an upper bound of $O(2^{-n/2})$ in a setting where the cloner $\Phi$ is restricted to splitting the state as is into two equal-sized halves, sending one to Bob and the other to Charlie. We cite $O(2^{-n/4})$ as the state of the art, as we are interested in games where the cloner $\Phi$ is unrestricted.} However, the techniques laid out by~\cite{Tomamichel_2013} \emph{provably} do not suffice for our applications:
\begin{itemize}
    \item In the multi-copy case, recall from Equation~\eqref{eq:cloningtomonogamy} that we need to prove a bound on $\mathsf{G}_{\mathrm{monogamy-like}}$ of $\ll 2^{-n(t-1)}$. Item~\ref{item:tfkwstate} of the~\cite{Tomamichel_2013} methodology proposes to ignore the structure of the state shared by Alice and $\algo P_1, \ldots, \algo P_{t+1}$, in order to reduce our task to bounding an operator norm.
    
    This is likely too lossy in our setting, as evidenced by the following simple counterexample (assume $t > 1$ is even for simplicity) that holds against any cloning game where the unitaries $U_\theta$ have real entries. 
    Alice will hold $tn/2$ EPR pairs (which are unentangled from the states held by $\algo P_1, \ldots, \algo P_{t+1}$). The $t+1$ players will each deterministically output $0^n$ as their guess (note that once again this strategy does not depend at all on $\theta$). The winning probability is now just the probability that Alice measures $0$ on each of her $tn/2$ EPR pairs (in whatever basis she samples), which is $2^{-nt/2} \geq 2^{-n(t-1)}$. We formalize this counterexample in~\Cref{sec:opnormcounterexample}.

    We note that this does not rule out the possibility of the~\cite{Tomamichel_2013} technique being adaptable to the multi-copy setting for a construction that uses unitaries with complex entries. However, the pre-existing constructions based on BB84 states or coset states --- as well as our main construction based on binary phase states (which we sketch in~\Cref{sec:techoverviewbinaryphase}) --- only use unitaries with real entries. Thus we still view our result as a bound against the adaptability of existing construction and techniques to the multi-copy setting.

    \item Even in the single-copy case, there is another inherent limitation that arises from using Item~\ref{item:tfkwoverlap} of the ~\cite{Tomamichel_2013} methodology: the maximal pairwise overlap $\max_{\theta \neq \theta'} \norm{\matA^\theta_x \matA^{\theta'}_{x'}}_\infty$ is provably at least $2^{-n/2}$ for any monogamy game, as we show in~\Cref{sec:tfkwendoftheline}. For completeness, we also show in~\Cref{sec:salting} that our binary phase state construction (which we will discuss more in~\Cref{sec:techoverviewbinaryphase}) essentially attains this maximum pairwise overlap.
    
    Howeover, we would ideally like to prove a tight bound (up to constant factors) of $O(2^{-n})$. This is not just a matter of aesthetic taste; this is actually crucial for our application to black hole physics, as we explain in~\Cref{sec:applicationbh}.
\end{itemize}
We next turn our attention to our construction and our new techniques for analyzing it, which make progress towards overcoming both of these barriers.

\subsubsection{Section~\ref{sec:worsttoav}: Haar Cloning Games and Worst-Case to Average-Case Reductions}\label{sec:techoverviewhaar}

Given our previous observations on the multi-copy insecurity of BB84 and coset states, it is clear that we need to look for entirely new constructions. Ideally, such a candidate ensemble of states would also remain \emph{unlearnable} in the presence of an arbitrary polynomial amount of identical copies. A natural idea is to consider Haar random states, which Werner~\cite{werneroptimalcloning} showed to be multi-copy unclonable. This suggests the following very natural approach: Alice will take $\left\{U_\theta: \theta \in \Theta\right\}$\footnote{The Haar random ensemble is infinite so this is not well-defined; this technicality can be circumvented by using a higher order unitary design or a pseudorandom unitary~\cite{DBLP:journals/corr/abs-2404-12647, hm24} in its place.} to be a Haar random ensemble and send the cloner $(U_\theta\ket{x})^{\otimes t}$. We call this the $t \mapsto t+1$ \emph{Haar cloning game}. However, existing techniques for analyzing the Haar measure, which we outline below, appear severely limited for our purposes:
\begin{itemize}
    \item Prior works~\cite{DBLP:journals/corr/abs-2404-12647,Chen_2024,allerstorfer2024monogamyhighlysymmetricstates} often rely on representation-theoretic techniques. In our setting, we would roughly need to prove spectral bounds on the \emph{mixed Haar twirl} of a certain operator $\Xi$; informally, in the $1\mapsto 2$ case, these amount to expressions of the form:
    $$\norm{\underset{U \sim \mathrm{U}(d)}{\EE} \left[\left(U \otimes U \otimes \bar{U}\right) \Xi \left(U \otimes U \otimes \bar{U}\right)^\dag\right]}_\infty.$$
    General expressions of this form have been studied by~\cite{PhysRevA.63.042111,grinko2023linearprogrammingunitaryequivariantconstraints, grinko2023gelfandtsetlinbasispartiallytransposed} using the machinery of \emph{mixed Schur-Weyl duality}, but their techniques appear to be very unwieldy in our more complicated setting with multiple non-communicating parties.
    \item The recent breakthrough result by Ma and Huang~\cite{hm24} uses a technically involved purification argument~\cite{hm24} that once again does not seem to adapt easily to the multi-party setting.
    \item Finally, the recent beautiful work by Bhattacharya and Culf~\cite{bhattacharyya2025uncloneableencryptiondecoupling} analyzes the Haar measure in the single-copy case using a modular application of the one-shot decoupling theorem~\cite{Dupuis2014}, but in the process only establishes a cloning bound of $\tilde{O}(1/n)$, whereas we would like a bound that is $O(2^{-n})$ or at the very least exponentially small in $n$.
\end{itemize}
Instead, we take a two-step approach which is based on the following insight: cloning games instantiated with a Haar (pseudo)random unitary are, in some sense, \emph{strictly harder to win} than any other cloning game. We prove this via a \emph{worst-case to average-case reduction} which, at a high level, follows from Haar invariance and some additional new insights into \emph{mixed} unitary designs, which we explain in more detail below. 

Our observation immediately suggests the following approach for analyzing a Haar cloning game:
\begin{enumerate}
    \item\label{item:wsttoavg} Argue that for \emph{any} distribution $\mathfrak{D}$ supported on $\mathrm{U}(2^n)$, we have: \begin{equation}\label{eq:worsttoavoverview}
        \sup_{\text{strategies }\mathsf{S}} \omega_{\mathsf{S}}(\mathsf{G}; U \sim \mathrm{U}(2^n)) \leq \sup_{\text{strategies }\mathsf{S}} \omega_{\mathsf{S}}(\mathsf{G}; U \sim \mathfrak{D}).
    \end{equation}

    \item\label{item:findawstcase} Find a convenient distribution $\mathcal{D}$ such that we can more easily show that $$\sup_{\text{strategies }\mathsf{S}} \omega_{\mathsf{S}}(\mathsf{G}; U \sim \mathfrak{D}) \leq O(2^{-n}),$$perhaps by passing first to an equivalent monogamy game as stated earlier.
\end{enumerate}
\noindent
We prove the aforementioned \emph{worst-case-to-average-case reduction} which is captured in Item~\ref{item:wsttoavg} in~\Cref{sec:worsttoav}.
To instantiate this argument, we need to be able to sample $V$ that appears Haar random, together with a classical description of it. This can be done using either a \emph{mixed unitary design} (in the bounded-query-oracular setting) or a pseudorandom unitary~\cite{DBLP:journals/corr/abs-2404-12647, hm24} (in the oracular setting with computationally bounded players). We formally define mixed unitary designs in~\Cref{sec:mixedunitarydesigns}, and---as a bonus---we also prove that the standard notion of an exact unitary $t$-design will also work as a mixed unitary design without modification. To the best of our knowledge, this was not previously observed in the literature, and we hope that this contribution might be of independent interest. We leave the task of adapting this reduction to the strong cloning game setting---that is, where a description of the unitary is revealed to the players in the clear rather than embedded in an oracle---as a direction for future work.
\vinod{remind the reader what are strong cloning games?}\seyoon{good point, thanks! done}



It now remains to address Item~\ref{item:findawstcase} i.e. find some other cloning game that we can more easily show an upper bound of $O(2^{-n})$ for. We address this next.

\subsubsection{Sections~\ref{sec:typesandsubtypes} and~\ref{sec:binaryphaseconstruction}: Construction and Analysis from Binary Phase States}\label{sec:techoverviewbinaryphase}

\paragraph{Our Construction Inspired by Quantum Pseudorandom States.} We begin from a simple starting point: in~\Cref{sec:techoverviewhaar}, we suggested having Alice send a Haar random state $U_\theta \ket{x}$ to the cloner. Instead, what if Alice were to send a \emph{pseudorandom} state~\cite{cryptoeprint:2018/544, brakerski2019pseudo}? These are also multi-copy unclonable by a trivial hybrid argument combined with Werner's result~\cite{werneroptimalcloning} in the Haar case. The advantage of working with binary phase states is that we have much simpler constructions that, as we will see, are easier to analyze. 

It was shown by~\cite{cryptoeprint:2018/544, brakerski2019pseudo} that if $\mathfrak{F}$ is a family of post-quantum pseudorandom functions~\cite{10.1145/3450745} $f: \left\{0, 1\right\}^n \rightarrow \left\{0, 1\right\}$, then the \emph{binary phase state} $$\ket{\psi^f} := 2^{-n/2} \sum_{y \in \bit^n} (-1)^{f(y)} \ket{y}$$ is pseudorandom. To use this in a cloning game, we follow the approach in~\cite{cryptoeprint:2023/282} in order to encode $x \in \bit^n$ into this state, which we do by taking:
\begin{align*}
    \ket{\psi_x^f} &:=  2^{-n/2}  \sum_{y \in \bit^n} (-1)^{f(y) + \langle x,y\rangle} \ket{y}
    =  \mathsf{U}_f \mathsf{H}^{\otimes n} \ket{x}, \quad \text{ where} \quad
    \mathsf{U}_f := \sum_{x \in \bit^n} (-1)^{f(x)} \proj{x}
\end{align*}
is a phase oracle for $f$. In other words, we are proposing to define a cloning game with $\Theta = \mathfrak{F}$ and $U_\theta = \mathsf{U}_f \mathsf{H}^{\otimes n}$. Thus in the equivalent monogamy game, Alice's projectors will be defined by 
$$\matA_x^f := \mathsf{U}_f \mathsf{H}^{\otimes n} \proj{x} \mathsf{H}^{\otimes n} \mathsf{U}_f.$$
The question is now how one should go about analyzing this game. As mentioned in~\Cref{sec:prevtechniques}, the usual~\cite{Tomamichel_2013} methodology for analyzing cloning games is firstly limited to the single-copy setting, and secondly even in this setting can only prove a bound of $2^{-n/2}$. For completeness, we show in~\Cref{sec:salting} that plugging our binary phase construction into~\cite{Tomamichel_2013} ``saturates'' this technique and yields a single-copy cloning bound of $\widetilde{O}(2^{-n/2})$, which is stronger than the previous results on BB84~\cite{Tomamichel_2013, broadbent_et_al:LIPIcs.TQC.2020.4} and coset~\cite{10.1007/978-3-030-84242-0_20, Culf_2022, schleppy2025winning} states.

\paragraph{Compressed Oracles and Binary Types.} The techniques discussed up to this point draw on the machinery of~\cite{Tomamichel_2013} and thus suffice to establish cloning bounds even in the setting where a classical description of the measurement basis $\theta$ (in the binary phase case, the function $f$) is sent to all players $\algo P_1, \ldots, \algo P_{t+1}$ in the clear. To our knowledge, the only other technique that works in this strong regime is the decoupling technique by~\cite{bhattacharyya2025uncloneableencryptiondecoupling}. However, as we explained in Sections~\ref{sec:prevtechniques} and~\ref{sec:techoverviewhaar}, both of these techniques run into limitations with respect to the multi-copy setting and/or attaining an optimal bound of $O(2^{-n})$.

We hence propose to migrate to the oracular setting, where each player can make oracle queries to $\mathsf{U}_f$, but is not given a description of $f$ in the clear. Note that this still suffices to recover $x$ from $\mathsf{U}_f\mathsf{H}^{\otimes n} \ket{x}$. In fact, we only need one query: a single query to $\mathsf{U}_f$ would leave us with $\mathsf{H}^{\otimes n} \ket{x}$, and now measuring in the Hadamard basis yields $x$. We explain this in the case of general cloning games in~\Cref{remark:veryrestricted}.

We now want to reason about algorithms that make oracle queries to $\mathsf{U}_f$ for a random (or pseudorandom) function $f$. The natural candidate technique for such a task is Zhandry's compressed oracle technique~\cite{DBLP:conf/crypto/Zhandry19}. The crucial idea is to purify the cloning game by adding a register that stores the function $f$. One can then argue that queries to $\mathsf{U}_f$ for a random $f: \bit^n \rightarrow \bit$ can be simulated as follows:
\begin{itemize}
    \item We will add a purifying ``database'' register to the system that we initialize to $\ket{\emptyset}$. In general, it will store some subset $S \subseteq \bit^n$.
    \item If the algorithm wishes to query the string $y \in \bit^n$, simply update $\ket{S} \gets \ket{S \oplus \left\{y\right\}}$ (note that if $y$ is in $S$ before the query, this will remove $y$ from $S$.)
\end{itemize}
Let us see how this technique would play out in our setting. Alice and $\algo P_1, \ldots, \algo P_{t+1}$ all share some global state, and act as follows:
\begin{itemize}
    \item Alice queries each of her $t$ $n$-qubit states, then applies a Hadamard to each copy. Her local transformation can be written as $\bigotimes_{i = 1}^t \left(\sum_{y_i \in \bit^n} \mathsf{H}^{\otimes n} \proj{y_i}\right)$, and she will XOR $\left\{y_1\right\} \oplus \ldots \oplus \left\{y_t\right\}$ into the database register.

    \item For each $i \in [t+1]$, let the adversary $\algo P_i$ make $q$ adaptive queries represented by unitaries $V_{i, 1}, \ldots, V_{i, q}$. Their local transformation can be written as a sum of terms of the form $$V_{i, q} \proj{z_{i, q}} V_{i, q-1} \proj{z_{i, q-1}} \ldots V_{i, 1} \proj{z_{i, 1}},$$ and they will XOR $\left\{z_{i, 1}\right\} \oplus \ldots \oplus \left\{z_{i, q}\right\}$ into the database register.
\end{itemize}
In summary, the database register will contain the set:
\begin{equation}\label{eq:typexor}
    \bigoplus_{i = 1}^t \left\{y_i\right\} \oplus \bigoplus_{i = 1}^{t+1} \bigoplus_{j = 1}^q \left\{z_{i, j}\right\}.
\end{equation}
At this point, it is unclear how to proceed, for the following simple conceptual reason: the utility of Zhandry's compressed oracle technique~\cite{DBLP:conf/crypto/Zhandry19} lies in the fact that it connects an algorithm's success probability with the contents of the database register in some way. For example, when reproving the~\cite{bbbv} lower bound showing the optimality of Grover search, Zhandry shows that the success probability of the algorithm is essentially upper bounded by the probability of a solution $x$ to the search problem appearing in the database register. However, there is no analogous notion in our setting, because we are considering a problem with \emph{inherently quantum inputs}; a successful adversary likely needs to query $\mathsf{U}_f$ on every input in superposition.

Instead, we will deviate from Zhandry's compressed oracle formalism by simply tracing out (or equivalently, measuring) the database register. This effectively conditions our superposition on the collection of strings that are listed an odd number of times in Equation~\eqref{eq:typexor}. This is exactly the notion of \emph{binary types} introduced by~\cite{DBLP:conf/tcc/AnanthGQY22}. In order to get a better handle on the binary type's combinatorial structure, we will restrict each of the $t+1$ players to only make \emph{one oracle query} to $\mathsf{U}_f$; as explained in~\Cref{remark:veryrestricted}, this is still sufficiently expressive to admit a trivial strategy attaining value $2^{-n}$.

In this case, a binary type $\blambda$ is specified by a subset $T_{\blambda} \subseteq [2^n]$ with $|T_{\blambda}| \leq 2t+1$. For $\vecx \in [2^n]^{2t+1}$, we say that $\BinType(\vecx) = \blambda$, or equivalently that $\vecx$ \emph{matches} $\blambda$, if every string in $T_{\blambda}$ appears an odd number of times in $\vecx$, while every string outside $T_{\blambda}$ appears an even number of times in $\vecx$. Thus, if Alice and the players jointly hold a standard basis state $\ket{\vecx}$, a simultaneous query to $\mathsf{U}_f$ by all parties will write $\BinType(\vecx)$ into the database register. Finally, we let $\Pi_{\blambda}$ denote the projector onto standard basis vectors $\vecx$ that match $\blambda$. We provide more precise definitions and properties of binary types in~\Cref{sec:phaseunitarydefinition}. We now model each player's projector as follows:
\begin{align*}
    \mathbf{P}_{i, x}^f &= \mathsf{U}_f V_i^\dagger \proj{x} V_i \mathsf{U}_f,
\end{align*}
for unitaries $Q_i$.\footnote{In reality, we later also allow these players additional ancillary workspace qubits; we define this generalization in Definition~\ref{def:veryrestrictedocg}. Moreover, we assume without loss of generality that the players do not perform any preprocessing before making their query to $\mathsf{U}_f$, by absorbing this preprocessing into the cloning channel $\Phi$ that constructs their initial states.}
This simplification together with the aforementioned binary type formalism allows us to succinctly characterize the value of the cloning game: if we define
\begin{align*}
    \Xi &:= \sum_{x \in \bit^n} \left[\left(\mathsf{H}^{\otimes n} \proj{x} \mathsf{H}^{\otimes n}\right)^{\otimes t} \otimes \bigotimes_{i = 1}^{t+1} \left(V_i^\dag \proj{x} V_i\right)\right],\text{ then} \\
    \omega(\mathsf{G}) &= \sum_{\blambda} \Tr{\Pi_{\blambda} \Xi \Pi_{\blambda} \rho},
\end{align*}
where $\rho$ is the shared state from the monogamy-like game that we introduced in~\Cref{sec:prevtechniques}. We face two challenges in bounding expressions of this form. We state them below and then describe how we address these challenges:
\begin{enumerate}
    \item\label{item:opnormnotenough} The~\cite{Tomamichel_2013} paradigm of discarding the tripartite state $\rho$ and simply bounding this expression by $\max_{\blambda} \norm{\Pi_{\blambda} \Xi \Pi_{\blambda}}_\infty$\footnote{This bound holds by noting that the projectors $\Pi_{\blambda} \Xi \Pi_{\blambda}$ are mutually orthogonal.} is provably too lossy, as we explained in~\Cref{sec:prevtechniques}.

    \item\label{item:typesareentangling} Even if it were somehow sufficient to bound $\norm{\Pi_{\blambda} \Xi \Pi_{\blambda}}_\infty$ for each $\blambda$, to the best of our knowledge, it appears difficult to directly establish such a bound. Informally, the reason is that the combinatorial structure arising from a type $\blambda$ entangles registers together; if we consider the $t = 1$ case and the type defined by $T_{\blambda} = \left\{x^*\right\}$ for some string $x^*$, then strings of the form $(x^*, y, y), (y, x^*, y),$ or $(x^*, y, y)$ would all match $\blambda$. It would be much cleaner if we could just analyze strings from one of these categories at a time.
\end{enumerate}

\paragraph{Idea 1 (Section~\ref{sec:freevariablesymbolcombi}): Staring at the Shared State.} To address Item~\ref{item:opnormnotenough}, we take a closer look at the structure of the shared state $\rho$. It is the result of applying some (adversarially chosen) channel to the right half of $tn$ $\mathsf{EPR}$ pairs. This can be seen from Equation~\eqref{eq:choidefinitiontechoverview} (appropriately generalized to the multi-copy setting). In other words, if we apply a partial trace to remove the $\reg{P_{1 \rightarrow t+1}}$ registers of the $t+1$ players, the residual state on Alice's register $\reg{A}$ will always be proportional to $\id_{2^n \times 2^n}$. 

This structure may seem mild, but it turns out to be enough to complete our analysis; we present this in Sections~\ref{sec:freevariablesymbolcombi} and~\ref{sec:noancillas}. This is perhaps not surprising; in the counterexample we presented in~\Cref{sec:prevtechniques} showing that just bounding the operator norm would be insufficient, Alice's local state was very far from maximally mixed. In fact, it was a pure state consisting of $t/2$ EPR qudit pairs.

At a high level, our analysis to use this structure of $\rho$ proceeds by showing that for any type $\blambda$ such that $\Pi_{\blambda} \Xi \Pi_{\blambda}$ has high operator norm, the shared state $\rho$ must place \emph{low} weight on the image of $\Pi_{\lambda}$. These effects roughly cancel each other out.

\paragraph{Idea 2 (Sections~\ref{sec:subtypes} and~\ref{sec:opnormbound}): From Types to Subtypes.} We now turn to the issue stated in Item~\ref{item:typesareentangling}. Continuing with the $t = 1$ example, reasoning about $\blambda$ directly requires simultaneously handling three categories of strings: $(x^*, y, y), (y, x^*, y),$ or $(x^*, y, y)$. Instead, we simplify matters by focusing on just one of these categories at a time --- we call such a category a \emph{subtype}, a novel notion we define formally in~\Cref{sec:subtypes}. We denote subtypes by $\bmu$ and their corresponding subtype projectors by $\Pi_{\bmu}$. In~\Cref{sec:subtypetotype}, we show that instead of bounding $\norm{\Pi_{\blambda} \Xi \Pi_{\blambda}}_\infty$ for a type $\blambda$, it suffices to bound $\norm{\Pi_{\bmu} \Xi \Pi_{\bmu}}_\infty$ for a \emph{subtype} $\bmu$. This added structure allows us to prove better spectral bounds, which we present in~\Cref{sec:opnormbound}.

It turns out that this technique allows us to prove the desired bound of $O(2^{-n})$ for $1 \mapsto 2$ cloning games, albeit with the restriction that Bob and Charlie can only make one query each to $\mathsf{U}_f$. At a very high level, the ``product structure'' of subtypes enables us to leverage a simple but novel spectral bound on the column-wise tensor product of several matrices, which we present in Lemma~\ref{lemma:colwisetensor}.

\paragraph{Technical Tool (Section~\ref{sec:linalg}): Spectral Bounds on Blockwise Tensor Products.} In order to prove spectral bounds on the norm of $\Pi_{\bmu} \Xi \Pi_{\bmu}$ for any subtype $\bmu$, and accommodate the possibility of the $t+1$ players using ancilla qubits, we require a novel bound on the norm of a blockwise tensor product of $d \times d$ block matrices. As a simple example, the $d = 2$ case is the following: we need to show that 
$$\norm{\begin{bmatrix} c_{1, 1} \vec A_{1, 1} \otimes \vec B_{1, 1} & c_{1, 2} \vec A_{1, 2} \otimes \vec B_{1, 2} \\ c_{2, 1} \vec A_{2, 1} \otimes \vec B_{2, 1} & c_{2, 2} \vec A_{2, 2} \otimes \vec B_{2, 2}\end{bmatrix}}_\infty \leq 1,$$provided that $\begin{bmatrix} \vec A_{1, 1} & \vec A_{1, 2} \\ \vec A_{2, 1} & \vec A_{2, 2}\end{bmatrix}, \begin{bmatrix} \vec B_{1, 1} & \vec B_{1, 2} \\ \vec B_{2, 1} & \vec B_{2, 2}\end{bmatrix}$ are unitary and $|c_{i, j}| \leq 1$ for all $i, j$. Proving this turns out to be rather technically challenging; we present this result and its proof in Theorem~\ref{thm:opnormmain}. We also discuss at the end of~\Cref{sec:blockwisetensor} why existing techniques fail to prove the general result we need. Given that this theorem is a purely linear algebraic statement unrelated to monogamy games, we are hopeful that it might be useful elsewhere in quantum information and even in other areas.


Putting these ideas together, we manage to prove a multi-copy cloning bound of $O_t(2^{-n})$, overcoming the limitations of previous techniques explained in~\ref{sec:prevtechniques} with a complete overhaul of the technical framework laid out by~\cite{Tomamichel_2013}, albeit at the expense of restricting the $t+1$ players to make a single oracle query to $\mathsf{U}_f$.

\subsection{Application I: Black Hole Cloning Games}\label{sec:applicationbh}

As one application of our techniques on cloning games, we study the notion of a \emph{black hole cloning game}---a three-player interactive game which is designed to capture no-cloning and entanglement monogamy which arises naturally in the context
of evaporating black holes. The main result we discuss in this section is an asymptotically tight upper bound on the success probability of a variant of the game.
In particular, we observe that the analysis of black hole cloning games is inextricably linked to the existence of standard cloning games which have asymptotically optimal bounds of the form $O(2^{-n})$---well beyond the pre-existing upper bound of $2^{-0.25n}$ from the analysis by~\cite{schleppy2025winning} of the coset state game~\cite{10.1007/978-3-030-84242-0_20, Culf_2022}. Our new contributions on \emph{optimal} games allow us to fill this gap, and to complete the analysis.

In this section, we first provide some relevant context on black holes, and then give an overview of how we revisit the problem using the language of cloning games. We present more detailed context and results in~\Cref{sec:black-hole-games}. 

\paragraph{Hayden-Preskill thought experiment.} 

Hawking~\cite{PhysRevD.14.2460} made the remarkable prediction that black holes are not completely black---they slowly emit what is now known as \emph{Hawking radiation}. But if black holes evaporate, what happens to information that falls inside of a black hole? Does it get destroyed, or is it effectively conserved and eventually radiated out in some scrambled form? This question has puzzled physicists for many decades. The endeavour of trying to reconcile the predictions of quantum mechanics and general relativity has led to the famous \emph{black hole information paradox}~\cite{PhysRevD.14.2460,Preskill:1992tc}.

Hayden and Preskill~\cite{Hayden_Preskill_2007} proposed a thought experiment that illustrates the black-hole information loss problem: Suppose that Alice throws $k$ qubits into a black hole, which are maximally entangled with a second register in her possession. For simplicity, we assume that the black hole initially consists of $n-k$ qubits. After a long period of time, another distant observer, say Bob, uses the intercepted Hawking radiation (say, in the form of photons) which he has collected in the meantime, feeds it into his quantum computer and applies an appropriate computation in an attempt to recover Alice's quantum state. Hayden and Preskill asked: how long would Bob have to wait before he finally starts to observe correlations between the outgoing Hawking radiation and the entangled infalling matter near the boundary? To answer this question, they made the following crucial assumption: black holes are extremely strong and efficient \emph{information scramblers}---their internal dynamics can be modeled as a more or less \emph{Haar random} unitary time-evolution.\footnote{Note that genuine \emph{Haar} dynamics have exponential circuit complexity with high probability~\cite{Knill00}. Hayden and Preskill opted for a weaker notion than Haar randomness which nevertheless suffices for their purposes; namely, that of a \emph{unitary $2$-design.}} This view has since been widely adopted as an idealized mathematical model of black hole evolution~\cite{Almheiri_2013,Harlow_2013,Kim_2023,engelhardt2024cryptographiccensorship}. Concretely, it assumes that, from the perspective of an outside observer, Alice's infalling information is scrambled by a random unitary and effectively spread across the entire horizon of the black hole, whereby the total amount of information---accounting for both the internal degrees of freedom, as well as Alice's infalling information---is encoded in qubits which lie at the surface of the black hole. Here, the \emph{surface} of the black hole refers to the \emph{stretched horizon}---a tiny region of space which is located "just outside" of the black hole horizon~\cite{PhysRevD.48.3743,Harlow_2013} and is typically considered to be part of the black hole.

 Hayden and Preskill~\cite{Hayden_Preskill_2007} showed that after slightly more than half of the black hole has evaporated (sometimes called the \emph{Page time}), Bob can in principle completely recover the information from Alice's infalling qubits by intercepting the outgoing Hawking radiation. This led them to conclude that black holes act as \emph{information mirrors}: while Alice's information remains concealed up
until the half-way point, it then starts to emerge fairly quickly in the form of scrambled Hawking radiation. 

\paragraph{Do Black Holes Clone Information?} While the Hayden-Preskill thought experiment~\cite{Hayden_Preskill_2007} suggests that the information from Alice's infalling qubits is ultimately preserved and encoded in the form of scrambled radiation, it does raise the question: how much of Alice's information is still retained by the black hole? Could it be that a distant observer, say Bob, can recover Alice's information from the outgoing radiation, and yet a "second copy" of Alice's information somehow also survives in the black hole interior?

Suppose that we partition the relevant qubits which result from the evolution of the black hole into two registers $\mathsf{H}$ and $\mathsf{R}$, where $\mathsf{H}$ corresponds to the qubits near the horizon which are still retained by the black hole, and where $\mathsf{R}$ corresponds to emitted Hawking radiation.
Hayden and Preskill~\cite{Hayden_Preskill_2007} consider two regimes: at the beginning of the experiment, $\mathsf{H}$ is significantly larger than $\mathsf{R}$ and virtually all of Alice's information is contained in $\mathsf{H}$; however, by the end of the experiment, when the black hole has evaporated long past the Page time (and most of the qubits have been transformed into radiation), $\mathsf{R}$ is now significantly larger than $\mathsf{H}$ and must therefore be highly correlated with Alice's infalling qubits. Is it possible that Alice's information is retained by the black hole (and thus present in $\mathsf{H}$) but
simultaneously also encoded in $\mathsf{R}$---even long after the Page time?
Could it be that black holes \emph{clone} information?\footnote{Hayden and Preskill~\cite{Hayden_Preskill_2007} explored this question using another and much more radical thought experiment: suppose that Bob quickly decodes Alice's information from the intercepted Hawking radiation, and then immediately jumps inside of the black hole and crosses the event horizon in an attempt to find a "second copy" of Alice's information. A series of follow-up works~\cite{Almheiri_2013,Harlow_2013,aaronson2016complexityquantumstatestransformations}
have since exposed and studied paradoxes which emerge out of this experiment, and which have led to the belief that black-hole radiation decoding must necessarily be \emph{computationally intractable}~\cite{Harlow_2013,10.1007/978-3-031-38554-4_2,bostanci2023unitarycomplexityuhlmanntransformation}.}
Note that the standard no-cloning theorem~\cite{1982Natur.299..802W} and its approximate variants~\cite{Bu_ek_1996} do not suffice to immediately answer this question; for example, if Alice's infalling information is \emph{classical} rather than \emph{quantum}, the black hole may not even need to fully clone a quantum state to retain her information; it merely needs to maintain a classical correlation with Alice's infalling bits. Therefore, to answer this question, a new approach is necessary.


\begin{figure}[t]
    \centering
    \includegraphics[width=0.99\linewidth]{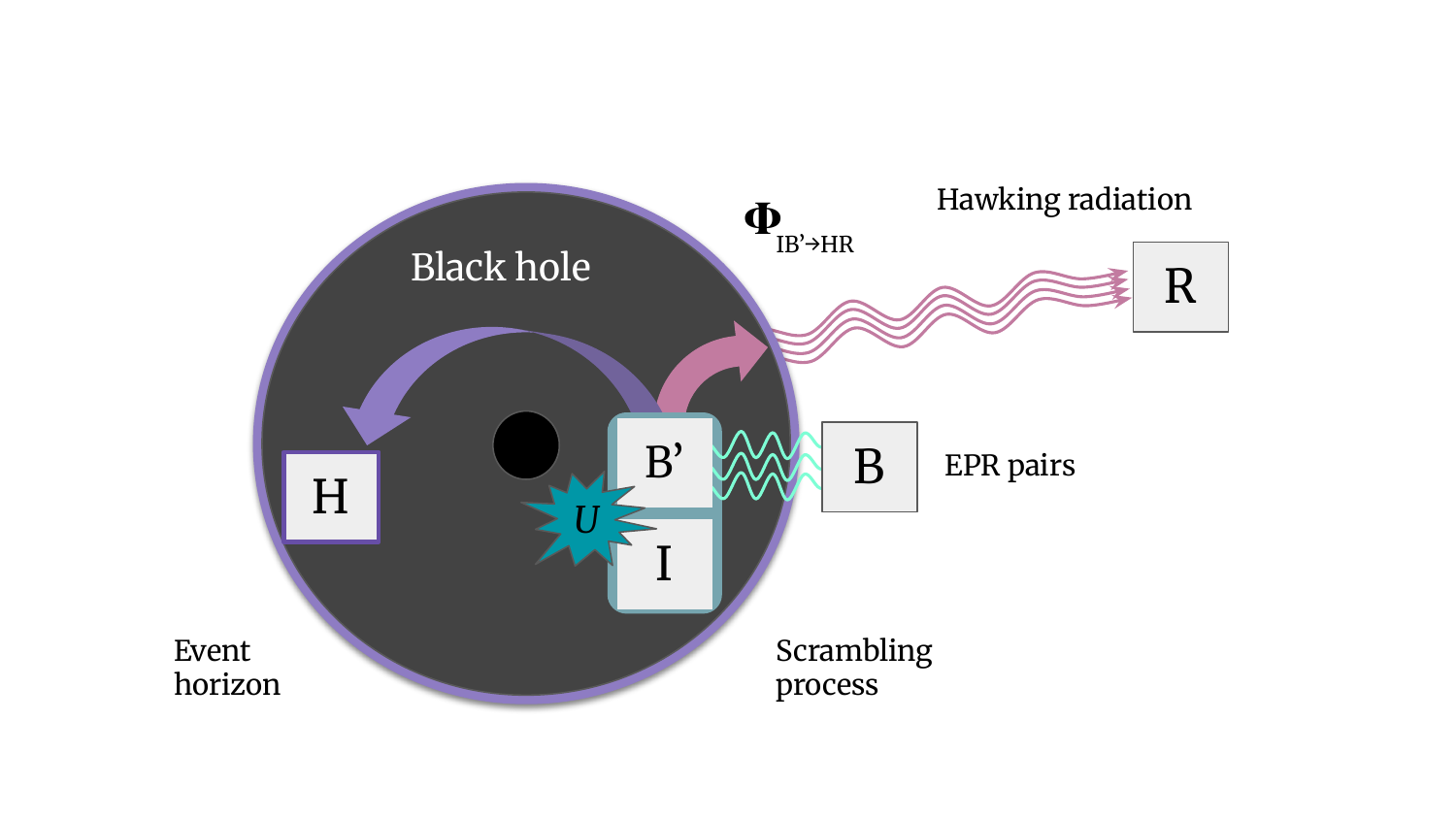}
    \caption{\textbf{Black Hole Cloning Game}. 
    Entangled particles emerge near the boundary and form a $k$-qubit EPR pair $\ket{\mathsf{EPR}}_{\reg{B'B}}$, of which register $\reg{B'}$ falls inside of the black hole, and register $\reg{B}$ is given to Alice. The   
   interior of the black hole, modeled as $\ket{0^{n-k}}_{\reg{I}}$, together with the $k$ infalling qubits in register $\reg{B'}$, undergo a \emph{scrambling process}. Here, the internal dynamics of the black hole are described by a random $n$-qubit unitary time-evolution operator $U\sim \nu$ which gets applied to registers $\reg{B'I}$. A quantum channel $\Phi_{\reg{IB' \rightarrow HR}}$ processes the internal qubits into two registers: a register $\reg{H}$ corresponding to the qubits at the event horizon which are retained by the black hole, and a register $\mathsf{R}$ corresponding to the emitted Hawking radiation. Charlie (who is anchored at the horizon) receives $\mathsf{H}$, whereas Bob (who
is a distant observer) receives $\mathsf{R}$. The two observers are allowed to have some knowledge of the internal dynamics $U$, and thus receive oracles for $U$ and $U^\dag$. Finally, Alice measures $\reg{B}$, and Bob and Charlie win if they simultaneously guess her outcome correctly.
    }
    \label{fig:black-hole-scrambling}
\end{figure}

\paragraph{This Work: Revisiting the Black Hole Information Paradox.} We seek to extend our existing understanding of the black hole information paradox in two ways. The first is that we would like to provide a new and \emph{quantitative} characterization of cloning and entanglement monogamy which arises in the context of evaporating black holes. Many seminal works~\cite{PhysRevLett.23.880, Tomamichel_2013} have quantified and enhanced our understanding of physical principles (e.g., the nature of entanglement) through the formulation and analysis of certain interactive games; our first goal is to do the same for the black hole information paradox:

\begin{quote}
    \begin{center}
      {\em \textbf{Question One}: Can we give a quantitative trade-off for how much of Alice's infalling information can be retained by the black hole, and how much can be present in its radiation?}
      \end{center}
    \end{quote} 
\noindent
Note that such a trade-off would significantly extend the analysis by Hayden and Preskill~\cite{Hayden_Preskill_2007} who merely studied two extreme cases, i.e., when when most of the qubits are either retained by the black hole or when most of the qubits have been converted into Hawking radiation.
As it turns out, however, such an information-theoretic analysis seems to lie way beyond the scope of existing techniques. (We will discuss the limitations of these existing techniques later in this section, as well as in Sections~\ref{sec:prevtechniques} and~\ref{sec:techoverviewhaar}.)

Secondly, we aim to revisit prior attempts for how to model the decoder's knowledge of the internal dynamics of the black hole. In the seminal Hayden-Preskill thought experiment~\cite{Hayden_Preskill_2007}, the authors assume that the decoder (say, Bob) holds a quantum memory that is maximally entangled with the qubits in the interior of the black hole. In other words, Bob is an extremely powerful observer that has complete control over the black hole and its resulting radiation. As noted by Hayden and Preskill, we would ideally like a more realistic model that captures Bob's knowledge of the black hole dynamics without giving him direct control over the black hole, which raises the question:
\begin{quote}
    \begin{center}
      {\em \textbf{Question Two:} Are there alternative---and perhaps more reasonable---models that capture the fact that the decoder has knowledge of the internal dynamics of the black hole?}
      \end{center}
    \end{quote} 
    
\noindent We believe that an affirmative answer to these two questions could offer new and valuable insights into the black-hole information paradox.

\paragraph{Our Approach: Black Hole Cloning Games.} To address \emph{Question One}, we cast the famous black-hole information paradox into the form of a cloning game (see \Cref{fig:black-hole-scrambling}). At the beginning of the game, Alice throws her entangled qubits into the black hole. Later, at the end of the game, two spatially separated ``adversaries'' called Bob and Charlie will attempt to recover Alice's information---either as a \emph{distant observer} with access to the emitted Hawking radiation (say, Bob), or as an \emph{anchored observer} (say, Charlie) who remains at the event horizon and has access to remaining qubits which are retained by the black hole. Concretely, we imagine that Alice measures her half of the entangled state at the end of the experiment, and Bob and Charlie are asked to simultaneously predict her measurement outcome. In our setting, Alice's infalling information should be thought of as being \emph{classical} rather than \emph{quantum}; indeed, our black hole cloning game in \Cref{fig:black-hole-game} is equivalent to a game in which Alice throws random classical bits into the black hole\footnote{This is an important distinction compared to the Hayden-Preskill experiment, where Alice's information can actually be thought of as being \emph{quantum} and where Bob tries to decode it in a coherent fashion by exploiting entanglement as a resource. Nevertheless, our setting captures the inherent correlation trade-off between Alice's system and the registers $\mathsf{H}$ and $\mathsf{R}$ in a similar spirit.

}. We work with the purified picture for purely aesthetic purposes; it helps us keep the notation consistent with Hayden and Preskill, and also makes the analysis via monogamy of entanglement games much more direct.

In line with prior works~\cite{Hayden_Preskill_2007,Harlow_2013,Kim_2023,engelhardt2024cryptographiccensorship}, we model the black hole's internal evolution in between as a ``scrambling process'' which is the result of some random unitary time-evolution $U$, followed by an arbitrary quantum channel $\Phi$ that processes the internal qubits into two systems: one corresponding to the qubits near the event horizon of the black hole (denoted by $\mathsf{H}$), and another corresponding to the emitted Hawking radiation (denoted by $\mathsf{R}$). In the work of Hayden and Preskill~\cite{Hayden_Preskill_2007}, $\Phi$ is essentially just a simple unitary channel that randomly partitions the scrambled qubits into two subsystems $\mathsf{H}$ and $\mathsf{R}$ of different sizes (in particular, where $\mathsf{R}$ is typically much larger than $\mathsf{H}$), whereas in our work we let $\Phi$ be an arbitrary (and possibly unitary)
completely-positive and trace-preserving map. This significantly generalizes the setting considered by Hayden and Preskill~\cite{Hayden_Preskill_2007} and, in particular, includes \emph{all possible partitionings} into two systems $\mathsf{H}$ and $\mathsf{R}$.
As in a conventional cloning game, it is also crucial that Bob and Charlie do not communicate while the decoding phase is taking place, which also consistent with our modeling assumption that Charlie is anchored at the horizon of the black hole, whereas Bob remains a distant observer.

To address \emph{Question Two}, we grant Bob and Charlie \emph{oracle access} to the internal scrambling dynamics $U$, as well as its inverse $U^\dag$. Additionally, we assume that Bob and Charlie have a complete description of the physical process $\Phi$ that results in the outgoing radiation. While Bob (and similarly, also Charlie) no longer has the ability to exercise direct control over the black hole dynamics (as in the Hayden-Preskill model), he does have the power (via the oracle for $U^\dag$) to instantaneously ``unscramble'' the black hole's time evolution at will. Here, the oracle access to the unitaries $U, U^\dag$ is meant to reflect the possibility that Bob and Charlie are powerful observers that have obtained some knowledge on the physical equations and parameters governing the black hole's evolution (see Figure~\ref{fig:bobcharlieBH} for a quantum circuit representation).

\paragraph{Analyzing Black Hole Cloning Games.}

Given the similarity between our black hole cloning game and the games studied in~\cite{Tomamichel_2013, broadbent_et_al:LIPIcs.TQC.2020.4}, this raises the question of whether one can indeed interpret one as an instance of the other.
Our main technical insight is that the analysis of black hole cloning games is inextricably linked to the existence of standard $1 \mapsto 2$ cloning games which have asymptotically optimal bounds of the form $O(2^{-n})$---well beyond the pre-existing upper bound of $2^{-0.25n}$ from the analysis by~\cite{schleppy2025winning} of the coset state game~\cite{10.1007/978-3-030-84242-0_20, Culf_2022}. 
In~\Cref{thm:black-hole-one-query}, we prove the following result without any restrictions on the choice of quantum channel $\Phi$; however, for technical reasons, we let $\nu$ be a unitary $3$-design and we assume that Bob and Charlie employ single-query strategies only. We visualize what Bob and Charlie's strategies might look like in~\Cref{fig:bobcharlieBH}, and we mention further potential improvements in \Cref{sec:openquestions}.

\begin{theorem}[Informal, see~\Cref{thm:black-hole-one-query} for formal statement]\label{thm:informalblackhole} Let $n,k \in \N$ be integers with $n \geq k$ and let $\nu = \{U_\theta\}_{\theta \in \Theta}$ be an $n$-qubit unitary $3$-design. Then, for any quantum channel $\Phi$ (of appropriate dimensions), the maximal single-query value $\omega(\mathsf{G}_{\text{\tiny{BH}}})$ of the black hole cloning game $\mathsf{G}_{\text{\tiny{BH}}}$ (as illustrated in \Cref{fig:black-hole-scrambling}) 
 with respect to $\nu$ and $\Phi$ is at most $O(2^{-k})$.
\end{theorem}

The bulk of our work in~\Cref{sec:black-hole-games} is to show that the maximal value $\omega(\mathsf{G}_{\text{\tiny{BH}}})$ (see Definition~\ref{def:bh-value} for a formal definition) can always be related to the maximal value of a related (but standard) cloning game $\mathsf{G}_{\text{\tiny{clone}}}$.
Specifically, we show that the game $\mathsf{G}_{\text{\tiny{BH}}}$ emerges as a special case of $\mathsf{G}_{\text{\tiny{clone}}}$ in which we post-select on the event that Alice's sampled message $y$ takes the form $y =x||0^{n-k}$, for some $x \in \bit^k$. Because this event occurs with probability $2^{-n+k}$, this allows us to deduce that $$ \underset{\text{strategies }\mathsf{S}} {\mathrm{sup}}\,\,\omega_{\mathsf{S}}(\mathsf{G}_{\text{\tiny{clone}}}) \, \geq \, 2^{-n+k} \, \cdot \, \omega(\mathsf{G}_{\text{\tiny{BH}}}).$$ Therefore, in order to obtain an asymptotically optimal bound of the form $\omega(\mathsf{G}_{\text{\tiny{BH}}}) = O(2^{-k})$, it suffices to show that the related cloning game $\mathsf{G}_{\text{\tiny{clone}}}$ has a maximal value of $\mathrm{sup}_{\text{strategies }\mathsf{S}}\,\,\omega_{\mathsf{S}}(\mathsf{G}_{\text{\tiny{clone}}}) = O(2^{-n})$. \textbf{Crucially, we require an $O(2^{-n})$ bound; a bound of the form $O(2^{-cn})$ for any $c < 1$ is insufficient.} This would yield $\omega(\mathsf{G}_{\text{\tiny{BH}}}) \leq 2^{-k} \cdot 2^{n(1-c)}$, which is a completely trivial bound if we assume $n \gg k$ (which is likely since presumably the black hole is a much larger system than the set of qubits Alice throws inside). As briefly mentioned in the introduction and elaborated in~\Cref{sec:prevtechniques}, our work is the first to show that a cloning game (even if in a restricted query model) admits an asymptotically optimal bound of the form $O(2^{-n})$. We explain this final point some more in Remark~\ref{remark:bhblcomparison}.


\paragraph{Implications for Black-Hole Physics.} 

\Cref{thm:informalblackhole}
yields the first quantitative trade-off for how much of Alice’s
information in the form of $k$ infalling qubits can be retained by the black hole, and how much can be present in its emitted Hawking radiation.
In fact, our bound of $ O(2^{-k})$ is also optimal (up to constant factors), since there always exists a particular Hawking radiation channel $\Phi$ together with a trivial strategy that attains it: we can consider a variant of the black hole cloning game where $\Phi$ is the channel that converts the entirety of all the qubits inside of the black hole into radiation (i.e., acting as the identity), which would allow Bob to perfectly recover the information from Alice's system by simply applying the inverse of the scrambling unitary. Now Charlie can guess randomly and succeed with probability $2^{-k}$.

We believe that our bound has several interesting implications. First, it suggests that the moment Bob has produced a register which is nearly maximally correlated with Alice's infalling qubits, then any additional qubits that lie in the interior of the black hole (i.e., in Charlie's system), must be almost completely uncorrelated from them. Second, such a strong decoupling result is achieved for \emph{any} choice of Hawking radiation channel $\Phi$---it arises precisely because of the strong scrambling properties of the unitary $3$-design itself. By contrast, the same would not be true for a \emph{classical} model of black-hole scrambling~\cite{Hayden_Preskill_2007}, say in the form of a random reversible circuit or a random permutation\footnote{Interestingly, one could interpret the one-round variant of the \emph{sponge construction} which underlies the international hash standard SHA-3~\cite{KeccakSub3,10.1007/978-3-031-68391-6_7,carolan2024quantumindifferentiabilityprecomputation} as a classical model of black hole scrambling, where the scrambling unitary is given by a random permutation and the Hawking radiation channel is an erasure channel that selects a subset of the final output bits.}. Despite the fact that a random permutation is already exponentially complex (i.e., requires exponential-sized circuits with overwhelming probability), it is simply \emph{not sufficiently scrambling} to allow for a similar decoupling to hold.

In summary, our results suggest that, in an idealized model of a black hole which features Haar random (or pseudorandom) scrambling dynamics, the information from infalling entangled qubits can only be recovered from either the interior (specifically, at the event horizon) or the exterior of the black hole (i.e., in the form of distant Hawking radiation), though never from both places at the same time.

\subsection{Application II: Unclonable Cryptography}\label{sec:introue}


In this section, we describe our applications to quantum cryptography; specifically, for how to construct succinct unclonable encryption schemes from the existence of \emph{pseudorandom unitaries}. This allows us to narrow the gap between the world of quantum pseudorandomness and unclonable cryptography, building on previous work by~\cite{cryptoeprint:2018/544} that showed that private-key quantum money exists assuming pseudorandom states.

\paragraph{Unclonable Encryption.} Cloning games have played a foundational role in the field of \emph{unclonable cryptography}---a branch of quantum cryptography that capitalizes on quantum no-cloning~\cite{1982Natur.299..802W} to achieve guarantees of ``unclonable security'' which are completely impossible classically. These include unclonable encryption~\cite{broadbent_et_al:LIPIcs.TQC.2020.4, 10.1007/978-3-031-38554-4_3, kundu2023deviceindependentuncloneableencryption, aky24}, encryption with unclonable decryption keys~\cite{cryptoeprint:2020/877}, unclonable commitments and proofs~\cite{cryptoeprint:2023/1538}, quantum copy-protection~\cite{10.1007/978-3-031-15979-4_8, coladangelo2022quantum}, and unclonable quantum advice~\cite{broadbent2023uncloneablequantumadvice}.
Most of these constructions rely at minimum on the existence of post-quantum one-way functions, placing unclonable cryptography in ``Post-Quantum MiniCrypt''~\cite{impfiveworlds}.\footnote{The work by~\cite{broadbent_et_al:LIPIcs.TQC.2020.4} does imply an information-theoretic construction of unclonable encryption based on BB84 states; however, this lacks succinctness as the size of the encryption and decryption keys scales with the message length $n$ rather than just the security parameter $\secp$.} 

Of particular interest to us is unclonable encryption, which is very closely related to cloning games. Indeed, any cloning game with value $\negl(n)$ immediately implies an unclonable encryption scheme with message space $\mathcal{X} = \bit^n$ that satisfies unclonable search security: to encrypt $x$ under secret key $\theta$, output $U_\theta \ket{x}$. This scheme has the obvious shortcoming that it is deterministic, and hence does not satisfy the ideal notions of indistinguishable or unclonable-indistinguishable security.

However, Broadbent and Lord~\cite{broadbent_et_al:LIPIcs.TQC.2020.4} proposed the following transformation that plausibly transforms an unclonable search secure scheme $\mathsf{SearchEnc}$ to an unclonable indistinguishable secure scheme: to encrypt $x \in \bit^n$ under secret keys $k, \theta$, sample a random PRF seed $r \in \bit^\secp$ and output the classical string $x \oplus \mathsf{PRF}(k, r)$ together with the quantum state $\ket{\mathsf{SearchEnc}(\theta, r)}$. Broadbent and Lord~\cite{broadbent_et_al:LIPIcs.TQC.2020.4} also provided some mild evidence that this may be secure if the PRF is instantiated with a random oracle. We emphasize that \emph{proving} the unclonable-indistinguishable security of this transformation (or a similar one) is a notoriously difficult open problem~\cite{kundu2023deviceindependentuncloneableencryption, 10.1007/978-3-031-38554-4_3, aky24}. Our point is just that studying unclonable encryption in the weaker search-secure setting is still an interesting and relevant cryptographic problem.

\paragraph{Multi-Copy Unclonable Cryptography.} Previous works on unclonable cryptography have exclusively focused exclusively on the case of $1 \mapsto 2$ cloning games, e.g. in the case of unclonable encryption~\cite{broadbent_et_al:LIPIcs.TQC.2020.4}, the adversarial cloner is given only one copy of a ciphertext state and aims to provide two receivers, say Bob and Charlie, with sufficient information to later recover the plaintext message.

These cryptographic primitives could naturally be extended to $t \mapsto t+1$ security: in the case of unclonable encryption, the cloner receives $t$ identical copies of a ciphertext state and aims to provide $t+1$ receivers with enough information to later recover the plaintext message. This raises the following question:
\begin{quote}
    \begin{center}
      {\em Can we construct $t \mapsto t+1$ unclonable cryptography from well-founded assumptions?}
      \end{center}
    \end{quote}
As explained in~\Cref{sec:prevtechniques}, existing constructions and techniques have little to say about this question. Answering this would resolve a question which was recently left open in~\cite{ananth2024revocableencryptionprogramsmore}, who asked whether the desirable property of multi-copy security is within reach in unclonable cryptography more generally.

\paragraph{Quantum Cryptography in ``MicroCrypt''.} Meanwhile, another line of work~\cite{cryptoeprint:2018/544, brakerski2019pseudo, DBLP:journals/corr/abs-2404-12647, hm24,bostanci2024efficientquantumpseudorandomnesshamiltonian} has introduced and constructed notions of \emph{pseudorandom quantum states and unitaries}. These are implied by the existence of post-quantum one-way functions; however, the reverse implication is not known. In fact, recent work~\cite{DBLP:conf/tqc/Kretschmer21, DBLP:conf/coco/AaronsonIK22, DBLP:conf/stoc/KretschmerQST23} has provided evidence that such an implication is unlikely to exist.  This has led to the development of new and \emph{inherently quantum} assumptions~\cite{bostanci2024efficientquantumpseudorandomnesshamiltonian,poremba2024learningstabilizersnoiseproblem}. As a result, the quantum cryptographic landscape includes yet another world, sometimes referred to as ``MicroCrypt'', which is potentially even weaker than that of MiniCrypt.

Moreover, pseudorandom states have proven to be powerful cryptographic tools in quantum cryptography, implying commitments~\cite{DBLP:conf/crypto/MorimaeY22} and oblivious transfer~\cite{DBLP:conf/crypto/BartusekCKM21a, DBLP:conf/eurocrypt/GriloL0V21}, and more. The fact that such powerful primitives live in MicroCrypt raises the following question:
\begin{quote}
    \begin{center}
      {\em What unclonable cryptographic primitives exist in MicroCrypt?}
      \end{center}
    \end{quote} 
\noindent
In fact, the authors of~\cite{DBLP:journals/corr/abs-2404-12647} explicitly asked whether pseudorandom unitaries (which have eluded major cryptographic application so far) imply the existence of unclonable cryptographic primitives. Previous work~\cite{cryptoeprint:2018/544} showed that private-key quantum money exists assuming just the existence of pseudorandom states. However, to our knowledge, this is the only known link between the worlds of unclonable cryptography and MicroCrypt before this work.

\paragraph{Our Results.} In this work, we make progress towards these foundational questions. 
Our main result of this section is the following:

\begin{theorem}[Informal, see Theorems~\ref{thm:UEfromPRU} and~\ref{thm:uemulticopy} for formal statements]\label{thm:informalue}
    We show the following statements:
    \begin{enumerate}
        \item Assuming the existence of pseudorandom unitaries, there exists an 
        unclonable encryption scheme with succinct keys which satisfies oracular $1 \mapsto 2$ search security (i.e., the adversaries are computationally bounded and only given oracle access to encryption and decryption functionality).

        \item If the message space is $\mathcal{X} = \bit^n$ and we fix $t = o(n/\log n)$ then, assuming the existence of post-quantum pseudorandom functions, there exists an 
        unclonable encryption scheme with succinct keys satisfies oracular $t \mapsto t+1$ search security, \textbf{provided} that the adversaries are computationally bounded and can only make \textbf{a single oracle query} to either the encryption or decryption functionality.
    \end{enumerate}
\end{theorem}
\begin{remark}
    Although we are only able to prove security for $t = o(n/\log n)$, we remark that our construction is plausibly secure for $t$ that is an arbitrary polynomial in $n$ (unlike previous constructions based on BB84 states~\cite{broadbent_et_al:LIPIcs.TQC.2020.4} and coset states~\cite{10.1007/978-3-030-84242-0_20}). Our justification for the plausible security of this construction is the fact that binary phase states are pseudorandom~\cite{cryptoeprint:2018/544, brakerski2019pseudo} and hence multi-copy unclonable~\cite{werneroptimalcloning}.
\end{remark}

\begin{remark}
    Broadbent and Lord~\cite{broadbent_et_al:LIPIcs.TQC.2020.4} prove a result that appears similar to our result for the $1 \mapsto 2$ case. The main difference is that they assume the existence of post-quantum PRFs, while we only need PRUs which is likely a weaker assumption~\cite{DBLP:conf/tqc/Kretschmer21}. Additionally, they assume a different oracle model; they instantiate their PRF as a random oracle, while we give the adversaries oracle access to unitaries for encryption and decryption.
\end{remark}
Our proof of the first result uses the BB84 cloning game~\cite{Tomamichel_2013} and its analysis by~\cite{broadbent_et_al:LIPIcs.TQC.2020.4}, together with the worst-case to average-case reduction outlined in~\Cref{sec:techoverviewhaar} and fleshed out in~\Cref{sec:worsttoav}. The latter can be thought of as an additional cryptographic application of pseudorandom unitaries\footnote{To the best of our knowledge, the use of pseudorandom unitaries in the context of efficient worst-case to average-case reductions has not previously appeared.} which was previously not known. While this gives us a security bound of $2^{-0.228n} + \negl(\secp)$ rather than the ideal $O(2^{-n}) + \negl(\secp)$, this is not a crucial difference for this application (unlike in the black hole setting); moreover, the analysis of the subspace coset monogamy game has the advantage that it does not need to restrict the queries made by the players.\footnote{We note for completeness that we could just as easily have used the analysis by~\cite{schleppy2025winning} of the subspace coset monogamy game~\cite{10.1007/978-3-030-84242-0_20, Culf_2022} in the place of the BB84 game and thus obtained a marginally stronger bound of $O(2^{-0.25n}) + \negl(\secp)$. We comment on this more in~\Cref{remark:uefromcoset}.}

Our proof of the second result uses our machinery outlined in~\Cref{sec:techoverviewbinaryphase} and fleshed out in Sections~\ref{sec:typesandsubtypes} and~\ref{sec:binaryphaseconstruction} for analyzing cloning games based on binary phase states. We emphasize that our construction is the first that could even be plausibly secure in the setting where $t$ can be an a priori unbounded polynomial in $\secp, n$ and the $t+1$ players are given the secret key $\theta$ in the clear. While our results are far from this ideal goal, we view our techniques as providing a stepping stone towards an ideal security result for unclonable encryption. We are also optimistic that our results and techniques might be adaptable to other problems in unclonable cryptography.

We visualize the landscape of some unclonable cryptographic primitives relative to the worlds of MicroCrypt, Post-Quantum Minicrypt, and Post-Quantum Obfustopia in~\Cref{fig:cryptoworlds}.

\input{fig-crypto-worlds}

\subsection{Open Questions}\label{sec:openquestions}

\paragraph{Cloning Games in General.} We first list some open questions related to cloning games in general; these would immediately yield applications to either or both of the black hole and unclonable encryption settings. We list some of these questions here:
\begin{enumerate}
    \item Can the security of the underlying $1 \mapsto 2$ oracular cloning game (i.e., as in Construction \ref{const:PRS-enc}) be proven even if the two players (say, Bob and Charlie) can adaptively make \emph{any} polynomial number of queries to the encoding underlying unitary and its inverse?

    This would immediately imply the security of our black hole cloning game against \emph{arbitrary} Bob and Charlie strategies, when instantiated with a pseudorandom unitary (PRU) rather than a unitary design. Due to their highly efficient (and yet strong) scrambling properties, pseudorandom unitaries are believed to be an excellent theoretical model of black hole dynamics~\cite{Kim_2023,engelhardt2024cryptographiccensorship}.

    \item More tantalizingly, can this security be shown if the measurement basis $\theta$ (either a PRU or PRF secret key, depending on whether we are considering the PRU or binary phase construction) is given to Bob and Charlie in the clear, rather than in the form of an oracle? This still plausibly satisfies unclonable security (as demonstrated by~\cite{Tomamichel_2013, broadbent_et_al:LIPIcs.TQC.2020.4, Culf_2022, schleppy2025winning} for BB84 and coset state cloning games), but is highly counterintuitive; the PRU/PRF security guarantees do not say anything about what could happen in a game where the secret key $\theta$ is eventually leaked.

    In our black hole cloning game, this would allow us to prove much stronger quantitative statements, even in the scenario in which Bob and Charlie have \emph{complete} knowledge of the internal scrambling dynamics of the black hole. 

    \item Can we achieve any of the above stronger security guarantees for $t \mapsto t+1$ cloning games? Or as a starting point: can we prove security against players $\algo{P}_1, \ldots, \algo{P}_{t+1}$ that are free to make multiple \emph{non-adaptive} queries to $U_\theta, U_\theta^\dag$?
\end{enumerate}

\paragraph{Applications to Black Hole Physics and Beyond.} Here, we list some questions specific to our black hole application:
\begin{enumerate}

     \item Can we make our modeling assumptions in our definition of black hole cloning games in \Cref{sec:black-hole-games}  more physically realistic? For example, can we model the (initial) internal qubits of the black hole as a more general quantum state (potentially even entangled with the exterior) rather than as the all-zero state $\ket{0^{n-k}}$? What if the scrambling dynamics do not just affect internal qubits, but also external qubits? And lastly, what if the scrambling dynamics is in the form of a Haar random isometry?

    \item Can we use the language of interactive games to offer new quantitative insights into information scrambling in other chaotic quantum systems, besides black holes?    
\end{enumerate}
\paragraph{Applications to Unclonable Cryptography.} Finally, we present some questions specific to our applications to unclonable cryptography:
\begin{enumerate}
    \item What other unclonable cryptography primitives can be instantiated in MicroCrypt?
    \item Can we obtain unclonable encryption with the stronger notion of indistinguishability security that we usually require of encryption schemes? (Our notion of unclonable security takes the form of ``search security'', which as we argue in~\Cref{sec:introue} offers a plausible but not yet proven path towards indistinguishable security.) This is an important but difficult problem that recent works have made some progress on~\cite{broadbent_et_al:LIPIcs.TQC.2020.4, kundu2023deviceindependentuncloneableencryption, 10.1007/978-3-031-38554-4_3, aky24}.

    \item Which unclonable cryptography primitives have natural, constructible, and applicable $t \mapsto t+1$ analogues, besides unclonable encryption?
\end{enumerate}

\paragraph{Organization of the Paper.} The remainder of this paper is organized as follows:
\begin{itemize}
    \item In~\Cref{sec:prelims}, we present some preliminaries including our novel spectral bounds in~\Cref{sec:linalg} on blockwise tensor products of matrices.

    \item In~\Cref{sec:monogamydefs}, we formally define monogamy of entanglement games and cloning games.

    \item In Sections~\ref{sec:typesandsubtypes} and~\ref{sec:binaryphaseconstruction}, we introduce our novel notion of binary subtypes and apply this to prove $O_t(2^{-n})$ cloning bounds for our binary phase state cloning game in the single-query-oracular setting.

    \item In~\Cref{sec:monogamyexisting}, we revisit pre-existing techniques for analyzing cloning and monogamy games, and demonstrate that there are inherent limitations that likely prevent them from being adaptable to our applications in either black hole physics or unclonable cryptography.

    \item In~\Cref{sec:worsttoav}, we prove a worst-case to average-case reduction for cloning games that allows us to adapt our result for binary phase state cloning games to Haar cloning games. This is integral both to our application to black hole physics, and to proving the existence of $1 \mapsto 2$ unclonable encryption assuming only PRUs.

    \item In~\Cref{sec:black-hole-games}, we formally define black hole cloning games, and use our aforementioned technical results on cloning games to prove an upper bound on the value of the black hole cloning game.

    \item Finally, in~\Cref{sec:unclonableenc}, we define the notion of succinct unclonable encryption schemes and show that it exists in an oracle model, assuming the existence of pseudorandom unitaries. We also provide a first result towards establishing multi-copy security of the same scheme.
\end{itemize}

\ifanon
\else
    \paragraph{Acknowledgements.} The authors would like to thank
    Aditya Nema, Aparna Gupte, Aram Harrow, Fermi Ma, Henry Yuen, Jiahui Liu, John Bostanci, Jonas Haferkamp, Jonathan Lu, Joseph Carolan, Lisa Yang, Makrand Sinha, Netta Engelhardt, Peter Shor, Prabhanjan Ananth, Ran Canetti, Saachi Mutreja, Soonwon Choi, Thomas Vidick, Tony Metger, William Kretschmer, and Yael Tauman Kalai for useful discussions.
    AP is supported by the U.S. Department of Energy, Office of Science, National Quantum Information Science Research Centers, Co-design Center for Quantum Advantage (C2QA) under contract number DE-SC0012704.
    SR is supported by an Akamai Presidential Fellowship, the grants of VV, the Defense Advanced Research Projects Agency (DARPA) under Contract No. HR0011-25-C-0300, and Amazon Research Awards.
    VV is supported by DARPA under Agreement No. HR00112020023, NSF CNS-2154149 and a Simons Investigator Award. Any
    opinions, findings and conclusions or recommendations expressed in this material are those of the
    author(s) and do not necessarily reflect the views of the United States Government or DARPA.
\fi

%% file: fig-crypto-worlds.tex
\tikzstyle{assumptionblock} = [rectangle, draw, fill=blue!20, 
    text width=2.2cm, text centered, rounded corners, minimum height=2cm]
    
\tikzstyle{applicationblock} = [rectangle, draw, 
    text width=2.2cm, text centered, rounded corners, minimum height=2cm]
\tikzstyle{line} = [draw, -latex']

\begin{figure}
\begin{center}
\begin{tikzpicture}[node distance=15mm]
    \node [assumptionblock] (io) {Post-Quantum $i\mathcal{O}$ + LWE};

    \node [assumptionblock, below=15mm of io] (owf) {Post-Quantum OWFs};

    \node [assumptionblock, below=15mm of owf] (pru) {Pseudorandom Unitaries};

    \node [assumptionblock, below=of pru] (prs) {Pseudorandom States};

    \node [applicationblock, left=of io] (pkqm) {One-Shot Signatures, Quantum Lightning};

    \node [applicationblock, right=of owf] (indue) {IND-Secure$^*$ Unclonable Encryption};

    \node [applicationblock, right=of pru] (searchue) {Search-Secure Unclonable Encryption};

    \node [applicationblock, right=of indue] (sde) {Single-Decryptor Encryption};

    \node [applicationblock, left=of owf] (iup) {Interactive Unclonable Proofs};
    

    \coordinate (first_left) at ($(pkqm.center)!0.35!(iup.center)$);
    \coordinate (first_right) at ($(io.center)!0.35!(owf.center)$);
    \draw[line width=0.5mm, dotted] ($(first_left)!-0.5!(first_right)$) -- ($(first_left)!3.5!(first_right)$);

    \coordinate (second_left) at ($(owf.center)!0.35!(pru.center)$);
    \coordinate (second_right) at ($(indue.center)!0.35!(searchue.center)$);
    \draw[line width=0.5mm, dotted] ($(second_left)!-1.5!(second_right)$) -- ($(second_left)!2.5!(second_right)$);

    \coordinate (center_ref) at ($(owf.center)!0.5!(indue.center)$);

    \node [above=50mm of center_ref] {\large \textbf{Post-Quantum Obfustopia}};

    \node [above=15mm of center_ref] {\large \textbf{Post-Quantum MiniCrypt}};

    \node [below=14mm of center_ref] {\large \textbf{``MicroCrypt''}};


    \path [line] (io) -- node[midway,above] {\cite{shmuelizhandry}} (pkqm);

    \path [line] (owf) -- node[midway,above] {\cite{cryptoeprint:2023/1538}} (iup);

    \path [line] (owf) -- node[midway,above] {\cite{broadbent_et_al:LIPIcs.TQC.2020.4}} (indue);

    \path [line] (indue) -- node[midway,above] {\cite{cryptoeprint:2020/877}} (sde);

    \path [line] (owf) -- node[midway,left] {\cite{DBLP:journals/corr/abs-2404-12647, hm24}} (pru);

    \path [line] (indue) -- node[midway,right] {Trivial} (searchue);

    \path [line] (pru) -- node[midway, left] {Trivial} (prs);

    \path [line] (pru) -- node [midway, below] {\parbox{0.7cm}{\textbf{This\\work}}} (searchue);

    \path [line] ($(io.south)!0.6!(io.south west)$) -- node [midway, left] {Trivial} ($(owf.north)!0.6!(owf.north west)$);


    \draw[black] ($(current bounding box.south west)!-0.05!(current bounding box.north east)$) rectangle ($(current bounding box.south west)!1.05!(current bounding box.north east)$);

\end{tikzpicture}
\end{center}
    \caption{A visualization of some primitives in unclonable cryptography and the assumptions that are known to imply them (we focus here on primitives that are relatively well-understood and related to monogamy of entanglement games). We segment these assumptions into three worlds, loosely following~\cite{impfiveworlds}: Obfustopia, MiniCrypt, and MicroCrypt. MicroCrypt is a world where we only assume the existence of pseudorandom states and unitaries, which could plausibly hold even if $\mathsf{P} = \mathsf{NP}$~\cite{DBLP:conf/stoc/KretschmerQST23}. Powerful cryptographic primitives such as bit commitments~\cite{DBLP:conf/crypto/MorimaeY22} and oblivious transfer~\cite{DBLP:conf/crypto/BartusekCKM21a, DBLP:conf/eurocrypt/GriloL0V21} have been shown to exist in MicroCrypt; however, instantiations of unclonable cryptography \emph{with succinct keys} in MicroCrypt have proven more limited. Our work takes a step in this direction by showing that pseudorandom unitaries imply search-secure succinct unclonable encryption in an oracle model.\\($^*$We note for clarity that the existing results on indistinguishability-secure unclonable encryption all come with some kind of caveat e.g. existing in an oracle model and requiring that the adversaries are disentangled~\cite{broadbent_et_al:LIPIcs.TQC.2020.4}, or requiring quantum decryption keys~\cite{aky24}.)}\label{fig:cryptoworlds}
\end{figure}

%% file: prelims.tex
\section{Preliminaries}\label{sec:prelims}

\subsection{Quantum Computation}\label{sec:quantumprelims} For a comprehensive background, we refer to \cite{NielsenChuang11}. We denote a finite-dimensional complex Hilbert space by $\mathcal{H}$, and we use subscripts to distinguish between different systems (or registers). For example, we let $\mathcal{H}_{\reg A}$ be the Hilbert space corresponding to a system ${\reg A}$. 
The tensor product of two Hilbert spaces $\algo H_{\reg A}$ and $\algo H_{\reg B}$ is another Hilbert space denoted by $\algo H_{\reg{AB}} = \algo H_{\reg A} \otimes \algo H_{\reg B}$.
The Euclidean norm of a vector $\ket{\psi} \in \algo H$ over the finite-dimensional complex Hilbert space $\mathcal{H}$ is denoted as $\| \psi \| = \sqrt{\braket{\psi|\psi}}$. 
Let $\linear(\algo H)$
denote the set of linear operators over $\algo H$. A quantum system over the $2$-dimensional Hilbert space $\mathcal{H} = \mathbb{C}^2$ is called a \emph{qubit}. For $n \in \mathbb{N}$, we refer to quantum registers over the Hilbert space $\mathcal{H} = \big(\mathbb{C}^2\big)^{\otimes n}$ as $n$-qubit states. We use the word \emph{quantum state} to refer to both pure states (unit vectors $\ket{\psi} \in \mathcal{H}$) and density matrices $\rho \in \mathcal{D}(\mathcal{H)}$, where we use the notation $\mathcal{D}(\mathcal{H)}$ to refer to the space of positive semidefinite matrices of unit trace acting on $\algo H$. 
The \emph{trace distance} of two density matrices $\rho,\sigma \in \mathcal{D}(\mathcal{H)}$ is given by
$$
\mathsf{TD}(\rho,\sigma) = \frac{1}{2} \| \rho - \sigma \|_1.
$$
A quantum channel $\Phi: \linear(\algo H_{\reg A}) \rightarrow \linear(\algo H_{\reg B})$ is a linear map between linear operators over the Hilbert spaces $\algo H_{\reg A}$ and $\algo H_{\reg B}$. Oftentimes, we use the compact notation $\Phi_{\reg{A \rightarrow B}}$ to denote a quantum channel between $\linear(\algo H_{\reg A})$ and $\linear(\algo H_{\reg B})$. We say that a channel $\Phi$ is \emph{completely positive} if, for a reference system $\reg R$ of arbitrary size, the induced map $\id_{\reg R} \otimes \Phi$ is positive, and we call it \emph{trace-preserving} if $\Tr{\Phi(X)} = \Tr{X}$, for all $X \in \linear(\algo H)$. A quantum channel that is both completely positive and trace-preserving is called a quantum $\mathsf{CPTP}$ channel.
A \emph{unitary} $U: \linear(\mathcal{H}_{\reg A}) \rightarrow \linear(\mathcal{H}_{\reg A})$ is a special case of a quantum channel that satisfies $U^\dagger U = U U^\dagger = \id_{\reg A}$. When $U$ acts on a density matrix $\rho$, it maps $\rho \mapsto U \rho U^\dag$, and we will denote this channel by $U \fullcirc U^\dag$. Whenever $d=2^n$, we refer to the group of unitaries acting on $n$ qubits as $\mathrm{U}(d)$. An isometry is a linear map $V: \linear(\mathcal{H}_{\reg A}) \rightarrow \linear(\mathcal{H}_{\reg B})$ with $\dim(\mathcal{H}_{\reg B}) \geq \dim(\mathcal{H}_{\reg A})$ and $V^\dag V = \id_{\reg A}$.
A \emph{projector} ${\Pi}$ is a Hermitian operator such that ${\Pi}^2 = {\Pi}$, and a \emph{projective measurement} is a collection of projectors $\{{\Pi}_i\}_i$ such that $\sum_i {\Pi}_i = \id$.
A positive-operator valued measure ($\mathsf{POVM}$) is a set of Hermitian positive semidefinite operators $\{M_i\}$ acting on a Hilbert space $\mathcal{H}$ such that $\sum_{i} M_i = \id$.

Given a bipartite state $\rho_{\reg{AB}}$, the \emph{partial trace} $\mathrm{Tr}_{\reg{B}}$ captures the residual state of the system on just the $\reg{A}$ register. $\mathrm{Tr}_{\reg{B}}$ is thus defined as a linear map from $\linear(\algo{H}_{\reg{A}} \otimes \algo{H}_{\reg{B}}) \rightarrow \linear(\algo{H}_{\reg{A}})$ that maps $R \otimes S \mapsto \Tr{S} \cdot R$. Given a multipartite operator $X \in \linear(\algo{H}_{\reg{A}} \otimes \algo{H}_{\reg{B}} \otimes \algo{H}_{\reg{C}})$, the \emph{partial transpose} applies a transpose to only some of these systems. For example, the partial transpose $X \mapsto X^{\top_{\reg{B}}}$ with respect to the second system is defined as a linear map satisfying $X_1 \otimes X_2 \otimes X_3 \mapsto X_1 \otimes X_2^\top \otimes X_3$. We can also define a $\mathsf{SWAP}$ operator that acts on say registers $\reg{A}$ and $\reg{C}$; this is a linear map that will map $X_1 \otimes X_2 \otimes X_3 \mapsto X_3 \otimes X_2 \otimes X_1$.


\paragraph{Operators.} Define the following unitary operators:

\begin{itemize}
    \item Phase oracle: For $f: \bit^n \rightarrow \bit$ we let
    $$
    \mathsf{U}_f = \sum_{x \in \bit^n} (-1)^{f(x)} \proj{x}.
    $$

    \item Multi-bit Pauli operator: For $m \in \bit^n$, let
    $$
\mathsf{Z}^m = \sum_{x \in \bit^n} (-1)^{\langle x,m \rangle} \proj{x}.
$$
    \item Hadamard: The $n$-qubit Hadamard operator is defined by $$\mathsf{H}^{\otimes n} =  2^{-n/2} \sum_{x,y \in \bit^n} (-1)^{\langle x, y \rangle} \ket{x}\hspace{-0.5mm}\bra{y}.$$
\end{itemize}

\paragraph{Choi-Jamiołkowski isomorphism.} 

Let $\algo H_{\reg A}$ be a $d$-dimensional Hilbert space with an orthonormal basis denoted by $\{\ket{1},\dots,\ket{d}\}$. Let $\ket{\Omega} = \sum_{i \in [d]} \ket{i} \otimes \ket{i}$ be the vectorization of the identity $\id_d = \sum_{i \in [d]} \proj{i}$. Then, the  Choi-Jamiołkowski isomorphism $J(\Phi) \in \linear(\algo H_{\reg B} \otimes \algo H_{\reg A'})$ with respect to a linear map of the form $\Phi: \linear(\algo H_{\reg A}) \rightarrow \linear(\algo H_{\reg B})$ is defined as 
$$
J(\Phi) =  (\Phi_{\reg A \rightarrow \reg B} \otimes \id_{\reg{A}'})(\proj{\Omega}) = \sum_{i,j \in [d]} \Phi(\ketbra{i}{j}) \otimes \ketbra{i}{j}.
$$
We use the following well known fact.
\begin{lemma}\label{fact:CI}
Let $\Phi: \linear(\algo H_{\reg A}) \rightarrow \linear(\algo H_{\reg B})$ be a linear map. 
Then, for any $\ket{\psi} \in \state(\algo H_{\reg{A}})$ and $\ket{\phi} \in \state(\algo H_{\reg{B}})$,
$$
\bra{\phi} \Phi(\proj{\psi}) \ket{\phi} = \bra{\phi} \otimes \bra{\bar\psi} J(\Phi) \ket{\phi} \otimes \ket{\bar\psi} \, ,
$$
where the complex conjugation is taken with respect to the computational basis. Equivalently, we have:
$$\Tr{\proj{\phi} \Phi(\proj{\psi})} = \Tr{\left(\proj{\phi} \otimes \proj{\bar{\psi}}\right)J(\Phi)}.$$
\end{lemma}
\noindent
By linearity, we immediately obtain the following corollary:
\begin{corollary}\label{cor:CIprojector}
Let $\Phi: \linear(\algo H_{\reg A}) \rightarrow \linear(\algo H_{\reg B})$ be any linear map. Then, 
    for any Hermitian operators $\vec{P} \in \mathrm{L}(\algo{H}_{\reg B})$ and $\vec{Q} \in \mathrm{L}(\algo{H}_{\reg A})$, it holds that $$\Tr{\vec{P} \Phi(\vec{Q})} = \Tr{\left(\vec{P} \otimes \bar{\vec{Q}}\right)J(\Phi)}.$$
\end{corollary}

\subsection{Unitary Designs}
\label{sec:unitary-designs}

In this section, we formally define the  \emph{Haar measure}~\cite{Simon1995RepresentationsOF}  and also define a new and more general version of a \emph{mixed} unitary $t$-design which also allows for inverses with respect to the adjoint of the unitary.

\begin{definition}[Haar measure] 
Let $d \in \N$ denote the dimension.
The Haar measure $\mu_H$ is the unique left and right unitarily-invariant measure over the unitary group $\mathrm{U}(d)$; that is, for every (possibly matrix-valued) integrable function $f$ with domain $\mathrm{L}(\mathbb{C}^{d})$ and every unitary $V \in \mathrm{U}(d)$, 
  $$ \int_{\mathrm{U}(d)} f(U)\  \mathrm{d}_{\mu_H}U = \int_{\mathrm{U}(d)} f(U\cdot V)\  \mathrm{d}_{\mu_H}U =  \int_{\mathrm{U}(d)} f(V\cdot U)\  \mathrm{d}_{\mu_H}U. $$
For brevity, we oftentimes denote the expectation of $f$ over the Haar measure by
$$\underset{U \sim \mathrm{U}(d)}{\mathbb{E}} [f(U)] = \int_{\mathrm{U}(d)} f(U)\  \mathrm{d}_{\mu_H}U.$$
\end{definition}

\paragraph{Non-Adaptive Mixed Unitary Designs.}

In this section, we introduce a generalization of the standard notion of a unitary $t$-design which also accounts for inverse queries (to the adjoint of the unitary). We exclusively consider exact designs; in particular, exact unitary $3$-designs via the Clifford group~\cite{webb2016cliffordgroupformsunitary}.

\begin{definition}[Non-Adaptive Mixed Unitary $(p,q)$-Design]
Let $\nu$ be an ensemble of unitary operators over $\mathbb{C}^d$. Then, $\nu$ is a (non-adaptive) unitary $(p,q)$-design if, for every $\vec O \in \linear((\CC^d)^{\otimes (p+q)})$,
$$
\underset{U \sim \nu}{\mathbb{E}}\left[(U^{\otimes p} \otimes (U^\dag)^{\otimes q}) \vec O (U^{\otimes p} \otimes  (U^\dag)^{\otimes q})^\dag \right] = \underset{U \sim \mathrm{U}(d)}{\mathbb{E}}\left[(U^{\otimes p} \otimes (U^\dag)^{\otimes q}) \vec O (U^{\otimes p} \otimes  (U^\dag)^{\otimes q})^\dag \right].
$$
\end{definition}

Note that a unitary $t$-design is a special case of the above definition.

\begin{definition}[Non-Adaptive Unitary $t$-Design]
Let $\nu$ be an ensemble of unitary operators over $\mathbb{C}^d$. Then, $\nu$ is a (non-adaptive) unitary $t$-design if it is a (non-adaptive) unitary $(t,q)$-design for $q=0$.
\end{definition}

\paragraph{Adaptive Mixed Unitary Designs.}

In this section, we generalize the notion of mixed unitary designs to algorithms which may query a unitary (and possibly its inverse) adaptively, rather than in parallel.

\begin{definition}[Adaptive Mixed Unitary $(p,q)$-Design]
Let $\nu$ be an ensemble of unitary operators over $\mathbb{C}^d$. Then, $\nu$ is an adaptive unitary $(p,q)$-design if, for every single-bit output (possibly adaptive) quantum algorithm $\algo A$ making at most $p$ many queries to a unitary and $q$ many queries to its adjoint,
$$
\Pr\left[ 1 \leftarrow \algo A^{U,U^\dag }(1^{\lceil \log d \rceil}) \, : \, U \sim \nu \right] = \Pr\left[ 1 \leftarrow \algo A^{U,U^\dag }(1^{\lceil \log d \rceil}) \, : \, U \sim \mathrm{U}(d) \right].
$$
We say that $\nu$ is an adaptive mixed unitary $t$-design if the property above holds for any adaptive quantum query algorithm $\algo A$ which submits no more than $t$ queries to either $U$ or $U^\dag$.
\end{definition}

\subsection{Pseudorandom Unitaries}\label{sec:prudefs}

Pseudorandom unitaries are ensembles of unitary operators that look indistinguishable from Haar random unitaries for all computationally bounded observers. These ensembles of unitaries have been first proposed in~\cite{cryptoeprint:2018/544}, and have only very recently been constructed assuming the existence of post-quantum one-way functions~\cite{DBLP:journals/corr/abs-2404-12647,hm24}.
We give a formal definition below.

\begin{definition}[Pseudorandom Unitary]\label{def:PRU}
Let $\secp$ be the security parameter, $n := n(\secp) \in \N$ be some polynomial, and $d=2^n$. An infinite sequence $\mathfrak{U} = \{\mathfrak{U}_n\}_{n \in \N}$ of $n$-qubit unitary ensembles $\mathfrak{U}_n = \{U_{\theta, n}\}_{\theta \in \bit^\secp}$ is a pseudorandom unitary if it satisfies the following conditions:
\begin{itemize}
    \item (\textbf{Efficient computation}) For all $\secp, n$, there exists a polynomial-time quantum algorithm $\mathcal{Q}$ such that for all keys $\theta \in \bit^\secp$, and any $\ket{\psi} \in (\CC^2)^{\ot n}$, it holds that
    $$
    \mathcal{Q}(\theta,\ket{\psi}) = U_{\theta, n} \ket{\psi}\,.
    $$


    \item (\textbf{Pseudorandomness}) The unitary $U_{\theta, n}$, for a random key $\theta \sim \bit^\secp$, is computationally indistinguishable from a Haar random unitary $U \sim \mathrm{U}(d)$. In other words,  for any QPT algorithm $\algo A$, it holds that $$
    \vline\, \underset{\theta \sim \bit^\secp}{\Pr}[\algo A^{U_{\theta, n},U_{\theta, n}^\dag}(1^\secp, 1^n)=1] - \underset{U \sim \mathrm{U}(d)}{\Pr}[\algo A^{U_{\theta, n},U_{\theta, n}^\dag}(1^\secp, 1^n) =1]  \,\vline \,\leq \, \negl(\secp)\,.
    $$ 
\end{itemize}
\end{definition}

\begin{remark}
    We note that this definition of pseudorandom unitary is quite strong; the adversary $\algo{A}$ is free to make its queries adaptively, and moreover it is allowed to query both $U$ and $U^\dag$, in analogy to strong pseudorandom permutations. This notion was constructed in recent work by Ma and Huang~\cite{hm24}.
\end{remark}

\subsection{Operator Norm Bounds}\label{sec:linalg}

In this section, we lay out some tools for bounding the operator norm $\norm{\vec A}_\infty$ of operators $A \in \CC^{d \times d}$. For matrices $\vec A, \vec B$ of the same dimensions, we use $\vec A \circ \vec B$ to denote their entrywise product.

\begin{lemma}[Well-known]\label{lemma:opnormdef}
    For any matrix $\vec A \in \CC^{d_1 \times d_2}$, we have $$\norm{\vec A}_\infty = \sqrt{\lambda_{\max}(\vec A^\dag \vec A)} = \sqrt{\lambda_{\max}(\vec A\vec A^\dag)} = \max\left\{\norm{\vec Ax}_2: \norm{x}_2 = 1\right\}.$$
    Moreover, if $\vec A$ has rank $\leq 1$, then we have $\lambda_{\max}(\vec A^\dag \vec  A) = \Tr{\vec A^\dag \vec A}$.
\end{lemma}

\begin{lemma}[Well-known]\label{lemma:opnormsubmatrix}
    For any pair of matrices $\vec A \in \CC^{d_1 \times d_2}$ and $\vec A' \in \CC^{d'_1 \times d'_2}$ such that $\vec A'$ is a submatrix of $\vec A$, we have $\norm{\vec A'}_\infty \leq \norm{\vec A}_\infty$.
\end{lemma}

\begin{lemma}[Well-known]\label{lemma:opnormtensor}
    For $\vec A \in \CC^{d_1 \times d_2}$ and $\vec B \in \CC^{d_3 \times d_4}$, we have $\norm{\vec A \otimes \vec B}_\infty = \norm{\vec A}_\infty \cdot \norm{\vec B}_\infty$.
\end{lemma}

\begin{lemma}\label{lemma:hadamardunitarygeneral}
    Let $\vec A_1, \ldots, \vec A_k \in \CC^{d \times d}$ be unitary matrices with $k \geq 2$. Then, the rows and columns of $\vec C = \vec A_1 \circ \ldots \circ \vec A_k$ all have $\ell_1$ norm $\leq 1$.
\end{lemma}
\begin{proof}
    In the case of rows, we have:
    \begin{align*}
        \sum_{j = 1}^d |C_{i, j}| &= \sum_{j = 1}^d |A_{1; (i, j)}| \cdot \ldots \cdot |A_{k; (i, j)}| \\
        &\leq \sum_{j = 1}^d |A_{1; (i, j)}| \cdot |A_{2; (i, j)}| \text{ (all entries of a unitary are $\leq 1$)} \\
        &\leq \sqrt{\left(\sum_{j = 1}^d |A_{1; (i, j)}|^2\right) \cdot \left(\sum_{j = 1}^d |A_{2; (i, j)}|^2\right)} \text{ (Cauchy-Schwarz)} \\
        &= 1.
    \end{align*}
    The case of columns is analogous.
\end{proof}

\begin{lemma}\label{lemma:rowcolsumopnorm}
    Let $\vec C \in \CC^{d_1 \times d_2}$ be such that the $\ell_1$ norm of each row is $\leq a$ and the $\ell_1$ norm of each column is $\leq b$. Then $\norm{\vec C}_\infty \leq \sqrt{ab}$.
\end{lemma}
\begin{proof}
    For any row $i$ of $\vec C^\dagger \vec  C \in \CC^{d_2 \times d_2}$, we have:
    \begin{align*}
        \sum_{j = 1}^{d_2} |(\vec C^\dagger \vec C)_{i, j}| &= \sum_{j = 1}^{d_2} |\sum_{k = 1}^{d_1} C^\dagger_{i, k}C_{k, j}| \\
        &\leq \sum_{j = 1}^{d_2} \sum_{k = 1}^{d_1} |C_{k, i}| |C_{k, j}| \\
        &\leq a \cdot \sum_{k = 1}^d |C_{k, i}| \\
        &\leq ab.
    \end{align*}
    Since the maximum eigenvalue of a square matrix is at most the maximum $\ell_1$ norm of its rows, we have $\norm{\vec C^\dagger \vec C}_\infty \leq ab \Rightarrow \norm{\vec C}_\infty \leq \sqrt{ab}$.
\end{proof}

We also define and state some simple properties of matrix inner products:

\begin{definition}
    For matrices $\vec A,\vec  B \in \CC^{d_1 \times d_2}$, define the inner product $$\langle\vec  A, \vec B \rangle = \sum_{i \in [d_1], j \in [d_2]} \overline{A_{i, j}} \cdot B_{i, j} = \Tr{\vec A^\dag \vec B}.$$

    We also define the Frobenius norm $\norm{\vec A}_F = \sqrt{\langle \vec A, \vec A \rangle}$.
\end{definition}

\begin{lemma}\label{lemma:matrixinnerproduct}
    For matrices $\vec A,\vec  B \in \CC^{d \times d}$, we have:
    \begin{itemize}
        \item (Cauchy-Schwarz) $|\langle \vec A,\vec  B \rangle| \leq \norm{\vec A}_F \cdot \norm{\vec B}_F$.
        \item (Well-known) If $\vec A$ is Hermitian and $\vec B$ is Hermitian PSD, then we have $|\langle \vec A, \vec B \rangle| \leq \norm{\vec A}_\infty \cdot \Tr{\vec B}$.
    \end{itemize}
\end{lemma}

\subsubsection{Blockwise Tensor Products}\label{sec:blockwisetensor}

This section is devoted to stating and proving~\Cref{thm:opnormmain}, which will serve as our central linear algebraic workhorse. We will make some comments about this theorem and its proof at the end of this section. We will then present some straightforward consequences of this theorem in~\Cref{sec:opnormmainconsequences}, which we will use directly when analyzing $t \mapsto t+1$ cloning games.

\begin{theorem}\label{thm:opnormmain}
    Let $R, C$ be positive integers. Let $r_{1}, r_{2}, \ldots, r_{R}, r'_1, r'_2, \ldots, r'_R, c_1, \ldots, c_C, c'_1, \ldots, c'_C$ be positive integers. For each $i \in [R], k \in [C]$, let $\vec A_{i, k} \in \CC^{r_i \times c_k}$ and $\vec B_{i, k} \in \CC^{r'_i \times c'_k}$ be matrices. Additionally, for each $i \in [R], k \in [C]$, let $\gamma_{i, k} \in \CC$ be a scalar of magnitude at most 1, i.e. $|\gamma_{i, k}| \leq 1$. 
    
    Define the following block matrices:

    \begin{align*}
        \vec A &:= \begin{bmatrix} \vec A_{1, 1} &\ldots & \vec A_{1, C} \\ \vdots & \ddots & \vdots \\ \vec A_{R, 1} &\ldots & \vec A_{R, C}\end{bmatrix} \in \CC^{(r_1 + \ldots + r_R) \times (c_1 + \ldots + c_C)}\\
        \vec B &:= \begin{bmatrix} \vec B_{1, 1} &\ldots & \vec B_{1, C} \\ \vdots & \ddots & \vdots \\ \vec B_{R, 1} &\ldots & \vec B_{R, C}\end{bmatrix} \in \CC^{(r'_1 + \ldots + r'_R) \times (c'_1 + \ldots + c'_C)} \\
        \vec M &:= \begin{bmatrix} \gamma_{1, 1} \vec A_{1, 1} \otimes \vec B_{1, 1} & \ldots & \gamma_{1, C} \vec A_{1, C} \otimes \vec B_{1, C} \\ \vdots & \ddots & \vdots \\ \gamma_{R, 1} \vec A_{R, 1} \otimes\vec  B_{R, 1} & \ldots & \gamma_{R, C} \vec A_{R, C} \otimes \vec B_{R, C}\end{bmatrix} \in \CC^{(r_1r'_1 + \ldots + r_R r'_R) \times (c_1c'_1 + \ldots + c_C c'_C)}.
    \end{align*}
    Suppose both of the following conditions hold:
    \begin{enumerate}
        \item $\norm{\vec A}_\infty \leq 1$.
        \item Each block column of $\vec B$ has operator norm $\leq 1$ i.e. for all $k \in [C]$, we have $$\norm{\begin{bmatrix} \vec B_{1, k} \\ \vdots \\ \vec B_{R, k} \end{bmatrix}}_\infty \leq 1.$$
    \end{enumerate}
    Then, it holds that $\norm{\vec M}_\infty \leq 1$.
\end{theorem}

\paragraph{High-level proof idea.} The main idea is as follows: it suffices to show that for any unit vectors $x, y$ of the right dimensions that $|x^\dag \vec M y| \leq 1$. As a function of $B$, $x^\dag\vec  M y$ is linear. We can hence express this as the inner product of $\vec B$ with some other matrix. It turns out that this matrix has a simple form; reformulating the problem in these terms will allow us to use the standard bounds stated in Lemma~\ref{lemma:matrixinnerproduct}.

\paragraph{Notation.} We begin by setting up some notation. If we have a sequence of matrices $\left\{\vec A_I: I \in \mathcal{I}\right\}$ indexed by $I$ with rows and columns indexed by $r \in \mathcal{R}$ and $c \in \mathcal{C}$, we use $(\vec A_I)_{r; c}$ to denote the entry in row $r$ and column $c$ of matrix $\vec A_I$.

Now let us define $\widetilde{\vec B}$ as follows, and let $\vec B'$ be its entrywise conjugate:
\begin{align*}
    \widetilde{\vec B} &:= \begin{bmatrix} \widetilde{\vec B_{1, 1}} &\ldots & \widetilde{\vec B_{1, C}} \\ \vdots & \ddots & \vdots \\ \widetilde{\vec B_{R, 1}} &\ldots & \widetilde{\vec B_{R, C}}\end{bmatrix} \\
    &= \begin{bmatrix} \gamma_{1, 1}\vec B_{1, 1} &\ldots & \gamma_{1, C}\vec B_{1, C} \\ \vdots & \ddots & \vdots \\ \gamma_{R, 1}\vec B_{R, 1} &\ldots & \gamma_{R, C} \vec B_{R, C}\end{bmatrix} \in \CC^{(r'_1 + \ldots + r'_R) \times (c'_1 + \ldots + c'_C)}.
\end{align*}
\noindent
As outlined earlier, it suffices to show for any unit vectors $x \in \CC^{r_1r'_1 + \ldots + r_Rr'_R}$ and $y \in \CC^{c_1c'_1 + \ldots + c_Cc'_C}$ that $|x^\dag \vec M y| \leq 1$. We index the entries of $x$ by an index $i \in [R]$ and then values $j \in [r_i]$ and $j' \in [r'_i]$. We similarly index the entries of $y$ by $(k, l, l')$. We can also apply this same indexing to the rows and columns of $\vec M$. Thus for $i \in [R]$ and $k \in [C]$, we can define $\vec M_{i, k} \in \CC^{r_ir'_i \times c_kc'_k}$ by $\vec M_{i, k} = \gamma_{i, k} \vec A_{i, k} \otimes \vec B_{i, k}$ (i.e. this is one block of $\vec M$).

For each $i \in [R]$, let $\vec X_i \in \CC^{r_i \times r'_i}$ be defined by $(\vec X_i)_{j; j'} = x_{(i, j, j')}$. Similarly define $\vec Y_k \in \CC^{c_k \times c'_k}$ for each $k \in [C]$. We also defined the following matrices:
\begin{align*}
    \vec X &:= \begin{bmatrix} \vec X_1 & 0 & \ldots & 0 \\ 0 &\vec  X_2 & \ldots & 0 \\ \vdots & \vdots & \ddots & \vdots \\ 0 & 0 & \ldots & \vec X_R \end{bmatrix} \in \CC^{(r_1 + \ldots + r_R) \times (r'_1 + \ldots + r'_R)}\\
    \vec Y &:= \begin{bmatrix} \vec Y_1 & 0 & \ldots & 0 \\ 0 & \vec Y_2 & \ldots & 0 \\ \vdots & \vdots & \ddots & \vdots \\ 0 & 0 & \ldots & \vec Y_C \end{bmatrix} \in \CC^{(c_1 + \ldots + c_C) \times (c'_1 + \ldots + c'_C)}.
\end{align*}
It is straightforward to see that $\norm{\vec X}_F = \norm{\vec Y}_F = 1$, since the nonzero entries in $\vec X$ are exactly the same as in $x$, and similarly for $Y$ and $y$.

Finally, over $\CC_{c'_1 + \ldots + c'_C}$, for each $k \in [C]$ define the projector $\Pi'_k$ to be onto the natural $c'_k$ coordinates (more precisely, all coordinates $z$ such that $c'_1 + \ldots + c'_{k-1} < z \leq c'_1 + \ldots + c'_k$. Note then that the block-diagonal structure of $Y$ implies that:
\begin{equation}\label{eqn:blockdiagonal}
   \vec  Y^\dag \vec Y = \sum_{k \in [C]} \Pi'_k \vec Y^\dag \vec Y \Pi'_k.
\end{equation}

\paragraph{Rewriting $x^\dag \vec  M \, y$ as a linear function of $\vec B'$.} We now show that $x^\dag \vec M \, y$ can be written as a linear function of $B'$. This is captured in the following lemma:

\begin{lemma}\label{lemma:opnormreindexing}
    We have $$x^\dag \, \vec M \, y = \Tr{\vec X^\dag\vec  A\vec  Y (\vec B')^\dag}.$$
\end{lemma}
\begin{proof}
We proceed as follows:
\begin{align*}
    x^\dag \vec M \, y &= \sum_{i \in [R], k \in [C]} \sum_{j \in [r_i], j' \in [r'_i], l \in [c_k], l' \in [c'_k]} \overline{x_{(i, j, j')}} \left(\vec M_{i, k}\right)_{(j, j'); (l, l')} y_{(k, l, l')} \\
    &= \sum_{i \in [R], k \in [C]} \sum_{j \in [r_i], j' \in [r'_i], l \in [c_k], l' \in [c'_k]} \overline{x_{(i, j, j')}} \left(\vec A_{i, k} \otimes \widetilde{\vec B_{i, k}}\right)_{(j, j'); (l, l')} y_{(k, l, l')} \\
    &= \sum_{i \in [R], k \in [C]} \sum_{j \in [r_i], j' \in [r'_i], l \in [c_k], l' \in [c'_k]} \overline{x_{(i, j, j')}} \left(\vec A_{i, k}\right)_{j; l} \left(\widetilde{\vec B}_{i, k}\right)_{j'; l'} y_{(k, l, l')} \\
    &= \sum_{i \in [R], k \in [C]} \sum_{j' \in [r'_i], l' \in [c'_k]} \left(\widetilde{\vec B}_{i, k}\right)_{j'; l'} \cdot \left(\sum_{j \in [r_i], l \in [c_k]} \overline{x_{(i, j, j')}} \left(\vec A_{i, k}\right)_{j; l} y_{(k, l, l')}\right) \\
    &= \sum_{i \in [R], k \in [C]} \sum_{j' \in [r'_i], l' \in [c'_k]} \left(\widetilde{\vec B}_{i, k}\right)_{j'; l'} \cdot \left(\sum_{j \in [r_i], l \in [c_k]} \overline{\left(\vec X_i\right)_{j; j'}} \left(\vec A_{i, k}\right)_{j; l} \left(\vec Y_k\right)_{l; l'}\right) \\
    &= \sum_{i \in [R], k \in [C]} \sum_{j' \in [r'_i], l' \in [c'_k]} \left(\widetilde{\vec B}_{i, k}\right)_{j'; l'} \cdot \left(\sum_{j \in [r_i], l \in [c_k]} \left(\vec X^\dag_i\right)_{j'; j} \left(\vec A_{i, k}\right)_{j; l} \left(\vec Y_k\right)_{l; l'}\right) \\
    &= \sum_{i \in [R], k \in [C]} \sum_{j' \in [r'_i], l' \in [c'_k]} \left(\widetilde{\vec B}_{i, k}\right)_{j'; l'} \cdot \left(\vec X_i^\dag \vec A_{i, k} \vec Y_k\right)_{j'; l'} \\
    &= \sum_{i \in [R], k \in [C]} \sum_{j' \in [r'_i], l' \in [c'_k]} \widetilde{\vec B}_{(i, j'); (k, l')} \cdot (\vec X^\dag \vec A \vec  Y)_{(i, j'); (k, l')} \\
    &= \langle \vec B', \vec X^\dag\vec  A \vec Y \rangle \, = \,\Tr{(\vec B')^\dag \vec X^\dag\vec  A \vec Y} \, = \,\Tr{\vec X^\dag \vec A\vec  Y (\vec B')^\dag}.
\end{align*}

\end{proof}

\paragraph{Bounding the operator norm of $B'\Pi'_k$.} Our second ingredient will be the following straightforward observation:

\begin{lemma}\label{lemma:opnormblockcol}
    For any $k \in [C]$, we have $\norm{\vec B' \Pi'_k}_\infty \leq 1$.
\end{lemma}
\begin{proof}
    After removing zero columns, $\vec B' \Pi'_k$ is just the following block matrix: $$\begin{bmatrix} \overline{\gamma_{1, k}} \, \overline{\vec B_{1, k}} \\ \vdots \\ \overline{\gamma_{R, k}}\, \overline{\vec B_{R, k}}\end{bmatrix}.$$
    This is the result of taking $\begin{bmatrix} \vec B_{1, k} \\ \vdots \\ \vec B_{R, k}\end{bmatrix}$, multiplying each row by a scalar of magnitude $\leq 1$, then conjugating all entries. The latter two operations do not increase operator norm, and we are assuming that the starting matrix has operator norm $\leq 1$. The conclusion follows.
\end{proof}

\paragraph{Completing the proof.} We are now ready to complete the proof of~\Cref{thm:opnormmain}; this will follow from the standard inequalities stated in Lemma~\ref{lemma:matrixinnerproduct}, in addition to using the block-diagonal structure of $Y$ (as in Equation~\eqref{eqn:blockdiagonal}).

\begin{proof}[Proof of~\Cref{thm:opnormmain}]
    Starting from Lemma~\ref{lemma:opnormreindexing}, we have:
    \begin{align*}
        |x^\dag \vec  M \, y| &= \left| \Tr{\vec X^\dag\vec  A\vec  Y (\vec B')^\dag}\right| \\
        &\leq \norm{\vec X}_F \cdot \norm{\vec A\vec Y(\vec B')^\dag}_F \text{ (Lemma~\ref{lemma:matrixinnerproduct})} \\
        &= \norm{\vec A\vec Y(\vec B')^\dag}_F.
    \end{align*}
    Continuing from here, we have:
    \begin{align*}
        \norm{\vec A \vec Y(\vec B')^\dag}_F^2 &= \Tr{\vec B'\vec  Y^\dag\vec  A^\dag \vec A \vec Y (\vec B')^\dag} \\
        &= \Tr{\vec A^\dag \vec A \vec Y (\vec B')^\dag \vec B' \vec Y^\dag} \\
        &\leq \Tr{\vec Y (\vec B')^\dag \vec B' \vec Y^\dag} \text{ (Lemma~\ref{lemma:matrixinnerproduct}; $\norm{\vec A}_\infty \leq 1$)} \\
        &= \Tr{\vec Y^\dag \vec Y (\vec B')^\dag \vec B'} \\
        &= \sum_{k \in [C]} \Tr{\Pi'_k \vec Y^\dag \vec Y \Pi'_k (\vec B')^\dag \vec B'} \text{ (Equation~\eqref{eqn:blockdiagonal})} \\
        &= \sum_{k \in [C]} \Tr{\left(\Pi'_k \vec Y^\dag\vec  Y \Pi'_k\right) \left(\Pi'_k (\vec B')^\dag \vec  B' \Pi'_k\right)} \\
        &\leq \sum_{k \in [C]} \Tr{\Pi'_k \vec Y^\dag \vec Y \Pi'_k} \cdot \norm{\Pi'_k (\vec B')^\dag \vec B' \Pi'_k}_\infty \text{ (Lemma~\ref{lemma:matrixinnerproduct})} \\
        &\leq \sum_{k \in [C]} \Tr{\Pi'_k \vec Y^\dag \vec Y \Pi'_k} \text{ (Lemma~\ref{lemma:opnormblockcol})} \\
        &= \sum_{k \in [C]} \Tr{\Pi'_k \vec Y^\dag \vec Y} \\
        &= \Tr{\vec Y^\dag \vec Y} \\
        &= \norm{\vec Y}_F^2 \\
        &= 1,
    \end{align*}
    thus completing the proof of the theorem.
\end{proof}

\paragraph{Discussion.} Given the simplicity of the statement of~\Cref{thm:opnormmain}, one might wonder why our proof is so involved. Here, we present some justification that this theorem is actually quite strong, and discuss some obstacles to more intuitive proof strategies. First, we note that our theorem captures some simple special cases:
\begin{itemize}
    \item When $R = C = 1$, this is immediate from Lemma~\ref{lemma:opnormtensor}.
    \item When $\gamma_{i, k} = 1$ for all $i, k$ (or more generally $\gamma_{i, k}$ is constant) and $\norm{\vec B}_\infty \leq 1$, this can be shown by noting that the matrix $\vec M$ would be a submatrix of $\vec A \otimes \vec B$, and then appealing to Lemma~\ref{lemma:opnormtensor}. (We rigorously argue this fact as part of the proof of Lemma~\ref{lemma:manymatrixkhatrirao}.)

    However, this argument completely breaks down if $\gamma_{i, k}$ is allowed to vary between blocks.
    \item When $r_i = r'_i = c_k = c'_k = 1$ for all $i, k$, this boils down to bounding the operator norm of any complex $R \times C$ matrix with the entry in row $i$ and column $k$ having magnitude equal to $\left|A_{i; k} B_{i; k}\right|$ (noting that in this setting $A_{i; k}, B_{i; k}$ are scalars). This is not straightforward but still easier to handle; one can argue by Cauchy-Schwarz that the rows and columns of such a matrix must have $\ell_1$ norm $\leq 1$, and it is well-known that such a matrix must have operator norm $\leq 1$ (see Lemmas~\ref{lemma:hadamardunitarygeneral} and~\ref{lemma:rowcolsumopnorm} for details).

    This argument also breaks down as soon as the block matrices $A_{i, k}, B_{i, k}$ are not just scalars; the $\ell_1$ norms of the rows and columns of $M$ will end up growing polynomially in \\$\max(r_1, \ldots, r_R, r'_1, \ldots, r'_R, c_1, \ldots, c_C, c'_1, \ldots, c'_C)$ in the worst case. (Jumping ahead, in the setting of oracular cloning games, this would yield a bound that degrades exponentially in the number of ancilla qubits that each adversary is allowed to use, which is of course undesirable.)
\end{itemize}

One could imagine ``interpolating'' between these two techniques by considering the operator norm of each block of $\vec M$ individually, and using this to obtain a bound on $\norm{\vec M}_\infty$. Perhaps surprisingly, this is also \emph{provably} insufficient. Suppose $R = C = n$ for some $n$, and $r_i = r'_i = c_k = c'_k = n$ for all $i, k$. Then let us take $\vec A, \vec B$ to be $n^2 \times n^2$ permutation matrices with exactly one 1 in each block. Now we will have $\norm{\vec A_{i, k}}_\infty = \norm{\vec B_{i, k}}_\infty = 1 \Rightarrow \norm{\vec A_{i, k} \otimes \vec B_{i, k}}_\infty = 1$ for all $i, k$. However, there exist $n \times n$ block matrices (i.e. containing $n^2$ blocks in total) with each block having operator norm 1, but where the overall matrix has operator norm growing with $n$; one such example is the $n \times n$ all 1's matrix (appropriately padded with zero rows and columns to obtain the right dimensions). This counterexample implies that just considering the operator norm of each block of $\vec M$ is too lossy.

Thus~\Cref{thm:opnormmain} is quite strong and there are natural barriers to proof strategies that might feel more simple and intuitive. The proof we have presented is the simplest one that we are aware of.

\subsubsection{Consequences}\label{sec:opnormmainconsequences}

We now state some corollaries of~\Cref{thm:opnormmain} that we will later directly apply when bounding the operator norms relevant to cloning games in~\Cref{sec:opnormbound}.

\begin{lemma}\label{lemma:colwisetensor}
    Let $\vec A_1, \ldots, \vec A_k$ be block matrices of $d$ columns. More formally, for each $i \in [k]$ set $$\vec A_i = \begin{bmatrix} \vec A_{i, 1} & \ldots & \vec A_{i, d} \end{bmatrix},$$for some block matrices $\vec A_{i, 1}, \ldots, \vec A_{i, d}$ that have the same number of rows but not necessarily the same number of columns. (Note that we do not require $\vec A_{i, 1}$ and $\vec A_{j, 1}$ to have the same number of rows when $i \neq j$.) Assume the following preconditions:
    \begin{enumerate}
        \item For all $i \in [k]$ and $j \in [d]$, we have $\norm{\vec A_{i, j}}_\infty \leq 1$; and
        \item There exists some $i \in [k]$ such that $\norm{\vec A_i}_\infty \leq 1$.
    \end{enumerate}
    Let $$\vec A = \begin{bmatrix} \bigotimes_{i = 1}^k \vec A_{i, 1} & \bigotimes_{i = 1}^k \vec  A_{i, 2} & \ldots & \bigotimes_{i = 1}^k\vec  A_{i, d} \end{bmatrix}$$be defined as a ``block column-wise tensor product'' of $\vec A_1, \vec A_2, \ldots, \vec A_d$. Then $\norm{\vec A}_\infty \leq 1$.

\end{lemma}
\begin{proof}
    Firstly, if $k = 1$ then we will have $\vec A = \vec A_1$ so the conclusion will follow from the second precondition. From now on, assume that $k \geq 2$. Also, by symmetry, let us assume without loss of generality that $\norm{\vec A_1}_\infty \leq 1$.

    Now define the matrix $\vec M$ as follows:$$\vec M = \begin{bmatrix} \bigotimes_{i = 2}^k \vec A_{i, 1} & \bigotimes_{i = 2}^k \vec  A_{i, 2} & \ldots & \bigotimes_{i = 2}^k\vec  A_{i, d} \end{bmatrix}.$$
    Notice that, for each $j \in [d]$, the $j$th block column of $\vec M$ has operator norm equal to $\prod_{i = 1}^k \norm{\vec A_{i, j}}_\infty \leq 1$. Since we also have $\norm{\vec A_1}_\infty \leq 1$, we can apply~\Cref{thm:opnormmain} to $\vec A_1$ and $\vec M$ (with $R = 1$, $C = d$, and $\gamma_{i, j} = 1$ for all $i, j$) to immediately obtain that $\norm{A}_\infty \leq 1$, as desired.
\end{proof}

\begin{lemma}\label{lemma:manymatrixkhatrirao}
    Let $R, C$ be positive integers. Let $r_1, \ldots, r_R, c_1, \ldots, c_C$ be positive integers. Fix some integer $d \geq 2$. For each $t \in [d], i \in [R], k \in [C]$, let $\vec A_{t, i, k} \in \CC^{r_i \times c_k}$ be a matrix. Additionally, for each $i \in [R], k \in [C]$, let $\gamma_{i, k} \in \CC$ be a scalar of magnitude at most 1 i.e. $|\gamma_{i, k}| \leq 1$. Then define the following block matrices:
    \begin{align*}
        \vec A_t &:= \begin{bmatrix} \vec A_{t, 1, 1} &\ldots & \vec A_{t, 1, C} \\ \vdots & \ddots & \vdots \\ \vec A_{t, R, 1} &\ldots & \vec A_{t, R, C}\end{bmatrix} \in \CC^{(r_1 + \ldots + r_R) \times (c_1 + \ldots + c_C)}\text{, for each $t \in [d]$} \\
        \vec M &:= \begin{bmatrix} \gamma_{1, 1} \bigotimes_{t = 1}^d \vec A_{t, 1, 1} & \ldots & \gamma_{1, C} \bigotimes_{t = 1}^d \vec A_{t, 1, C} \\ \vdots & \ddots & \vdots \\ \gamma_{R, 1} \bigotimes_{t = 1}^d \vec A_{t, R, 1} & \ldots & \gamma_{R, C} \bigotimes_{t = 1}^d \vec A_{t, R, C}\end{bmatrix} \in \CC^{(r_1^d + \ldots + r_R^d) \times (c_1^d + \ldots + c_C^d)}.
    \end{align*}
    Suppose that for all $t \in [d]$, we have $\norm{\vec A_t}_\infty \leq 1$. Then $\norm{\vec M}_\infty \leq 1$. (Note that the $d = 2$ case is immediate from~\Cref{thm:opnormmain}.)
\end{lemma}
\begin{proof}
    Define the matrix $$\vec B = \begin{bmatrix} \bigotimes_{t = 2}^d \vec A_{t, 1, 1} & \ldots & \bigotimes_{t = 2}^d \vec A_{t, 1, C} \\ \vdots & \ddots & \vdots \\ \bigotimes_{t = 2}^d \vec A_{t, R, 1} & \ldots & \bigotimes_{t = 2}^d \vec A_{t, R, C}\end{bmatrix} \in \CC^{(r_1^{d-1} + \ldots + r_R^{d-1}) \times (c_1^{d-1} + \ldots + c_C^{d-1})}.$$
    We claim that $\vec B$ is a submatrix of $\vec A_2 \otimes \ldots \otimes \vec A_d$. This is intuitive, but nevertheless we justify this rigorously before completing the proof. To this end, let us index each row of each $A_t$ by an index $i \in [R]$ together with an index $\alpha \in [i]$. Similarly, we can index each column by an index $k \in [C]$ together with $\beta \in [k]$. We can hence index rows of $\vec A_2 \otimes \ldots \otimes \vec A_d$ by $(i_2, \ldots, i_t, \alpha_2, \ldots, \alpha_d)$ and similarly the columns by $(k_2, \ldots, k_d, \beta_2, \ldots, \beta_d)$; so that:
    $$\left(\vec A_2 \otimes \ldots \otimes \vec A_d\right)_{(i_2, \ldots, i_d, \alpha_2, \ldots, \alpha_d); (k_2, \ldots, k_d, \beta_2, \ldots, \beta_d)} = \prod_{t = 2}^d \left(\vec A_{t, i_t, k_t}\right)_{\alpha_t; \beta_t}.$$
    On the other hand, we can index the rows of $\vec B$ by one index $i \in [R]$ and indices $\alpha_2, \ldots, \alpha_d \in [i]$, and similarly the columns by $k, \beta_2, \ldots, \beta_d$, so that: $$B_{(i, \alpha_2, \ldots, \alpha_d); (k, \beta_2, \ldots, \beta_d)} = \prod_{t = 2}^d \left(\vec A_{t, i, k}\right)_{\alpha_t; \beta_t}.$$
    It is now clear that $\vec B$ can be obtained by restricting $\vec A_2 \otimes \ldots \otimes \vec A_d$ to rows where $i_2 = \ldots = i_d$ and columns where $k_2 = \ldots = k_d$. This establishes our claim.

    We now complete the proof as follows. By Lemma~\ref{lemma:opnormsubmatrix}, our claim implies that $$\norm{\vec B}_\infty \leq \prod_{t = 2}^d \norm{\vec A_t}_\infty \leq 1.$$ The conclusion now follows by applying~\Cref{thm:opnormmain} to $\vec A_1$ and $\vec B$ for the choices of $R, C$, and scalars $\gamma_{i, k}$. This proves the claim.
\end{proof}

%% file: monogamy-and-cloning-games-defns.tex
\section{Monogamy of Entanglement and Oracular Cloning Games}\label{sec:monogamydefs}

In this section, we formally define monogamy of entanglement games, as well as the closely related notion of (oracular) cloning games, which is the central object of study in this work.

\subsection{Monogamy of Entanglement Games}\label{sec:MOE}

A monogamy of entanglement game~\cite{Tomamichel_2013} is an interactive game which is played by three players: a trusted referee called Alice, and two colluding and adversarial parties Bob and Charlie. 

\begin{figure}
    \centering
 \begin{protocol}[Monogamy of Entanglement Game] \label{prot:MOE}\ \\
A monogamy of entanglement game $\mathsf{G} = (\algo H_{\reg{A}},\Theta, \algo X, \{\vec{A}_x^{\theta}\}_{\theta \in \Theta, x \in \algo X})$ for a quantum strategy $\mathsf{S} = (\algo H_{\reg{B}},\algo H_{\reg{C}}, \rho_{\reg{ABC}},\{\vec{B}_x^{\theta}\}_{\theta \in \Theta, x \in \algo X},\{\vec{C}_x^{\theta}\}_{\theta \in \Theta, x \in \algo X})$ is the following game between a trusted referee (called Alice) and two collaborating players (called Bob and Charlie):
\begin{enumerate}
  \item (\textbf{Setup phase}) Bob and Charlie prepare a tripartite quantum state $\rho \in \algo D(\algo H_{\reg{A}} \otimes \algo H_{\reg{B}} \otimes \algo H_{\reg{C}})$. They send register $\reg A$ to Alice, and hold onto registers $\reg B$ and $\reg C$, respectively. Afterwards, they are no longer allowed to communicate for the remainder of the game.

  \item (\textbf{Question phase}) Alice first samples a uniformly random question $\theta \sim \Theta$, and then applies the corresponding measurement $\{\vec{A}_x^{\theta}\}_{x \in \algo X}$ to her register $\reg A$. Afterwards, Alice announces the question $\theta$ to both Bob and Charlie.

\item (\textbf{Answer phase}) Bob and Charlie independently output a guess for Alice's outcome by applying the measurements $\{\vec{B}_x^{\theta}\}_{x \in\algo X}$ and $\{\vec{C}_x^{\theta}\}_{x \in \algo X}$ to their registers $\reg B$ and $\reg C$, respectively.

\item (\textbf{Outcome phase}) Bob and Charlie win if they both guess Alice's outcome correctly.
\end{enumerate}
\end{protocol}   
    \caption{A monogamy of entanglement game.}
    \label{fig:MOE-game}
\end{figure}

\begin{definition}[Monogamy of Entanglement Game] A monogamy of entanglement (MOE) game is specified by a tuple $\mathsf{G} = (\algo H_{\reg{A}},\Theta, \algo X, \{\vec{A}_x^{\theta}\}_{\theta \in \Theta, x \in \algo X})$ which consists of the following elements:
\begin{itemize}
    \item A finite dimensional Hilbert space $\algo H_{\reg{A}}$ corresponding to a register $\reg{A}$ that Alice holds;

    \item A finite set $\Theta$ corresponding to the set of possible questions;

    \item A finite set $\algo X$ corresponding to the set of all possible answers;

    \item A set of positive operator-valued measurements $\big\{\vec{A}_x^{\theta}\big\}_{\theta \in \Theta, x \in \algo X}$ to be performed on Alice's system.
\end{itemize}
\end{definition}

\begin{definition}[Quantum Strategy]
A quantum strategy $\mathsf{S} = (\algo H_{\reg{B}},\algo H_{\reg{C}}, \rho_{\reg{ABC}},\{\vec{B}_x^{\theta}\}_{\theta \in \Theta, x \in \algo X},\{\vec{C}_x^{\theta}\}_{\theta \in \Theta, x \in \algo X})$ for a monogamy of entanglement game $\mathsf{G} = (\algo H_{\reg{A}},\Theta,\algo X, \{\vec{A}_x^{\theta}\}_{\theta \in \Theta, x \in \algo X})$ consists of
\begin{itemize}
    \item A finite dimensional Hilbert space $\algo H_{\reg{B}}$ corresponding to a register $\reg{B}$ that Bob holds;

    \item A finite dimensional Hilbert space $\algo H_{\reg{C}}$ corresponding to a register $\reg{C}$ that Charlie holds;

    \item A tripartite quantum state $\rho \in \algo D(\algo H_{\reg{A}} \otimes \algo H_{\reg{B}} \otimes \algo H_{\reg{C}})$;

    \item A set of positive operator-valued measurements $\big\{\vec{B}_x^{\theta}\big\}_{\theta \in \Theta, x \in \algo X}$ to be performed on Bob's system.

    \item A set of positive operator-valued measurements $\big\{\vec{C}_x^{\theta}\big\}_{\theta \in \Theta, x \in \algo X}$ to be performed on Charlie's system.
\end{itemize}
\end{definition}

\begin{definition}[Value of a Monogamy Game]
Let $\mathsf{G} = (\algo H_{\reg{A}},\Theta, \algo X, \{\vec{A}_x^{\theta}\}_{\theta \in \Theta, x \in \Sigma})$ be monogamy game. Then, the winning probability of a quantum strategy $\mathsf{S} = (\algo H_{\reg{B}},\algo H_{\reg{C}}, \rho_{\reg{ABC}},\{\vec{B}_x^{\theta}\}_{\theta \in \Theta, x \in \algo X},\{\vec{C}_x^{\theta}\}_{\theta \in \Theta, x \in \algo X})$ for the particular monogamy game $\mathsf{G}$ is defined by the quantity
$$
\omega_{\mathsf{S}}(\mathsf{G}) := 
\underset{\theta \sim \Theta}{\mathbb{E}} \sum_{x \in \algo X} \mathrm{Tr}\left[ (\vec{A}_x^{\theta} \otimes \vec{B}_x^{\theta} \otimes \vec{C}_x^{\theta}) \rho_{\reg{ABC}}  \right].
$$
Moreover, we define the value of the monogamy game $\mathsf{G}$ as the optimal winning probability
$$
\omega(\mathsf{G}) := \underset{\mathsf{S} = (\algo H_{\reg{B}},\algo H_{\reg{C}}, \rho_{\reg{ABC}},\{\vec{B}_x^{\theta}\},\{\vec{C}_x^{\theta}\})}{\sup} \,\,\omega_{\mathsf{S}}(\mathsf{G}).
$$
\end{definition}

\begin{remark}
    As noted in~\cite{Tomamichel_2013}, a standard purification argument and Neumark's dilation theorem show that we can assume without loss of generality that all POVMs are projective. We will assume this going forward.
\end{remark}

\begin{definition}[Parallel-Repeated Monogamy Game]\label{def:parallel-rep}
    Let $r \in \NN$ be a parameter. For any monogamy game $\mathsf{G} = (\mathcal{H}_{\reg A}, \Theta, \mathcal{X}, \left\{\vec{A}_x^\theta\right\}_{\theta \in \Theta, x \in \mathcal{X}})$, we define the \emph{$r$-fold parallel-repeated monogamy game} $\mathsf{G}^{\times r}$ as follows:
    \begin{itemize}
        \item The Hilbert space for Alice's register will be $\mathcal{H}_{\reg A}^{\otimes r}$.
        \item The set of questions will now be $\Theta^r$.
        \item The set of answers will be $\mathcal{X}^r$.
        \item For any $(x_1, \ldots, x_r) \in \mathcal{X}^r$ and $(\theta_1, \ldots, \theta_r) \in \Theta^r$, we define Alice's measurement to be $$\vec{A}_{(x_1, \ldots, x_r)}^{(\theta_1, \ldots, \theta_r)} = \bigotimes_{i = 1}^r \vec{A}_{x_i}^{\theta_i}.$$
    \end{itemize}
    Informally, Alice will carry out $r$ parallel measurements and Bob and Charlie succeed if they successfully guess the outcomes of all $r$ measurements.
\end{definition}

\begin{example}[BB84 Monogamy Game]
    As a simple example, the following monogamy game is known as the ``BB84 monogamy game'':
    \begin{itemize}
        \item The Hilbert space $\mathcal{H}_{\reg A}$ is $\CC^2$.
        \item The sets of questions $\Theta$ and answers $\mathcal{X}$ are both $\bit$.
        \item For any $x, \theta \in \bit$, we have $$\vec{A}_x^\theta = \mathsf{H}^\theta \proj{x}\mathsf{H}^\theta.$$
    \end{itemize}
\end{example}

\begin{remark}\label{remark:moetrivialstrategy}
    We note that any MOE game admits a trivial strategy with success probabiliy $1/|\mathcal{X}|$. Bob and Charlie could set up the state $\rho_{\reg{ABC}}$ so that Bob and Alice are maximally entangled. This would enable Bob to always guess $x$ correctly, and now Charlie can guess randomly. (He cannot do better as in this case he must be completely decoupled from Alice and Bob.)
\end{remark}

\subsection{Oracular Cloning Games}\label{sec:oracularcloning}

The monogamy of entanglement games which we encounter in physics and cryptography often deal with some restrictions on the types of strategies that can be employed.
Motivated by this, in this section, we introduce the notion of a $t \mapsto t+1$ \emph{cloning game}. In the case when $t=1$, this notion turns out to be a special case of a monogamy of entanglement game in \Cref{sec:MOE}, with the following additional restrictions: 
\begin{itemize}
        \item The tripartite state $\rho \in \algo D(\algo H_{\reg{A}} \otimes \algo H_{\reg{B}} \otimes \algo H_{\reg{C}})$ which is shared between Alice, Bob and Charlie is the result of applying a cloning channel $\Phi_{\reg{A' \rightarrow BC}}$ to one half of an EPR pair, i.e.,
    \begin{equation*}
    \rho_{\reg{ABC}} = (\id_{\reg A} \otimes \Phi_{\reg{A' \rightarrow BC}})(\proj{\mathsf{EPR}}_{\reg{AA'}}).
    \end{equation*}
    In other words, $\rho_{\reg{ABC}}$ is the normalized Choi state of some channel $\Phi_{\reg{A' \rightarrow BC}}$.
    
    \item Alice's measurement $\big\{\vec{A}_x^{\theta}\big\}_{\theta \in \Theta, x \in \algo X}$ on register $\reg A$ is a projective measurement of the form
    $$
    \vec{A}_x^{\theta} = \bar{U}_\theta \proj{x} 
    \bar{U}_\theta^\dag \,,
    $$
    for some family of unitary operators $\{U_\theta\}_{\theta \in\Theta}$ acting on $\algo H_{\reg A}$.

    \item (If we are in the oracular setting) Bob and Charlie's measurements can only depend on oracle queries to $U_\theta$ and $U_\theta^\dag$, rather than directly on $\theta$.
    \end{itemize}
This equivalence was first observed and used by Broadbent and Lord~\cite{broadbent_et_al:LIPIcs.TQC.2020.4} for analyzing BB84 cloning games; for completeness, we give a proof of the general case (which proceeds along almost identical lines) in Lemma~\ref{lem:equiv}. For $t \geq 2$, the notion of a $t \mapsto t+1$ cloning game includes $t+1$ colluding parties and thus starts to become incomparable to a monogamy of entanglement game in \Cref{sec:MOE}. However, a similar argument still allows us to establish an equivalence between $t \mapsto t+1$ cloning games and another game that resembles a monogamy game, and we will use this equivalence in our analysis. This equivalence is demonstrated in Lemma~\ref{lemma:CImultiplayer}.

Let us now give a formal definition of a $t \mapsto t+1$ cloning game.

\begin{definition}[(Oracular) Cloning Game]\label{def:OCG} Let $t \in \N$ be an integer. A $t \mapsto t+1$ (oracular) cloning game $(\mathsf{(O)CG})$ is a tuple $\mathsf{G}_{t \mapsto t+1} = (t,\algo H_{\reg{A^t}},\Theta, \algo X, \{U_{\theta}\}_{\theta \in \Theta})$ which consists of the following elements:
\begin{itemize}
    \item A finite dimensional Hilbert space $\algo H_{\reg{A^t}}$ consisting of registers $\reg{A^t} := \reg{A_1} \cdots \reg{A_t}$ given to the cloner;

    \item A finite set $\Theta$ corresponding to the set of possible questions;

    \item A finite set $\algo X$ corresponding to the set of all possible answers;

    \item A finite ensemble of unitary operators $\{U_\theta\}_{\theta \in \Theta}$ acting on the $\reg{A}$ systems.
\end{itemize}
\end{definition}

\begin{definition}[Quantum Strategy for Cloning Games]\label{def:quantumstrategycloninggame}

Let $t\in \N$ and let $\mathsf{G}_{t \mapsto t+1} = (t,\algo H_{\reg{A^t}},\Theta, \algo X, \{U_{\theta}\}_{\theta \in \Theta})$ be a cloning game.
A quantum strategy $\mathsf{S} = (\algo H_{\reg{B^{t+1}}}, \Phi_{\reg{A^t \rightarrow B^{t+1}}},\{\vec{P}_{1,x}^{\theta}\}_{\theta \in \Theta, x \in \algo X},\,\dots,\{\vec{P}_{t+1,x}^{\theta}\}_{\theta \in \Theta, x \in \algo X})$ for the game $\mathsf{G}_{t \mapsto t+1}$ is characterized by the following elements:
\begin{itemize}
    \item A finite dimensional Hilbert space $\algo H_{\reg{B^{t+1}}}$ consisting of registers $\reg{B^{t+1}} := \reg{B_1} \cdots \reg{B_{t+1}}$ which are held by the $k+1$ many players in the game;

    \item A completely positive and trace-preserving channel $\Phi_{\reg{A^t \rightarrow B^{t+1}}}$ performed by the cloner;

    \item A sequence of measurements $\{\vec{P}_{1,x}^{\theta}\}_{\theta \in \Theta, x \in \algo X},\,\dots,\{\vec{P}_{t+1,x}^{\theta}\}_{\theta \in \Theta, x \in \algo X}$ which are to be performed by the $t+1$ players on the registers $\reg{B_1}, \cdots, \reg{B_{t+1}}$, respectively.
\end{itemize}
\end{definition}

\begin{definition}[Quantum Strategy for \textbf{Oracular} Cloning Games]\label{def:quantum-strategy-OCG}

A quantum strategy for an \textbf{oracular} cloning game is the same as a quantum strategy for a cloning game, with the following crucial restriction: the measurements by the $t+1$ players will now be \emph{oracle-aided}. We denote these as $$\left\{\vec{P}_{1,x}^{U_\theta,U_{\theta}^\dag}\right\}_{\theta \in \Theta, x \in \algo X},\,\dots,\left\{\vec{P}_{t+1,x}^{U_\theta,U_{\theta}^\dag}\right\}_{\theta \in \Theta, x \in \algo X}.$$

Informally, an oracular cloning game is one where the players are only given oracle access to $U_\theta, U_\theta^\dag$, whereas in Definition~\ref{def:quantumstrategycloninggame} the players are given the question $\theta$ in the clear.
\end{definition}

\begin{definition}[\textbf{Restricted} Quantum Strategy for Oracular Cloning Games]\label{def:restricted-quantum-strategy-OCG}\label{def:veryrestrictedocg}
    Assume $\mathcal{X} = \bit^n$. Then a \textbf{restricted} quantum strategy for an oracular cloning game further restricts the players in the following way. For each $i \in [t+1]$, player $\algo{P}_i$ must output their guess $x \in \bit^n$ after applying some quantum algorithm that makes \textbf{at most one query} to either $U_\theta$ or $U_\theta^\dag$.

    We let $\algo{S}_{\mathrm{rest}}$ denote the collection of restricted quantum strategies $\mathsf{S}$ for the oracular cloning game $\mathsf{G}$.
\end{definition}
\noindent
We first observe that we can impose some structure on the $t+1$ players' strategies without loss of generality; this will make our analysis easier:

\begin{lemma}\label{lemma:ocgwlog}
    Without loss of generality, a restricted quantum strategy for $\mathsf{G}$ may be taken to the have the following much more restricted structure: Each player $\algo{P}_i$ will hold a register $\reg{B_i}$ that splits into the following registers:
    \begin{itemize}
        \item A query register $\reg{C_i}$ of $n$ qubits;
        \item An ancilla register $\reg{D_i}$ of $a$ qubits (we allow $a$ to be arbitrary, but assume WLOG that it is the same for all the players); and
        \item A classical control bit $b_i$ from the cloning channel $\Phi$, which we store in a single-qubit register $\reg{E_i}$ for formality's sake.
    \end{itemize}
    The player $\algo P_i$ will then proceed as follows:
    \begin{enumerate}
        \item They first make \textbf{exactly one query} to either $U_\theta$ or $U_\theta^\dag$, which will be applied to the $\reg{C_i}$ register. Which of these unitaries they query will be controlled by $b_i$.
        \item They can then apply a unitary $Q_i$ of their choice to their entire system $\reg{B_i}$. (We assume without loss of generality that the same unitary $Q_i$ is applied regardless of the value of the control bit $b_i$; the cloner could simply include a copy of the control bit in register $\reg D_i$ as well, which $Q_i$ acts on.)
        \item They now measure the $n$ qubits in the $\reg{C_i}$ register to obtain a string $x \in \bit^n$.
        \item They output $x$. 
    \end{enumerate}
    Formally: for every $i \in [t+1]$ and $x \in \bit^n$, player $\algo{P}_i$'s projector has the form:
    \begin{align*}
        \vec{P}_{i, x}^{U_\theta, U_\theta^\dag} = &\left[(U_\theta^\dag \otimes \id_{\reg D_i})Q_i^\dag(\proj{x} \otimes \id_{D_i})Q_i(U_\theta \otimes \id_{\reg{D_i}})\right] \otimes \proj{0}_{\reg{E_i}} \\
        + &\left[(U_\theta \otimes \id_{\reg D_i})Q_i^\dag(\proj{x} \otimes \id_{D_i})Q_i(U_\theta^\dag \otimes \id_{\reg{D_i}})\right] \otimes \proj{1}_{\reg{E_i}}.
    \end{align*}
\end{lemma}
\begin{proof}
    Any preprocessing that player $\algo{P}_i$ might carry out before their query can be absorbed into the cloning channel $\Phi$, including the decision about which of $U_\theta, U_\theta^\dag$ to query, which we represent in the control bit $b_i$. (If the player does not want to query either, we can just treat $\reg{C_i}$ as dummy qubits and make a query there.)

    The conclusion now follows from the Stinespring and Neumark dilation theorems~\cite{NielsenChuang11}.
\end{proof}

\begin{remark}\label{remark:veryrestricted}
    Some comments are in order about Definition~\ref{def:veryrestrictedocg}:
    \begin{itemize}
        \item The cloner $\Phi$ remains entirely unrestricted; they can apply an arbitrary quantum channel to $(U_\theta\ket{x})^{\otimes t}$.
        \item While quite restrictive, this model is still sufficiently expressive and captures a standard approximate no-cloning bound as a special case: once  the cloning map $\Phi$ has been applied, each player $\algo{P}_i$ immediately makes a query to $U_\theta^\dag$ and measures in the computational basis. In this case, it is easy to check that the winning probability corresponds precisely to the average cloning fidelity for $t \mapsto t+1$.

    \end{itemize}
\end{remark}

\begin{definition}[Value of a (Oracular) Cloning Game]\label{def:value-OCG}
Let $t \in \N$. The winning probability of a quantum strategy $\mathsf{S} = (\algo H_{\reg{B^{t+1}}}, \Phi_{\reg{A^t \rightarrow B^{t+1}}},\{\vec{P}_{1,x}^{U_\theta,U_{\theta}^\dag}\}_{\theta \in \Theta, x \in \algo X},\,\dots,\{\vec{P}_{t+1,x}^{U_\theta,U_{\theta}^\dag}\}_{\theta \in \Theta, x \in \algo X})$ for a particular $t\mapsto t+1$ oracular cloning game $\mathsf{G}_{t \mapsto t+1} = (t,\algo H_{\reg{A^t}},\Theta, \algo X, \{U_{\theta}\}_{\theta \in \Theta})$ is defined by the quantity
$$
\omega_{\mathsf{S}}(\mathsf{G}_{t \mapsto t+1}) := 
\underset{\theta \sim \Theta}{\mathbb{E}}  \,\underset{x \sim \algo X}{\mathbb{E}}  \, \Tr{\left( \vec{P}_{1,x}^{U_\theta,U_{\theta}^\dag}\ot \dots \ot \vec{P}_{t+1,x}^{U_\theta,U_{\theta}^\dag}\right) \Phi_{\reg{A^t \rightarrow B^{t+1}}}\left( (U_\theta \proj{x}U_\theta^\dag)_{\reg{A^t}}^{\otimes t} \right)} .
$$
Moreover, we define the value of the oracular cloning game $\mathsf{G}$ as the optimal winning probability
$$
\omega(\mathsf{G}_{t \mapsto t+1}) := \underset{\mathsf{S} = (\algo H_{\reg{B^{t+1}}}, \Phi_{\reg{A^t \rightarrow B^{t+1}}},\{\vec{P}_{1,x}\},\,\dots,\{\vec{P}_{t+1,x}\})}{\sup} \,\,\omega_{\mathsf{S}}(\mathsf{G}_{t \mapsto t+1}).
$$
We analogously define the value of a cloning game $\mathsf{G}$, using the measurements $\left\{\vec{P}_{i, x}^\theta\right\}$ instead.
\end{definition}

\begin{remark}[Comparison with~\cite{10.1007/978-3-031-38554-4_3}]\label{remark:aklcomparison}
    Our definition of cloning games (in the single copy setting) is much more specific than that in~\cite[Definitions 2-6]{10.1007/978-3-031-38554-4_3}. Using their terminology, the secret key is $\theta$ and each of the steps has the following structure:
    \begin{itemize}
        \item The secret key $\mathsf{sk}$ will be the basis $\theta$.
        \item Token generation $\mathsf{GenT}(\mathsf{sk}, x)$ is restricted to deterministically output the pure state $U_\theta \ket{x}$.
        \item Challenge generation $\mathsf{GenC}(\mathsf{sk}, x)$ simply releases $\theta = \mathsf{sk}$, either in the clear or in the form of an oracle.
        \item For verification, we specialize to the \emph{cloning search} setting (see~\cite[Definition 4]{10.1007/978-3-031-38554-4_3}), where the players win if they all output the string $x$ as their answer.
    \end{itemize}
    Our reason for focusing on games of this form is because we are primarily concerned with the difficulty of cloning one very specific property of a quantum state, namely its measurement statistics with respect to the bases specified by $\Theta$ (or even more weakly, just the string $x$). 
\end{remark}

\noindent
We also make the straightforward observation that cloning games are closely related to unclonable encryption ($\sUE$) schemes (which we will formally define in~\Cref{sec:unclonableenc}):

\begin{lemma}\label{lemma:cloningtoue}
    Consider a cloning game with $t$ players, $\mathcal{X} = \bit^n$, and $\Theta = \bit^\secp$. Then all of the following hold:
    \begin{itemize}
        \item If the corresponding cloning game has value $\leq \epsilon$ with computationally unbounded (respectively, computationally bounded) $(\Phi, \algo P_1, \ldots, \algo P_{t+1})$, then there exists a $\sUE$ scheme satisfying statistical (respectively, computational) $t \mapsto t+1$ $\epsilon$-$\sUE$ security.
        \item If the corresponding \emph{oracular} cloning game has value $\leq \epsilon$ with computationally unbounded (respectively, computationally bounded) $(\Phi, \algo P_1, \ldots, \algo P_{t+1})$, then there exists a $\sUE$ scheme satisfying (statistical, respectively computational) $t \mapsto t+1$ $(\epsilon, \infty)$-$\sUE$ \emph{oracular} security.
        \item If in the corresponding oracular cloning game, any computationally unbounded (respectively, computationally bounded) \emph{restricted} strategy $(\Phi, \algo P_1, \ldots, \algo P_{t+1})$ has value $\leq \epsilon$, then there exists a $\sUE$ scheme satisfying (statistical, respectively computational) $t \mapsto t+1$ $(\epsilon, 1)$-$\sUE$ oracular security.
    \end{itemize}
\end{lemma}
\begin{proof}
    In all cases, the construction proceeds as follows: we will take $\Enc(\theta, x) = U_\theta\ket{x}$, and $\Dec$ will apply $U_\theta^\dag$ and measure in the standard basis. The conclusions are now straightforward to verify.
\end{proof}

\begin{figure}
    \centering
 \begin{protocol}[Oracular Cloning Game] \label{prot:OCG}\ \\
A $t \mapsto t+1$ oracular cloning game $\mathsf{G}_{t \mapsto t+1} = (k,\algo H_{\reg{A^t}},\Theta, \algo X, \{U_{\theta}\}_{\theta \in \Theta})$ for a quantum strategy of the form $\mathsf{S} = (\algo H_{\reg{B^{t+1}}}, \Phi_{\reg{A^t \rightarrow B^{t+1}}},\{\vec{P}_{1,x}^{U_\theta,U_{\theta}^\dag}\}_{\theta \in \Theta, x \in \algo X},\,\dots,\{\vec{P}_{t+1,x}^{U_\theta,U_{\theta}^\dag}\}_{\theta \in \Theta, x \in \algo X})$ is the following game between a trusted challenger, a cloner and $t+1$ many players:
\begin{enumerate}
  \item (\textbf{Setup phase}) The challenger samples a random $x \sim \algo X$ and a random $\theta \sim \Theta$, and sends the state $(U_\theta\ket{x})^{\otimes t}$ consisting of registers $\reg{A^t} := \reg{A_1} \cdots \reg{A_t}$ to the cloner.
  
  The cloner applies the channel $\Phi_{\reg{A^t \rightarrow B^{t+1}}}$ to $(U_\theta\ket{x})^{\otimes t}$ and sends the resulting registers $\reg{B^{t+1}}=\reg{B_1} \cdots \reg{B_{t+1}}$ to the $t+1$ many players, respectively. Afterwards, the players may no longer communicate with each other for the remainder of the game.

  \item (\textbf{Question phase}) Each of the players receives oracles for both $U_\theta$ and $U_\theta^\dag$.

\item (\textbf{Answer phase}) The players independently output a guess for the element $x$ by applying the measurements $\{\vec{P}_{1,x}^{U_\theta,U_{\theta}^\dag}\}_{x \in \algo X}, \, \dots,\{\vec{P}_{t+1,x}^{U_\theta,U_{\theta}^\dag}\}_{x \in \algo X}$ to their registers, respectively.

\item (\textbf{Outcome phase}) The players win if they all guess $x$ correctly.
\end{enumerate}
\end{protocol}   
    \caption{A $t\mapsto t+1$ oracular cloning game. A regular cloning game is defined analogously, except the measurements are now $\vec{P}_{i, x}^\theta$ and free to depend on $\theta$ in any way. Informally, in a standard cloning game $\theta$ is revealed to the $t+1$ players in the clear, while in the oracular cloning game the players are only given oracle access to $U_\theta$ and $U_\theta^\dag$.}
    \label{fig:OCG}
\end{figure}

%% file: binary-phase-technical.tex
\section{Types and Subtypes}\label{sec:typesandsubtypes}

In order to improve on the limitations of the~\cite{Tomamichel_2013} framework for bounding the value of monogamy games, we essentially restrict attention to oracular cloning games, and restrict each player to only make one query. This allows us to analyze this game using the language of \emph{binary phase twirls} (defined and analyzed in~\Cref{sec:phaseunitarydefinition}). To effectively capture the effect of binary phase twirls on an operator, we revisit the formalism of \emph{types} introduced by~\cite{DBLP:conf/tcc/AnanthGQY22} in~\Cref{sec:types} and extend this to \emph{subtypes} in~\Cref{sec:subtypes}. Later, in~\Cref{sec:binaryphaseconstruction}, we will leverage these tools to analyze our construction using binary phase states and prove cloning bounds of $O_t(2^{-n})$.

\subsection{Binary Types}\label{sec:types}

Let $N, M \in \N$ and $r \in \N$. For a vector $\vec x = (x_1,\dots,x_r) \in [N]^r$ and an ancilla input $y \in [M]$, we denote by $\Type(\vec x, y) \in [0:r]^N$ the so-called \emph{type vector} in which the $i$-th entry corresponds to the number of occurrences of $i \in [N]$ in $\vecx$. Note that the ancillary information $y$ is just representing some auxiliary input that we do not consider when evaluating $\Type$. We denote by $\BinType(\vec x, y) \in \bit^N$ the \emph{binary type vector} in which the $i$-th entry corresponds to the parity of the number of occurrences of $i \in [N]$ in $\vecx$. In other words, we let 
$$\BinType(\vecx, y) =  \Type(\vecx, y) \pmod{2}.$$

We note that our definition of $\Type$ and $\BinType$ is a natural extension of the standard definition in the literature (which does not consider auxiliary input); in particular, when $M = 0$ and $y$ is the empty string, our definitions and the standard definitions coincide.

\paragraph{$\BinType$ decomposition.}

When working with the vector space $\algo H = (\CC^N)^{\otimes r} \otimes \CC^M$, we use the following $\BinType$ decomposition into orthogonal subspaces $V_{\blambda}$ indexed by binary types $\blambda \in \bit^N$ such that
$$
(\CC^N)^{\otimes r} \otimes \CC^M \cong \bigoplus_{\blambda} V_{\blambda} \, ,
$$
where each subspace $V_{\blambda} \subseteq \algo H$ corresponds to vectors with a particular binary type $\blambda$, i.e.,
$$V_{\blambda} = \mathrm{span}_{\CC} \{\ket{v_1,\dots,v_r, w} \, : \,\BinType((v_1,\dots,v_r), w) = \blambda\}.$$

\subsection{Phase Twirling}\label{sec:phaseunitarydefinition} For a binary function $f: [N]\rightarrow \bit$, we let $\mathsf{U}_f$ define the phase unitary
    $$
    \mathsf{U}_f = \sum_{x \in [N]} (-1)^{f(x)} \proj{x}.
    $$
Using the $\BinType$ decomposition, we can show the following identity for the $r$-wise twirl with $\mathsf{U}_f$.
We note that the below lemma is also immediate from Zhandry's compressed oracle technique~\cite{DBLP:conf/crypto/Zhandry19}; although we do not present our results in these terms here, we outline this formulation in~\Cref{sec:techoverviewbinaryphase}.

\begin{lemma}\label{lem:phase-type-twirl} Let $O \in L(\algo H)$ be a linear operator acting on the vector space $\algo H = (\CC^N)^{\otimes r} \otimes \CC^M$. Then,
$$
\underset{f \sim \algo F_n}{\mathbb{E}}\left[\left(\sfU_f^{\otimes r} \otimes \id\right) O \left(\sfU_f^{\otimes r} \otimes \id\right)\right] = \sum_{\blambda \in \bit^N} \Pi_{\blambda} O \Pi_{\blambda} \, ,
$$
where $\Pi_{\blambda}$ projects onto $V_{\blambda} = \mathrm{span}_{\CC} \{\ket{x_1,\dots,x_r, v} \in (\CC^N)^{\otimes r}\otimes \CC^M\, : \,\BinType((x_1,\dots,x_r), v) = \blambda\}$. 
\end{lemma}
\begin{proof}
Expanding $O$ in the standard basis and using the linearity of expectation, we get
\begin{align*}
 &\underset{f \sim \algo F_n}{\mathbb{E}}\left[\left(\sfU_f^{\otimes r} \otimes \id\right) O \left(\sfU_f^{\otimes r} \otimes \id\right)\right]\\
 &= \sum_{\substack{\vec x,\vec y \in [N]^r\\v, w \in [M]}}  O_{(\vec x, v); (\vec y, w)} \underset{f \sim \algo F_n}{\mathbb{E}}\left[\sfU_f^{\otimes r} \ketbra{\vec x}{\vec y} \sfU_f^{\otimes r} \otimes \ketbra{v}{w}\right]\\
 &= \sum_{\substack{\vec x,\vec y \in [N]^r\\v, w \in [M]}}  O_{(\vec x, v); (\vec y, w)} \underset{f \sim \algo F_n}{\mathbb{E}}\left[(-1)^{f(x_1) + \hdots + f(x_r) + f(y_1) + \hdots +f(y_r)}\right] 
 \ketbra{\vec x}{\vec y} \otimes \ketbra{v}{w} \\
 &= \sum_{\substack{
 \vec x,\vec y \in [N]^r\\
 v, w \in [M]\\
 \BinType(\vec x, v) = \BinType(\vec y, w)
 }}  O_{(\vec x, v); (\vec y, w)}
 \ketbra{\vec x, v}{\vec y, w} \,\,= \sum_{\blambda \in \bit^N}  \Pi_{\blambda}O \Pi_{\blambda}.
\end{align*}
\end{proof}

\subsection{Binary Subtypes}\label{sec:subtypes}

\subsubsection{Definitions and Combinatorial Properties}\label{sec:subtypebasics}

While $\BinType$ is very simple to define, it comes with an ``entangled''\footnote{This comment is qualitative, and does not relate in any way to quantum entanglement.} combinatorial structure that is difficult to work with. As a simple example, consider the case where $r = 3, M = 0$, and the binary type $\blambda$ is $(1, 0, 0, \ldots, 0)$. There are a few different ways for a vector in $[N]^3$ to attain this $\BinType$: the vector could be of the form $(0, x, x)$ for any $x \in [N]$ or any permutation of this, and moreover these collections of vectors will overlap on $(0, 0, 0)$.

Instead of working with the $\BinType$ directly, it is more natural and convenient to address each of these different collections of vectors separately. Within each of these collections, there is now a very clean combinatorial structure that we will be able to exploit.

To formalize the above intuition, we will work with the notion of \emph{subtypes}. As in Section \ref{sec:types}, let $N, M, r$ be positive integer parameters:

\begin{definition}\label{defn:subtype}
    A \emph{subtype} of a given type $\blambda = (c_1, \ldots, c_N) \in \bit^{N}$ is a string $\bmu$ of length $r$. Each entry of $\bmu$ is either an integer $i \in [N]$ such that $\blambda_i = 1$, or a \emph{variable symbol} $x_i$ for some index $i$. We have the following constraints:
    \begin{itemize}
        \item For each $i \in [N]$ such that $\blambda_i = 1$, $i$ should appear an odd number of times in $\bmu$.
        \item For any $i$ such that $x_i$ appears at least once in $\bmu$, the first $i$ distinct variable symbols that appear in $\bmu$ are $x_1, x_2, \ldots, x_i$ in that order.
        \item Each variable symbol $x_i$ appears an even number of times in $\bmu$.
    \end{itemize}
\end{definition}

\begin{definition}\label{defn:queryrestriction}
    For a vector $(\vecx, y) \in [N]^{r} \times [M]$, define its \emph{query restriction} to be $\vecx \in [N]^r$. (Informally, the query restriction discards any auxiliary information.) 
\end{definition}

\begin{definition}\label{defn:subtypeinduced}
    We say a vector $(\vecx, y) \in [N]^{r} \times [M]$ matches a subtype $\bmu$ if there exist assignments of values in $[N]$ to the variable symbols in $\bmu$ to yield the query restriction $\vecx$ of $(\vecx, y)$.

    For a subtype $\bmu$, we define $S_{\bmu} \subseteq [N]^r \times [M]$ to be the set of vectors $(\vecx, y)$ that match $\bmu$, and let $\Pi_{\bmu}$ denote the projection onto standard basis vectors in $S_{\bmu}$.
\end{definition}

\begin{definition}\label{defn:subtypereconstructor}
    For any subtype $\bmu$ with variable symbols $x_1, \ldots, x_k$ and some specific values $y_1, \ldots, y_k \in [N]$ and $z \in [M]$, define $\reconstruct(\bmu, (y_1, \ldots, y_k), z)$ to be the vector in $[N]^r \times [M]$ obtained by taking $\bmu$ and replacing the variable symbol $x_i$ with $y_i$ for each $i$, then finally appending $z$.
\end{definition}

At this point, we make some straightforward observations. Firstly, membership of a vector $(\vecx, y)$ in a subtype $\bmu$ or a type $\blambda$ depends only on its query restriction. Also, any vector $(\vecx, y)$ that matches a subtype $\bmu$ of a type $\blambda$ must have type $\blambda$. This is due to the parity constraints in Definition \ref{defn:subtype}. Conversely, for any vector $(\vecx, y)$ of type $\blambda$, there is at least one subtype $\bmu$ of $\blambda$ that $(\vecx, y)$ matches: we can take $\vecx$, leave entries $i$ such that $\blambda_i = 1$ as they are, and replace all other distinct values by variable symbols $x_1, x_2, \ldots$. This suggests that we might be able to relate the collection of vectors in a given $\BinType$ to the collection of vectors in a given subtype.

\begin{lemma}\label{lemma:numberofsubtypes}
    Any type $\blambda$ has at most $(2r)^r$ subtypes. 
\end{lemma}
\begin{proof}
    Consider a subtype $\bmu$ of $\blambda$. Any entry in the string defining $\bmu$ must be one of the following:
    \begin{itemize}
        \item A fixed integer $i \in [N]$ such that $\blambda_i = 1$. There are at most $r$ such integers.
        \item A variable symbol $x_i$, where $i \leq r$.
    \end{itemize}
    $\bmu$ has $r$ entries, so the conclusion follows.
\end{proof}

\subsubsection{Relating Subtype Projectors to Type Projectors}\label{sec:subtypetotype}

It turns out that our main technical task to prove bounds on monogamy games in Section \ref{sec:binaryphaseconstruction} is to bound expressions of the form $$\Tr{\Pi_{\blambda} \Xi \Pi_{\blambda} \rho},$$where $\rho$ is some quantum mixed state in $\state((\CC^{N})^{\otimes r} \otimes \CC^M)$, $\Xi$ is some PSD operator, and $\blambda$ is a $\BinType$. Here, we will use the combinatorial machinery we just introduced in Section \ref{sec:subtypebasics} to reduce this to bounding expressions of the form $$\norm{\Pi_{\bmu} \Xi \Pi_{\bmu}}_\infty \cdot \Tr{\Pi_{\bmu} \rho},$$where $\bmu$ is now a subtype. Our starting point is the following lemma:
\begin{lemma}\label{lemma:pie}
    For any type $\blambda$, there exist projectors $P_{\bmu}$ for each subtype $\bmu$ of $\blambda$ such that both of the following are true:
    \begin{enumerate}
        \item\label{item:subtypedecomposition} We have $$\Pi_{\blambda} = \sum_{\bmu} P_{\bmu}.$$
        \item\label{item:psdcontainment} For every $\bmu$, we have $P_{\bmu} \leq \Pi_{\bmu}$ (with respect to the PSD ordering).
    \end{enumerate}
\end{lemma}
\begin{proof}
    Let $\bmu_1, \bmu_2, \ldots, \bmu_K$ be all the subtypes of $\blambda$ in an arbitrary order. For each $i \in [K]$, let $P_{{\bmu}_i}$ be the projector onto the span of standard basis vectors in the set $$T_{{\bmu}_i} := S_{{\bmu}_i} \backslash \bigcup_{j < i} S_{{\bmu}_j}.$$
    First, we check condition~\ref{item:subtypedecomposition}: the sets $T_{{\bmu}_i}$ are disjoint by construction and clearly contained in the set of vectors of type $\blambda$. Conversely, any vector of type $\blambda$ is contained in at least one $S_{{\bmu}_i}$ (and hence some $T_{{\bmu}_i}$). Hence the disjoint union of $\left\{T_{{\bmu}_i}: i \in [K]\right\}$ is exactly the collection of standard basis vectors with type $\blambda$. The claim follows. Finally, condition~\ref{item:psdcontainment} is immediate from the fact that $T_{{\bmu}_i} \subseteq S_{{\bmu}_i}$.
\end{proof}
\noindent
Finally, we completely reduce our problem to working with subtypes instead of types via the following lemma:
\begin{lemma}\label{lemma:subtypereduction}
    For any PSD matrix $A$, type $\blambda$, and mixed state $\rho$, we have $$\Tr{\Pi_{\blambda} A \Pi_{\blambda} \rho} \leq (2r)^r \cdot \left(\sum_{\bmu\text{ subtype of }\blambda} \norm{\Pi_{\bmu} A \Pi_{\bmu}}_\infty \cdot \Tr{\Pi_{\bmu} \rho}\right).$$
\end{lemma}
\begin{proof}
    By linearity, it suffices to prove the result when $\rho$ is a pure state $\proj{\phi}$. Moreover, let us write $A = B^\dag B$ for some matrix $B$. Then we have:
    \begin{align*}
        \Tr{\Pi_{\blambda} A \Pi_{\blambda} \rho} &= \braket{\phi | \Pi_{\blambda} B^\dag B \Pi_{\blambda} |\phi} \\
        &= \norm{B \Pi_{\blambda} \ket{\phi}}^2 \\
        &= \norm{\sum_{\bmu\text{ subtype of }\blambda} B P_{\bmu} \ket{\phi}}^2 \text{ (Lemma~\ref{lemma:pie}, condition~\ref{item:subtypedecomposition})} \\
        &\leq (2r)^r \cdot \sum_{\bmu\text{ subtype of }\blambda} \norm{BP_{\bmu} \ket{\phi}}^2 \text{ (Cauchy-Schwarz; Lemma~\ref{lemma:numberofsubtypes})} \\
        &= (2r)^r \cdot \sum_{\bmu\text{ subtype of }\blambda} \braket{\phi | P_{\bmu} B^\dag B P_{\bmu}|\phi} \\
        &= (2r)^r \cdot \sum_{\bmu\text{ subtype of }\blambda}\Tr{P_{\bmu} A P_{\bmu} \rho} \\
        &\leq (2r)^r \cdot \sum_{\bmu\text{ subtype of }\blambda} \norm{P_{\bmu} A P_{\bmu}}_\infty \cdot \Tr{P_{\bmu} \rho} \\
        &\leq (2r)^r \cdot \sum_{\bmu\text{ subtype of }\blambda} \norm{\Pi_{\bmu} A \Pi_{\bmu}}_\infty \cdot \Tr{\Pi_{\bmu} \rho}. \text{ (Lemma~\ref{lemma:pie}, condition~\ref{item:psdcontainment})}
    \end{align*}
    In more detail, the final step follows from two straightforward observations. Firstly: $\Tr{P_{\bmu} \rho} \leq \Tr{\Pi_{\bmu} \rho}$ is clear since $\Pi_{\bmu} - P_{\bmu}$ is PSD. Secondly: since $P_{\bmu}$ is a projector, we have that:
    \begin{align*}
        \norm{P_{\bmu} A P_{\bmu}}_\infty &= \max_{\ket{\phi} \in \mathrm{im} P_{\bmu}} \braket{\phi | A | \phi}\text{; and similarly} \\
        \norm{\Pi_{\bmu} A \Pi_{\bmu}}_\infty &= \max_{\ket{\phi} \in \mathrm{im} \Pi_{\bmu}} \braket{\phi | A | \phi}.
    \end{align*}
    Now we have $P_{\bmu} \leq \Pi_{\bmu} \Rightarrow \mathrm{im} P_{\bmu} \subseteq \mathrm{im} \Pi_{\bmu} \Rightarrow \norm{P_{\bmu} A P_{\bmu}}_\infty \leq \norm{\Pi_{\bmu} A \Pi_{\bmu}}_\infty$. This completes our proof.
\end{proof}

\section{Cloning Game Construction from Binary Phase States}\label{sec:binaryphaseconstruction}

In this section, we prove upper bounds on the value of restricted oracular cloning games (defined in Definition~\ref{def:veryrestrictedocg}). This section is organized as follows:
\begin{itemize}
    \item In~\Cref{sec:binaryphasesetup}, we formally state our binary phase construction and prove some preliminary lemmas, in particular relating the cloning game to a monogamy-like game with some state $\rho$ shared between the challenger and the $t+1$ players $\algo P_1, \ldots, \algo P_{t+1}$.

    \item In~\Cref{sec:expandingwithsubtypes}, we expand out the relevant operators and states in terms of \emph{subtypes} (defined in~\Cref{sec:subtypes}).

    \item In~\Cref{sec:opnormbound}, we prove spectral bounds on the operator norms of the relevant operators.

    \item In Sections~\ref{sec:freevariablesymbolcombi} and~\ref{sec:noancillas}, we provide some additional tools---namely, an analysis of the structure of the shared state $\rho$ and how it splits across subtypes---that are necessary for handling $t \mapsto t+1$ cloning games when $t > 1$. We then put everything together to prove our desired bounds in the restricted oracular cloning setting for general $t \geq 1$.
\end{itemize}

\subsection{Setup and Notation}\label{sec:binaryphasesetup}

We begin by presenting our construction. Note the qualitative similarity of this construction with Construction~\ref{construction:PRS-monogamy}; both rely centrally on binary phase states.

\begin{construction}\label{const:PRS-enc}
Let $\mathfrak{F} = \{f_\theta : \bit^n \rightarrow  \bit\}_{\theta \in \Theta}$ be a family of functions parametrized by elements $\Theta = \bit^\secp$. Consider the following $t \mapsto t+1$ oracular cloning game (as defined in Definition~\ref{def:OCG}) $\mathsf{G}_{\mathfrak{F}, t}$ with question set $\Theta$ and answer set $\mathcal{X} := \bit^n$. For any $\theta$, we will take the unitary $U_\theta$ to be $\mathsf{U}_{f_\theta} \mathsf{H}^{\otimes n}$. (Here, $\mathsf{U}_{f_\theta}$ is the phase oracle for $f_\theta$ as defined in~\Cref{sec:quantumprelims}.) In other words, for any $x \in \bit^n$, we have $$U_\theta\ket{x} = 2^{-n/2} \sum_{u \in \bit^n} (-1)^{f_\theta(u) \oplus \langle x, u \rangle} \ket{u}.$$
\end{construction}
\begin{remark}\label{remark:functionfamilychoice}
    We are being intentionally vague about the choice of function family $\mathfrak{F}$. One could imagine instantiating it with a post-quantum PRF family, to obtain a construction that is plausibly secure against arbitrary polynomial-time adversaries in the oracular cloning game. (We believe this construction would be plausibly secure because $U_\theta \ket{x}$ is pseudorandom~\cite{cryptoeprint:2018/544, brakerski2019pseudo} for any fixed $x \in \bit^n$, and hence multi-copy unclonable~\cite{werneroptimalcloning}.)

    Since we only prove oracular security in the case where $t = O(1)$ and each player can make $q = 1$ query in total, we will instead instantiate $\mathfrak{F}$ as an $O(1)$-wise uniform function family, which is statistically indistinguishable from the family of all functions from $\bit^n \rightarrow \bit$ in this query bounded game. We will reiterate this formally when establishing our final theorem in~\Cref{sec:noancillas}.
\end{remark}
\noindent
We will consider restricted quantum strategies (defined in Definition~\ref{def:veryrestrictedocg}). Recall that we use $\mathcal{S}_{\mathrm{rest}}$ to denote the collection of all such strategies. We now make a crucial observation: 

\begin{remark}\label{remark:controlqubitregister}
    Since $U_\theta^\dag = \mathsf{H}^{\otimes n} \mathsf{U}_{f_\theta}$ and each player $\algo P_i$ is given a control bit in register $\reg{E_i}$ dictating whether they will query $U_\theta$ or $U_\theta^\dag$, we can assume without loss of generality that each player simply makes one non-adaptive query to $\mathsf{U}_{f_\theta}$ as their first step. (In the event that the player is querying $U_\theta = \mathsf{U}_{f_\theta} \mathsf{H}^{\otimes n}$, they would technically need to query $\mathsf{H}^{\otimes n}$ first. We can get around this by absorbing this query to $\mathsf{H}^{\otimes n}$ into the cloning channel $\Phi$.)
\end{remark}

Recall from Definition~\ref{def:veryrestrictedocg} that each player's register $\reg{B_i}$ splits into a query register $\reg{C_i}$, an ancilla register $\reg{D_i}$, and a control qubit register $\reg{E_i}$. 
Recalling the setup in Definition~\ref{def:veryrestrictedocg} together with Lemma~\ref{lemma:ocgwlog} and Remark~\ref{remark:controlqubitregister}, we can write
\begin{equation}\label{eq:playerprojectors}
    \vec{P}_{i, x}^{U_\theta, U_\theta^\dag} = (\mathsf{U}_{f_\theta} \otimes \id_{\reg{D_iE_i}})Q_i^\dag (\proj{x} \otimes \id_{\reg{D_iE_i}}) Q_i (\mathsf{U}_{f_\theta} \otimes \id_{\reg{D_iE_i}}),
\end{equation}
for some unitaries $Q_1, \ldots, Q_{t+1}$ such that $Q_i$ acts on all the three registers $\reg{C_iD_iE_i}$.

With this in mind, we now switch from a cloning-based formulation to an entanglement-based formulation. At a high level, this follows from the ricochet property of EPR pairs (formally, Lemma~\ref{fact:CI}). Substituting $\vec{P}_{i, x}$ as defined in Equation~\eqref{eq:playerprojectors} and $U_\theta = \mathsf{U}_{f_\theta} \mathsf{H}^{\otimes n}$ (note that all entries of this unitary are real) into Lemma~\ref{lemma:CImultiplayer} implies that:
\begin{align*}
    \omega_{\mathsf{S}}(\mathsf{G}) =& 2^{n(t-1)} \cdot \underset{\theta}{\mathbb{E}} \,\underset{x \in \bit^n}{\sum}  \mathrm{Tr}\Bigg[ \Bigg(\Big( \bigotimes_{i \in [t+1]} (\mathsf{U}_{f_\theta} \otimes \id_{\reg{D_iE_i}}) Q_i^\dag (\proj{x} \otimes \id_{\reg{D_iE_i}}) Q_i (\mathsf{U}_{f_\theta} \otimes \id_{\reg{D_iE_i}})\Big) \\
    &\otimes \left(\mathsf{U}_{f_\theta}\mathsf{H}^{\otimes n} \proj{x}\mathsf{H}^{\otimes n}\mathsf{U}_{f_\theta}\right)_{\reg{A'_{1:t}}}^{\otimes t}\Bigg)\rho \Bigg], 
\end{align*}
where $\rho$ is the Choi state \begin{equation}\label{eq:multicopychoistate}
    \rho_{\reg{B_{1:t+1}A'_{1:t}}} := (\Phi_{\reg{A_1\dots A_t \rightarrow B_1 \dots B_{t+1}}} \otimes \id_{\reg{A'_{1:t}}}) \left(\proj{\mathsf{EPR}^n}^{\otimes t}\right).
\end{equation}
Now define the projector
\begin{align}\label{eq:Xi-projector}
\Xi = \sum_{x \in \bit^n} \left(\left(\underset{i \in [t+1]}{\bigotimes} Q_i^\dag \left(\proj{x}_{\reg{C_i}} \otimes \id_{\reg{D_iE_i}}\right)Q_i\right)  \otimes \left(\mathsf{H}^{\otimes n} \proj{x}\mathsf{H}^{\otimes n}\right)_{\reg{A'_{1:t}}}^{\otimes t}\right).
\end{align}

Let $d=2^n$, $r = 2t+1$, $d' = 2^{a+1}$, and $r' = t+1$. Recall that $(\CC^d)^{\otimes r} \otimes (\CC^{d'})^{\otimes r'} \cong \bigoplus_{\blambda} V_{\blambda}$ decomposes into a collection of subspaces corresponding to binary type vectors $\blambda \in \bit^d$. Here, $(\CC^{d'})^{\otimes r'}$ serves as an auxiliary register; in terms of the notation in Section \ref{sec:types}, we are taking $N = d$ and $M = d'^{r'}$ (in other words, we are packing all the players' ancillary registers into one auxiliary input).

Moreover, we assume going forward that $\mathfrak{F}$ is a $(4t+2)$-wise uniform family of functions from $\bit^n \rightarrow \bit$. As noted in Remark~\ref{remark:functionfamilychoice}, this is statistically indistinguishable from instantiating $\mathfrak{F}$ as the family of all functions from $\bit^n \rightarrow \bit$, since the expression in Lemma~\ref{lemma:CImultiplayer} has degree $4t+2$ in $\mathsf{U}_{f_\theta}$.

Then using Lemma~\ref{lemma:CImultiplayer} and then Lemma~\ref{lem:phase-type-twirl}, we get:
\begin{align*}
    \omega_\mathsf{S}(\mathsf{G}) &= 2^{n (t-1)} \Tr{\underset{f}{\mathbb{E}} \left[\left(\Uf^{\otimes r} \otimes \id_{\reg{D_{1:t+1}}}\right) \, \Xi \, \left(\Uf^{\otimes r} \otimes \id_{\reg{D_{1:t+1}}}\right) \right] \rho} \\
    &= 2^{n (t-1)}\sum_{\blambda \in \bit^d} \Tr{\Pi_{\blambda} \Xi \Pi_{\blambda} \rho},\stepcounter{equation}\tag{\theequation}\label{eqn:monogamyboundwithtypes}
\end{align*}
\noindent
where $\Pi_{\blambda}$ is the projector onto the subspace of $(\CC^d)^{\otimes r} \otimes (\CC^{d'})^{\otimes (t+1)}$ given by $$V_{\blambda} = \mathrm{span}_{\CC} \{\ket{(v_1,\dots,v_r, a_1, \ldots, a_{t+1}) }\, : \,\BinType(v_1,\dots,v_r, a_1, \ldots, a_{t+1}) = \blambda\}.$$
We now state a simple high-level bound on $\omega(\mathsf{G})$ in terms of \emph{subtypes} $\bmu$. Unlike the bounds obtained using techniques in~\cite{Tomamichel_2013}, this bound will depend on the Choi state $\rho$. This is \emph{provably} necessary for any $t \geq 2$, as we will show in~\Cref{sec:opnormcounterexample}.
\begin{lemma}\label{lemma:monogamyboundwithstate}
    We have $$\omega(\mathsf{G}) \leq \exp(O(t \log t)) \cdot 2^{n(t-1)} \cdot \sum_{\bmu} \Tr{\Pi_{\bmu} \rho} \cdot \norm{\Pi_{\bmu} \Xi \Pi_{\bmu}}_\infty.$$
\end{lemma}
\begin{proof}
This follows immediately by plugging Lemma~\ref{lemma:subtypereduction} into~\Cref{eqn:monogamyboundwithtypes}.
\end{proof}

\subsection{Expanding out $\Xi$ using Subtypes}\label{sec:expandingwithsubtypes}

We now set up some additional notation. For each $i \in [t+1]$ and $j_1,j_2 \in [d]$ and $l_1, l_2 \in [d']$, we let $Q_{i, (j_1, l_1); (j_2, l_2)}^\dag$ denote the entry in the $(j_1, l_1)$-th row and $(j_2, l_2)$-th column of the unitary $Q_i^\dag$. To keep track of the ancillary indices in registers $\reg{D_{1:t+1}}$, we will introduce the values $z_1, \ldots, z_{t+1} \in [d']$ and denote $\vecz = (z_1, \ldots, z_{t+1})$ for brevity. We can now write the projector $\Xi$ in \Cref{eq:Xi-projector} as:
\begin{align*}
    \Xi&=\sum_{\substack{x \in \bit^n\\z_1, \ldots, z_{t+1} \in \bit^{a}}} \proj{\Xi^{x, \vecz}},\text{ where}
\end{align*}
\begin{align*}
\ket{\Xi^{x, \vecz}} &= Q_1^\dag \left(\ket{x} \otimes \ket{z_1}\right) \otimes \dots \otimes Q_{t+1}^\dag \left(\ket{x} \otimes \ket{z_{t+1}}\right) \otimes \left(\mathsf{H}^{\otimes n} \ket{x}\right)^{\otimes t} \\
&= 2^{-nt/2} \sum_{\substack{v_1,\dots,v_r \in [d]\\w_1, \ldots, w_{t+1} \in [d']}} \left((-1)^{\langle v_{t+2} + \dots + v_r, x\rangle} \prod_{i=1}^{t+1} Q_{i,(v_{i},w_i); (x, z_i)}^\dag \right)\ket{v_1,\dots,v_r} \otimes \ket{w_1, \ldots, w_{t+1}}.
\end{align*}
\noindent
We now begin unpacking the operator $\Pi_{\bmu} \Xi \Pi_{\bmu}$, using the formalism of subtypes introduced in Section \ref{sec:subtypes}. Recall that we have:
\begin{align*}
    \Xi &= \sum_{\substack{x \in \bit^n\\z_1, \ldots, z_{t+1} \in \bit^{a}}} \proj{\Xi^{x, \vecz}} \\
    \Rightarrow \Pi_{\bmu} \Xi \Pi_{\bmu} &= \sum_{\substack{x \in \bit^n\\z_1, \ldots, z_{t+1} \in \bit^{a}}} \Pi_{\bmu}\proj{\Xi^{x, \vecz}}\Pi_{\bmu}.
\end{align*}
Therefore we can define a matrix $\vec A \in \CC^{d_1 \times d_2}$, where $d_1 = 2^{rn + (t+1)a}$ is the dimension of $\ket{\Xi^{x, \vecz}}$ and $d_2 = 2^{n+(t+1)a}$ is the number of possible values of $x, \vecz$. The columns of $\vec A$ are indexed by $x, \vecz$ and the corresponding column is exactly $\Pi_{\bmu} \ket{\Xi^{x, \vecz}}$. Then we have $\Pi_{\bmu} \Xi \Pi_{\bmu} = \vec A \vec A^\dag \Rightarrow \norm{\Pi_{\bmu} \Xi \Pi_{\bmu}}_\infty = \norm{\vec A}_\infty^2$.

Recall also that we have:
\begin{align*}
\ket{\Xi^{x, \vecz}} &= 2^{-nt/2} \sum_{\substack{v_1,\dots,v_r \in [d]\\w_1, \ldots, w_{t+1} \in [d']}} \left((-1)^{\langle v_{t+2} + \dots + v_r, x\rangle} \prod_{i=1}^{t+1} Q_{i,(v_{i},w_i); (x, z_i)}^\dag \right)\ket{v_1,\dots,v_r} \otimes \ket{w_1, \ldots, w_{t+1}}.
\end{align*}
Therefore, once we project onto the subspace corresponding to the subtype $\boldsymbol{\mu}$, we get the state
\begin{align*}
\Pi_{\bmu} \ket{\Xi^{x, \vecz}} &= 2^{-nt/2} \sum_{\substack{v_1,\dots,v_r \in [d]\\w_1, \ldots, w_{t+1} \in [d']\\(\vecv, \vecw) \in S_{\bmu}}} \left((-1)^{\langle v_{t+2} + \dots + v_r, x\rangle} \prod_{i=1}^{t+1} Q_{i,(v_{i},w_i); (x, z_i)}^\dag \right)\ket{v_1,\dots,v_r} \otimes \ket{w_1, \ldots, w_{t+1}}.\stepcounter{equation}\tag{\theequation}\label{eqn:expandedwithancillas}
\end{align*}
Now note that any row of $\vec A$ that does not correspond to a standard basis vector in $S_{\bmu}$ will be 0. We can discard all such rows without affecting the operator norm of $A$. We can therefore re-index the rows of $\vec A$ by the variable symbols $x_1, \ldots, x_l$ of $\bmu$ and the ancilla indices $w_1, \ldots, w_{t+1}$, so that $\vec A$ is effectively a $2^{nl + a(t+1)} \times 2^{n+a(t+1)}$ matrix.

With this setup in mind, we introduce a couple more definitions that will help us complete our analysis:

\begin{definition}\label{defn:Bmatrix}
    Let $\ell \in [0, t]$ be an integer parameter (typically we will work with $\ell = t$). Let $\bmu$ be a subtype of $(\CC^{d})^{\otimes (t+1+\ell)} \otimes (\CC^{d'})^{\otimes r'}$ with variable symbols $x_1, \ldots, x_l$. Then, define the matrix $$\vec B := \vec B_{\bmu}(\vec Q_1, \ldots, \vec Q_{t+1})$$ with dimensions $2^{nl + a(t+1)} \times 2^{n + a(t+1)}$ as follows:
    
    Its rows are indexed by $y_1, \ldots, y_l \in [d]$ and $w_1, \ldots, w_{t+1} \in [d']$. Its columns are indexed by $x \in [d]$ and $z_1, \ldots, z_{t+1} \in [d']$. For any such indices, take $$(v_1, \ldots, v_{t+1+\ell}, w_1, \ldots, w_{t+1}) = \reconstruct(\bmu, (y_1, \ldots, y_l), (w_1, \ldots, w_{t+1})) \in [d]^{t+1+\ell} \times [d']^{t+1}.$$(The function $\reconstruct$ is defined in Definition \ref{defn:subtypereconstructor}.) Then we define the entry
    \begin{equation}\label{eqn:Bdefinition}
        \vec B_{(y_1, \ldots, y_l, w_1, \ldots, w_{t+1}); (x, z_1, \ldots, z_{t+1})} = (-1)^{\langle v_{t+2} + \ldots + v_{t+1+\ell}, x \rangle} \prod_{i = 1}^{t+1} Q^\dag_{i, (v_i, w_i); (x, z_i)}.
    \end{equation}
\end{definition}

We remark that when $\ell = t$, this definition coincides with the matrix $2^{nt/2} A$. The reason we generalize to $\ell < t$ is for technical reasons; there could be variable symbols that only appear in the ``phase entries'' $v_{t+2}, \ldots, v_{t+1+\ell}$, in which case they appear an even number of times and do not have any effect on the value of that entry in the matrix $B$. These variable symbols artificially blow up the operator norm of $B$ and will need to be dealt with separately. To capture this, we have the following notion:

\begin{definition}\label{defn:subtypefreevariables}
    For a subtype $\bmu$ with respect to $(\CC^d)^{\otimes (t+1+\ell)} \otimes (\CC^{d'})^{\otimes r'}$ and variable symbol $x_i$, we say $x_i$ is a \emph{free variable symbol} of $\bmu$ if it only appears in entries $t+2, t+3, \ldots, t+1+\ell$ of $\bmu$. (Informally, a free variable symbol is one that only appears in the phase.)
\end{definition}

\subsection{Bounding \texorpdfstring{$\norm{\vec B_{\bmu}(\vec Q_1, \ldots, \vec Q_{t+1})}_\infty$}{$\norm{B_{\mu}(Q_1, \ldots, Q_{t+1})}_\infty$}}\label{sec:opnormbound}

In this section, we provide estimates on the operator norm of $\vec B_{\bmu}(\vec Q_1, \ldots, \vec Q_{t+1})$. We first begin with a lemma that allows us to dispose of free variable symbols:

\begin{lemma}\label{lemma:handlefreevariablesymbols}
    Suppose $\bmu$ is a subtype with respect to $(\CC^d)^{\otimes (2t+1)} \otimes (\CC^{d'})^{\otimes r'}$, and suppose it has $b$ free variable symbols that appear in a total of $p$ positions in indices $k+2, k+3, \ldots, 2k+1$.

    Then define $\ell := t-p$, and $\bmu'$ to be the subtype with respect to $(\CC^{d})^{\otimes (t+1+\ell)} \otimes (\CC^{d'})^{\otimes r'}$ obtained by taking $\bmu$ and removing all free variable symbols. Then we have $$\norm{\vec B_{\bmu}(\vec Q_1, \ldots, \vec Q_{t+1})}_\infty = 2^{nb/2} \norm{\vec B_{\bmu'}(\vec Q_1, \ldots, \vec Q_{t+1})}_\infty.$$
\end{lemma}
\begin{proof}
    For brevity, write $\vec B = \vec B_{\bmu}(\vec Q_1, \ldots, \vec Q_{t+1})$ and $\vec B' = \vec B_{\bmu'}(\vec Q_1, \ldots, \vec Q_{t+1})$. Let the variable symbols of $\bmu$ be $y_1, \ldots, y_l$, so that its free variable symbols are $y_{l-b+1}, \ldots, y_l$. Note that $\bmu'$ will have $l-b$ variable symbols, none of which are free variable symbols. Now since each free variable symbol appears an even number of times in the phase, we have for any indices that $$B_{(y_1, \ldots, y_l, w_1, \ldots, w_{t+1}); (x, z_1, \ldots, z_{t+1})} = B'_{(y_1, \ldots, y_{l-b}, w_1, \ldots, w_{t+1}); (x, z_1, \ldots, z_{t+1})}.$$

    In other words, the matrix $\vec B$ is obtained by vertically stacking $2^{nb}$ copies of $\vec B'$ (up to a permutation of rows). Put another way, $\vec B$ is equal to $B'$ tensored with a column vector consisting of $2^{nb}$ 1's. It follows by Lemma \ref{lemma:opnormtensor} that:
    \begin{align*}
        \norm{\vec B}_\infty &= \norm{\vec B'}_\infty \cdot \norm{\begin{bmatrix} 1 \\ 1 \\ \vdots \\ 1\end{bmatrix}}_\infty = 2^{nb/2} \norm{\vec B'}_\infty.
    \end{align*}
\end{proof}

For most of the remainder of this section, we focus on subtypes $\bmu$ that do not have any free variable symbols. We first set up some more notation. At a high level, the idea is to cluster the terms being multiplied in Equation~\eqref{eqn:Bdefinition} according to which of the variable symbols they depend on. This allows us to write $B$ as a block column-wise tensor product of several much simpler block matrices, and then we will appeal to Lemma \ref{lemma:colwisetensor}.

To this end, let $\bmu$ have variable symbols $x_1, \ldots, x_l$. For each for $i \in [l]$, let $I_i = \left\{j \in [t+1+\ell]: {\bmu}_j = x_i\right\}$ and $J = \left\{j \in [t+1+\ell]: {\bmu}_j\text{ is fixed}\right\}$. Note that $[t+1+\ell]$ is the disjoint union of $I_1, I_2, \ldots, I_l, J$. Also, for convenience we will make the following abuse of notation: for any integer $h$, subset $I \subseteq [h]$, and vector $\vecb$ with $h$ entries, we use $\vecb_I$ to denote the sub-vector of length $|I|$ obtained by taking only the indices in $I$ from $\vecb$. With this in mind, for each $i \in [l]$ define the following matrix $\vec M_i$ of dimensions $2^{n + a(|I_i| \cap [t+1])} \times 2^{n + a(|I_i| \cap [t+1])}$:

$$M_{i, (y_i, \vecw_{I_i \cap [t+1]}); (x, \vecz_{I_i \cap [t+1]})} = \prod_{j \in I_i \cap [t+2, t+1+\ell]} (-1)^{\langle y_i, x \rangle} \cdot \prod_{j \in I_i \cap [t+1]} Q^\dag_{j, (y_i, w_j); (x, z_j)}.$$
Additionally, define the following matrix $\vec T$ of dimensions $2^{a(|J| \cap [t+1])} \times 2^{n+a(|J| \cap [t+1])}$:

$$\vec T_{\vecw_{J \cap [t+1]}; (x, \vecz_{J \cap [t+1]})} = \prod_{j \in J \cap [t+2, t+1+\ell]} (-1)^{\langle{\bmu}_j, x\rangle} \cdot \prod_{j \in J \cap [t+1]} Q^\dag_{j, ({\bmu}_j, w_j); (x, z_j)}.$$

It is clear by inspection that $\vec B_{\bmu}$ is the result of applying the block column-wise tensoring operation described in Lemma \ref{lemma:colwisetensor} to $\vec M_1, \ldots,\vec  M_l, \vec T$. Here, the block columns are indexed by $x$. We will proceed by applying this lemma to these matrices. We thus need to check that the preconditions of the lemma apply, which we do in the next few lemmas:

\begin{lemma}\label{lemma:colwisetensorpreconditionT}
    For any $x^* \in [2^n]$, consider the matrix $\vec T_{x^*}$ obtained by restricting $\vec T$ to columns where $x = x^*$. Then $\norm{\vec T_{x^*}}_\infty \leq 1$. Moreover, if $J \cap [t+1]$ is nonempty, then we have $\norm{\vec T}_\infty \leq 1$.
\end{lemma}
\begin{proof}
    We have:$$T_{\vecw_{J \cap [t+1]}; (x, \vecz_{J \cap [t+1]})} = \prod_{j \in J \cap [t+2, t+1+\ell]} (-1)^{\langle{\bmu}_j, x\rangle} \cdot \prod_{j \in J \cap [t+1]} Q^\dag_{j, ({\bmu}_j, w_j); (x, z_j)}.$$
    Now consider the matrix $T'$ with the same dimensions as $T$ defined by:$$T'_{\vecw_{J \cap [t+1]}; (x, \vecz_{J \cap [t+1]})} = \prod_{j \in J \cap [t+1]} Q^\dag_{j, ({\bmu}_j, w_j); (x, z_j)}.$$
    Since $\bmu_j$ is fixed for $j \in J$, $\vec T'$ can be obtained from $\vec T$ by just flipping the signs of some columns. This preserves the operator norm (this can be seen from Lemma \ref{lemma:opnormdef} for example), so we have $\norm{\vec T'}_\infty = \norm{\vec T}_\infty$. It also follows analogously that $\norm{\vec T'_{x^*}}_\infty = \norm{\vec T_{x^*}}_\infty$, where we analogously define $\vec T'_{x^*}$ as the result of restricting $\vec T'$ to columns where $x = x^*$. At this point, we split into two cases:
    \begin{itemize}
        \item If $J \cap [t+1]$ is nonempty, we claim that $\vec T'$ is a submatrix of $\vec Q := \bigotimes_{j \in J \cap [t+1]} Q^\dag_j$. Indeed, we can index the rows of $\vec Q$ by $(\veca, \vecb)$ and the columns by $(\vecc, \vecd)$, and write $$Q_{(\veca, \vecb); (\vecc, \vecd)} = \prod_{j \in J \cap [t+1]} Q^\dag_{j, (a_j, b_j); (c_j, d_j)}.$$
        Then $\vec T'$ is the submatrix of $\vec Q$ obtained by restricting to rows $(\veca, \vecb)$ such that $a_j = {\bmu}_j$ for all $j \in J \cap [t+1]$ and columns $(\vecc, \vecd)$ such that $c_{j_1} = c_{j_2}$ for any $j_1, j_2 \in J \cap [t+1]$.

        Since the operator norm of a unitary matrix is 1 and there is at least one unitary matrix in this tensor product, it follows from Lemmas \ref{lemma:opnormsubmatrix} and \ref{lemma:opnormtensor} that $\norm{\vec T'}_\infty \leq 1 \Rightarrow \norm{\vec T}_\infty \leq 1$. Then, it also follows that $\norm{\vec T_{x^*}}_\infty \leq 1$ by Lemma \ref{lemma:opnormsubmatrix}.

        \item If $J \cap [t+1]$ is empty then $T'$ is really just a vector of $2^n$ many 1's. Hence $T'_{x^*}$ is just the scalar 1, which trivially has operator norm $\leq 1$.
    \end{itemize}
\end{proof}

\begin{lemma}\label{lemma:colwisetensorpreconditionM}
    Assume $\bmu$ does not have free variable symbols. For any $x^* \in [2^n]$ and $i \in [l]$, consider the matrix $\vec M_{i, x^*}$ obtained by restricting $M_i$ to columns where $x = x^*$. Then $\norm{\vec M_{i, x^*}}_\infty \leq 1$.
\end{lemma}
\begin{proof}
    We have:
    $$\vec M_{i, (y_i, \vecw_{I_i \cap [t+1]}); (x^*, \vecz_{I_i \cap [t+1]})} = \prod_{j \in I_i \cap [t+2, t+1+\ell]} (-1)^{\langle y_i, x^* \rangle} \cdot \prod_{j \in I_i \cap [t+1]} Q^\dag_{j, (y_i, w_j); (x^*, z_j)}.$$
    We can now define another matrix $\vec M'_{i, x^*}$ with the same dimensions as $\vec M_{i, x^*}$ defined by:
    $$\vec M'_{(i, x^*), (y_i, \vecw_{I_i \cap [t+1]}); \vecz_{I_i \cap [t+1]}} = \prod_{j \in I_i \cap [t+1]} Q^\dag_{j, (y_i, w_j); (x^*, z_j)}.$$
    Since we are fixing $x^*$, $\vec M'_{i, x^*}$ can be obtained from $\vec M_{i, x^*}$ by just flipping the signs of some rows. It follows that $\norm{\vec M'_{i, x^*}}_\infty = \norm{\vec M_{i, x^*}}_\infty$. Now to finish, we argue that $\vec M'_{i, x^*}$ is a submatrix of $\vec Q := \bigotimes_{j \in I_i \cap [t+1]} \vec Q^\dag_j$. Note that $I_i \cap [t+1]$ must be non-empty as otherwise $x_i$ would be a free variable symbol. Given this, this claim would imply the conclusion by Lemmas \ref{lemma:opnormsubmatrix} and \ref{lemma:opnormtensor}.

    To see this claim, note that we can index the rows of $\vec Q$ by $(\veca, \vecb)$ and the columns by $(\vecc, \vecd)$, and write:
    $$Q_{(\veca, \vecb); (\vecc, \vecd)} = \prod_{j \in I_i \cap [t+1]} Q^\dag_{j, (a_j, b_j); (c_j, d_j)}.$$
    Then $\vec M'_{i, x^*}$ is the submatrix of $\vec Q$ obtained by restricting to rows $(\veca, \vecb)$ where $a_{j_1} = a_{j_2}$ for any $j_1, j_2 \in I_i \cap [t+1]$ and columns $(\vecc, \vecd)$ where $c_j = x^*$ for all $j \in I_i \cap [t+1]$. This completes the proof of the lemma.
\end{proof}

\begin{lemma}\label{lemma:opnormhardnophase}
    Assume $\bmu$ does not have free variable symbols. Consider some $i \in [l]$ such that the integer $|I_i \cap [t+2, t+1+\ell]|$ is even. Then, it holds that $\norm{\vec M_i}_\infty \leq 1$.
\end{lemma}
\begin{proof}
    \begin{align*}
        M_{i, (y_i, \vecw_{I_i \cap [t+1]}); (x, \vecz_{I_i \cap [t+1]})} &= \prod_{j \in I_i \cap [t+2, t+1+\ell]} (-1)^{\langle y_i, x \rangle} \cdot \prod_{j \in I_i \cap [t+1]} Q^\dag_{j, (y_i, w_j); (x, z_j)} \\
        &= \prod_{j \in I_i \cap [t+1]} Q^\dag_{j, (y_i, w_j); (x, z_j)},
    \end{align*}
    since there are an even number of identical terms being multiplied together in the first product. In this case, we just argue that $\vec M_i$ is a submatrix of $\vec Q := \bigotimes_{j \in I_i \cap [t+1]} \vec Q^\dag_j$. This is a non-empty tensor product since otherwise $x_i$ would be a free variable symbol. Given this, this claim would imply the conclusion by Lemmas \ref{lemma:opnormsubmatrix} and \ref{lemma:opnormtensor}.

    To see this claim, we once again index the rows of $\vec Q$ by $(\veca, \vecb)$ and the columns by $(\vecc, \vecd)$, and write:
    $$Q_{(\veca, \vecb); (\vecc, \vecd)} = \prod_{j \in I_i \cap [t+1]} Q^\dag_{j, (a_j, b_j); (c_j, d_j)}.$$

    Then, $\vec M_i$ is the submatrix of $\vec Q$ obtained by restricting to rows $(\veca, \vecb)$ such that $a_{j_1} = a_{j_2}$ for any $j_1, j_2 \in I_i \cap [t+1]$ and columns $(\vecc, \vecd)$ such that $c_{j_1} = c_{j_2}$ for any $j_1, j_2 \in I_i \cap [t+1]$. This completes the proof of the lemma.
\end{proof}

Our final technical lemma handles the case where a variable symbol appears multiple times among $v_1, \ldots, v_{t+1}$:
\begin{lemma}\label{lemma:opnormhardcase}
    Consider some $i \in [l]$ be such that $|I_i \cap [t+1]| \geq 2$ (we are assuming that such $i$ exists; this may not always be the case). Then $$\norm{\vec M_i}_\infty \leq 1.$$
\end{lemma}
\begin{proof}
    Firstly, if $|I_i \cap [t+1]|$ is even, then by the parity constraints (the variable symbol $x_i$ should appear an even number of times in $\bmu$), we must also have that $|I_i \cap [t+2, t+1+\ell]|$ is even. In this case, the conclusion would follow from Lemma~\ref{lemma:opnormhardnophase}. Hence from now on we assume that $|I_i \cap [t+1]|$ is odd. We hence have:
    \begin{align*}
        \vec M_{i, (y_i, \vecw_{I_i \cap [t+1]}); (x, \vecz_{I_i \cap [t+1]})} &= \prod_{j \in I_i \cap [t+2, t+1+\ell]} (-1)^{\langle y_i, x \rangle} \cdot \prod_{j \in I_i \cap [t+1]} Q^\dag_{j, (y_i, w_j); (x, z_j)} \\
        &= (-1)^{\langle y_i, x \rangle} \cdot \prod_{j \in I_i \cap [t+1]} Q^\dag_{j, (y_i, w_j); (x, z_j)}.
    \end{align*}
    The conclusion now follows by applying Lemma~\ref{lemma:manymatrixkhatrirao} with the following inputs:
    \begin{itemize}
        \item We will take $R = C = 2^n$, and $d = |I_i \cap [t+1]| \geq 2$. (In fact, $d \geq 3$ since $|I_i \cap [t+1]|$ is odd, but this will not matter for us.)
        \item The matrices will be $\left\{\vec Q_j^\dag \in \CC^{2^{n+a} \times 2^{n+a}}\right\}_{j \in I_i \cap [t+1]}$. Accordingly, we will have $$r_1 = \ldots = r_R = c_1 = \ldots = c_C = d'.$$

        These matrices are unitary so they have operator norm exactly 1.
        \item For each $y_i, x \in \bit^n$, the scalar $\gamma_{y_i, x}$ will be $(-1)^{\langle y_i, x \rangle}$, which clearly has magnitude 1.
    \end{itemize}
\end{proof}

Finally, we can put these lemmas together to prove the bounds that we want:

\begin{lemma}\label{lemma:finalopnormbound}
    Suppose $\bmu$ is a subtype with respect to $(\CC^d)^{\otimes 2t+1} \otimes (\CC^{d'})^{r'}$ with $b$ free variable symbols. Then we have $\norm{\Pi_{\bmu} \Xi \Pi_{\bmu}}_\infty \leq 2^{-nt+nb}$.

    Moreover, when $t = 1$, we must have $b = 0$ and hence $\norm{\Pi_{\bmu} \Xi \Pi_{\bmu}}_\infty \leq 2^{-n}$.
\end{lemma}
\begin{proof}
    We first address the final claim about the $t = 1$ case. Indeed, a subtype with respect to $(\CC^d)^{\otimes 2t+1} \otimes (\CC^{d'})^{\otimes t+1} = (\CC^d)^{\otimes 3} \otimes (\CC^{d'})^{\otimes 2}$ cannot have free variable symbols. Definition~\ref{defn:subtypefreevariables} states that a free variable symbol of $\bmu$ could only appear in entry 3 of $\bmu$. But a free variable symbol must appear an even number of times (as specified by Definition~\ref{defn:subtype}), so in fact it cannot appear at all. Now let us turn to proving the desired bound.

    Now let $\bmu'$ be defined as in the statement of Lemma \ref{lemma:handlefreevariablesymbols} i.e. it is $\bmu$ but with free variable symbols removed. Then we would like to show:
    \begin{align*}
        \norm{\Pi_{\bmu} \Xi \Pi_{\bmu}}_\infty &\leq 2^{-nt+nb} \\
        \Leftrightarrow \norm{\vec A}_\infty^2 &\leq 2^{-nt+nb} \\
        \Leftrightarrow \norm{\vec B_{\bmu}(\vec Q_1, \ldots, \vec Q_{k+1})}_\infty^2 &\leq 2^{nb} \\
        \Leftrightarrow \norm{\vec B_{\bmu'}(\vec Q_1, \ldots, \vec Q_{k+1})}_\infty &\leq 1. 
        \text{ (Lemma \ref{lemma:handlefreevariablesymbols})}
    \end{align*}
    Let $\bmu'$ have $l$ variable symbols. As hinted at earlier, we will bound this by applying Lemma \ref{lemma:colwisetensor} to the matrices $\vec M_1, \ldots, \vec M_l$ and $\vec T$ defined with respect to $\bmu'$. The first precondition follows from Lemmas \ref{lemma:colwisetensorpreconditionT} and \ref{lemma:colwisetensorpreconditionM}. To check the second precondition, we need only show that $\min(\norm{\vec M_1}_\infty, \ldots, \norm{\vec M_l}_\infty, \norm{\vec T}_\infty) \leq 1$. For this, we have some light casework:

    \begin{enumerate}
        \item If at least one of $\bmu_1, \ldots, \bmu_{t+1}$ is fixed, then $J \cap [t+1]$ is nonempty, so it follows that $\norm{T}_\infty \leq 1$ by Lemma \ref{lemma:colwisetensorpreconditionT}.

        \item Otherwise, all of $\bmu_1, \ldots, \bmu_{t+1}$ must be variable symbols. However, every variable symbol must appear at least twice and we only have $2t+1$ entries in total, so the total number of variable symbols must be $\leq t$. Therefore by the pigeonhole principle, some two of $\bmu_1, \ldots, \bmu_{t+1}$ are the same variable symbol i.e. there exists $i \in [l]$ such that $|I_i \cap [t+1]| \geq 2$. In this case, Lemma~\ref{lemma:opnormhardcase} tells us that $\norm{\vec M_i}_\infty \leq 1$. 
    \end{enumerate}
\end{proof}

\subsection{Combinatorial Lemmas about Free Variable Symbols}\label{sec:freevariablesymbolcombi}

In the case where $t > 1$, free variable symbols could exist, and as indicated by Lemma \ref{lemma:finalopnormbound}, they can blow up the operator norms we care about. To mitigate this, we establish some simple lemmas about free variable symbols:

\begin{definition}
    For any $l \in [t]$, define the projector $\Gamma_l$ over $(\CC^d)^{\otimes r} \otimes (\CC^{d'})^{\otimes (t+1)}$ as the projector onto $$W_l := \mathrm{span}_{\CC} \{\ket{(v_1,\dots,v_r, a_1, \ldots, a_{t+1}) } \, : \,\text{exactly }l\text{ distinct values among }v_{t+2}, \ldots, v_r\}.$$
\end{definition}

\begin{lemma}\label{lemma:freevartodistinctval}
    We have $$\sum_{b \leq t/2} \sum_{\bmu\text{ with }b\text{ free variable symbols}} 2^{nb} \Pi_{\bmu} \leq \exp(O(t \log t)) \cdot \sum_{l \leq t} 2^{n(t-l)} \Gamma_l,$$with respect to the PSD ordering.
\end{lemma}
\begin{proof}
    Note that the LHS and RHS are both diagonal in the standard basis. Hence it suffices to show for any $\vecx = (v_1, \ldots, v_r, a_1, \ldots, a_{t+1})$ that: $$\sum_{b \leq t/2} \sum_{\bmu\text{ with }b\text{ free variable symbols}} 2^{nb} \braket{\vecx | \Pi_{\bmu} | \vecx} \leq \exp(O(t \log t)) \cdot \sum_{l \leq t} 2^{n(t-l)} \braket{\vecx | \Gamma_l | \vecx}.$$

    Let $l^*$ be the number of distinct values among $v_{t+2}, \ldots, v_r$, then the RHS is $\exp(O(t \log t)) \cdot 2^{n(t- l^*)}$. On the other hand, the LHS is equal to: $$\sum_{b \leq t/2} \sum_{\substack{{\bmu}\text{ with }b\text{ free variable symbols} \\ \vecx \in S_{\bmu}}} 2^{nb}.$$

    Now consider any subtype $\bmu$ with $b$ free variable symbols such that $\vecx \in S_{\bmu}$. Each of its $b$ free variable symbols must appear at least twice among $v_{t+2}, \ldots, v_r$ due to the parity constraint, which implies that we must have $l^* \leq t - b \Leftrightarrow b \leq t - l^*$. Hence every term in the above sum is at most $2^{n(t-l^*)}$. Moreover, by Lemma~\ref{lemma:numberofsubtypes}, there are at most $\exp(O(t \log t))$ subtypes $\bmu$ with $\vecx \in S_{\bmu}$, so there are at most $\exp(O(t \log t))$ terms in the above sum. It follows that the LHS is at most $\exp(O(t \log t)) \cdot 2^{n(t-l^*)}$, which is exactly the RHS, as desired.
\end{proof}

We make one more observation:
\begin{lemma}\label{lemma:distinctcount}
    The number of tuples $(x_1, \ldots, x_t) \in [2^n]^t$ with $l$ distinct values is at most $\exp(t \log t) \cdot 2^{nl}$.
\end{lemma}
\begin{proof}
    There are at most $2^{nl}$ ways to choose the $l$ distinct values. Then there are $l^t \leq t^t = \exp(t \log t)$ ways to assign a value to each individual $x_i$. The conclusion follows.
\end{proof}

Next, we present the only specific property of the shared state $\rho$ that we need. In the following, we partition the Hilbert space $(\CC^d)^{\otimes r} \otimes (\CC^{d'})^{\otimes t+1}$ as the tensor product of Hilbert spaces on the following systems:
\begin{itemize}
    \item $\reg{R_1}$: this consists of the values $(v_{t+2}, \ldots, v_r)$. Thus $\algo{H}_{\reg{R_1}} \cong (\CC^d)^{\otimes t}$.
    \item $\reg{R_2}$: all other values i.e. $(v_1, \ldots, v_{t+1}, a_1, \ldots, a_{t+1})$. Thus $\algo{H}_{\reg{R_2}} \cong (\CC^d)^{\otimes t+1} \otimes (\CC^{d'})^{\otimes t+1}$.
\end{itemize}

\begin{lemma}\label{lemma:choipartialtrace}
    We have $$\mathrm{Tr}_{\reg R_2} \left[\rho\right] = 2^{-nt} \cdot \id_{\reg R_1}.$$

    Informally, if we take $\rho$ and trace out the system $\reg{R_2}$, we are left with a maximally mixed state.
\end{lemma}
\begin{proof}
    We have by definition that:
    \begin{align*}
        \rho &= (\id_{\reg{R_1}} \otimes \Phi_{\reg{R_2}})\left(\proj{\mathsf{EPR}}^{\otimes nt}_{\reg{R_1, R_2}}\right) \\
        &= 2^{-nt} \cdot \sum_{\vecx, \vecy \in [2^n]^t} (\id_{\reg{R_1}} \otimes \Phi) \left(\ketbra{\vecx}{\vecy}_{\reg R_1} \otimes \ketbra{\vecx}{\vecy}_{\reg R_2}\right) \\
        &= 2^{-nt} \cdot \sum_{\vecx, \vecy \in [2^n]^t} \ketbra{\vecx}{\vecy}_{\reg R_1} \otimes \Phi(\ketbra{\vecx}{\vecy})_{\reg R_2} \\
        \Rightarrow \mathrm{Tr}_{\reg{R_2}}\left[\rho\right] &= 2^{-nt} \cdot \sum_{\vecx, \vecy \in [2^n]^t} \ketbra{\vecx}{\vecy}_{\reg R_1} \cdot \Tr{\Phi(\ketbra{\vecx}{\vecy})_{\reg R_2}} \\
        &= 2^{-nt} \cdot \sum_{\vecx \in [2^n]^t} \proj{\vecx}_{\reg{R_1}},
    \end{align*}
    which implies the conclusion. In the last step, we are using the fact that $\Phi$ is trace-preserving.
\end{proof}

Finally, we put these two together to show the following:

\begin{lemma}\label{lemma:gammatrace}
    For any integer $l \in [1, t]$, we have $$\Tr{\Gamma_l \rho} \leq \exp(t \log t) \cdot 2^{-nt+nl}.$$
\end{lemma}
\begin{proof}
    We can clearly write $$\Gamma_l = \Gamma'_{l, \reg{R_1}} \otimes \id_{\reg{R_2}},$$
    where $\Gamma'_l$ is the projector onto standard basis vectors $\ket{v_{t+2}, \ldots, v_r}$ with exactly $l$ distinct values. We hence have:
    \begin{align*}
        \Tr{\Gamma_l \rho} &= \Tr{\left(\Gamma'_{l, \reg{R_1}} \otimes \id_{\reg{R_2}}\right) \rho} \\
        &= \Tr{\Gamma'_{l, \reg{R_1}} \left(\mathrm{Tr}_{\reg{R_2}} \rho\right)} \\
        &= 2^{-nt} \cdot \Tr{\Gamma'_{l, \reg{R_1}}} \text{ (Lemma~\ref{lemma:choipartialtrace})} \\
        &\leq \exp(t \log t) \cdot 2^{-nt+nl},
    \end{align*}
    where in the last step we are using the fact that $\Gamma'_{l, \reg{R_1}}$ is a projector together with Lemma~\ref{lemma:distinctcount}.
\end{proof}

\subsection{Putting Everything Together}\label{sec:noancillas}

In this section, we complete our treatment of $t \mapsto t+1$ cloning games for general $t \geq 1$. Our bounds here are once again independent of the number of ancilla qubits $a$ used by each player.

\begin{theorem}\label{thm:tcopycloning}
    Let $\mathfrak{F}$ be a $(4t+2)$-wise uniform family of functions from $\bit^n \rightarrow \bit$. Then for all $n$, we have $$\underset{\mathsf{S} \in \algo{S}_{\mathrm{rest}}}{\sup} \omega_{\mathsf{S}}(\mathsf{G}_{\mathfrak{F}, t}) \leq \exp(O(t \log t)) \cdot 2^{-n}.$$
\end{theorem}
\begin{proof}
    For any strategy $\mathsf{S} \in \algo{S}_{\mathrm{rest}}$, we have:
    \begin{align*}
        \omega_{\mathsf{S}}(\mathsf{G}) &\leq \exp(O(t \log t)) \cdot 2^{n(t-1)} \cdot \sum_{\bmu} \Tr{\Pi_{\bmu}\rho} \cdot \|\Pi_{\bmu} \hspace{0.3mm} \Xi \, \Pi_{\bmu}\|_\infty \text{ (Lemma~\ref{lemma:monogamyboundwithstate})} \\
        &\leq \exp(O(t \log t)) \cdot 2^{n(t-1)} \cdot \sum_{b \leq t/2} \sum_{{\bmu}\text{ with }b\text{ free variable symbols}} 2^{-nt+nb} \cdot \Tr{\Pi_{\bmu}\rho} \text{ (Lemma \ref{lemma:finalopnormbound})} \\
        &\leq \exp(O(t \log t)) \cdot 2^{-n} \cdot \sum_{l \leq t} 2^{n(t-l)} \cdot  \Tr{\Gamma_l \rho} \text{ (Lemma~\ref{lemma:freevartodistinctval})} \\
        &\leq \exp(O(t \log t)) \cdot 2^{-n} \cdot \sum_{l \leq t} 2^{n(t-l)} \cdot 2^{-nt+nl} \text{ (Lemma \ref{lemma:gammatrace})} \\
        &\leq \exp(O(t \log t)) \cdot 2^{-n},
    \end{align*}
    as desired.
\end{proof}

%% file: tfkw-salted-phase.tex
\section{Limitations of Analyzing Monogamy Games Using Existing Techniques}\label{sec:monogamyexisting}

In this section, we revisit the existing techniques laid out by~\cite{Tomamichel_2013} for upper bounding the value of monogamy games (and hence cloning games in particular). We begin by summarizing their technique and main result, modifying the presentation to accommodate the multi-copy setting. For a particular strategy $\mathsf{S}$, define the operator
\begin{equation}\label{eq:pitheta}
    \Pi^\theta := |\mathcal{X}|^{t-1} \cdot   \sum_{x \in \algo X}  \, \left( \bigotimes_{i = 1}^{t+1} \vec{P}_{i, x}^{\theta} \otimes (\matA_x^\theta)^{\otimes t}\right),\text{ where we define}
\end{equation}
$$\matA_x^\theta := \bar{U}_\theta \proj{x}_{\reg{A}} \bar{U}_\theta^\dag.$$
Then, by Lemma~\ref{lemma:CImultiplayer}, we have:
\begin{align*}
    \omega_{\mathsf{S}}(\mathsf{G}) &= \underset{\theta \sim \Theta}{\mathbb{E}} \left[\mathrm{Tr}\left[\Pi^\theta \rho\right]\right] \\
    &\leq \norm{\underset{\theta \sim \Theta}{\mathbb{E}}\left[\Pi^\theta\right]}_\infty,\stepcounter{equation}\tag{\theequation}\label{eqn:cloningboundviaopnorm}
\end{align*}
since $\rho$ has trace 1 (by Lemma~\ref{lemma:matrixinnerproduct}). To bound this in the $t = 1$ case, they show the following result:%
\begin{theorem}[Essentially~\cite{Tomamichel_2013}, Theorem 4]\label{thm:tfkwmain}
    When $t = 1$, we have
    \begin{align*}
        \norm{\underset{\theta \sim \Theta}{\mathbb{E}}\left[\Pi^\theta\right]}_\infty &\leq \frac{1}{|\Theta|} + \frac{|\Theta|-1}{|\Theta|} \cdot \underset{\substack{\theta, \theta' \in \Theta\\\theta \neq \theta'}}{\max}\text{ } \underset{x, x' \in \mathcal{X}}{\max} \norm{\sqrt{\matA_x^{\theta}} \sqrt{\matA_{x'}^{\theta'}}}_\infty \\
        &= \frac{1}{|\Theta|} + \frac{|\Theta|-1}{|\Theta|} \cdot \underset{\substack{\theta, \theta' \in \Theta\\\theta \neq \theta'}}{\max}\text{ } \underset{x, x' \in \mathcal{X}}{\max} \norm{\matA_x^{\theta} \matA_{x'}^{\theta'}}_\infty.
    \end{align*}
    (The second line holds because $\matA_x^\theta$ is always a projector.)
\end{theorem}
\noindent
Using this theorem, they are able to show that the BB84 monogamy game has value $\frac{1}{2} + \frac{1}{2\sqrt{2}}$. Moreover, they show that the $n$-fold parallel repetition of the BB84 monogamy game has value $\left(\frac{1}{2} + \frac{1}{2\sqrt{2}}\right)^n$. These techniques have since been adapted by~\cite{Culf_2022, schleppy2025winning} to cloning games based on coset states as well. In this section, we will demonstrate three things that we previously alluded to in~\Cref{sec:prevtechniques}:

\begin{itemize}
    \item Section~\ref{sec:opnormcounterexample}: In the case of the case of multi-copy cloning games (i.e. $t > 1$), we will show that ignoring the shared state $\rho$ can \emph{provably} lose too much. Specifically, we will show for any $\left\{U_\theta\right\}$ with real entries, there exist projectors $\vec{P}_{i, x}^\theta$ such that
    $$\norm{\underset{\theta \sim \Theta}{\mathbb{E}} \left[\Pi^\theta\right]}_\infty \geq 1.$$
    Moreover, we will show this even if $\vec{P}_{i, x}^\theta$ does not even depend on $\theta$. (This can informally be thought of as the case where the players $\algo P_1, \ldots, \algo P_{t+1}$ are never given access to the basis $\theta$ in any form.)

    Since previous constructions~\cite{Tomamichel_2013, broadbent_et_al:LIPIcs.TQC.2020.4, 10.1007/978-3-030-84242-0_20, Culf_2022, schleppy2025winning} as well as our binary phase state construction only make use of unitaries with real entries, this justifies the necessity of our finer analysis in~\Cref{sec:freevariablesymbolcombi} of the specific structure imposed on the shared state $\rho$ by the oracular cloning setting.

    \item Section~\ref{sec:tfkwendoftheline}: For the single-copy case, we will demonstrate another inherent limitation that arises from the~\cite{Tomamichel_2013} approach of bounding a spectral norm in terms of the maximal pairwise overlap $\max_{\theta \neq \theta'} \norm{\matA^\theta_x \matA^{\theta'}_{x'}}_\infty$. Namely, we will show that this pairwise overlap is provably at least $1/\sqrt{|\mathcal{X}|} = 2^{-n/2}$ (as opposed to the ideal $O(2^{-n})$ for any monogamy game). Recall from~\Cref{sec:applicationbh} that this is crucial for our application to black hole physics (we will elaborate on this more in~\Cref{sec:black-hole-games} and Remark~\ref{remark:bhblcomparison} in particular).

    Thus, in order to bypass this limitation of Theorem~\ref{thm:tfkwmain}, we need entirely new techniques. This is yet another reason why we restrict attention to \emph{oracular cloning games} (defined in~\Cref{sec:oracularcloning}) and analyze these with a completely different technique based on binary subtypes in Section \ref{sec:binaryphaseconstruction}.

    \item Finally, for completeness, we show in~\Cref{sec:salting} that this bound of $2^{-n/2}$ is essentially tight in the single-copy case; we will show that a slight tweak of our binary phase state construction --- introduced in~\Cref{sec:techoverviewbinaryphase} and fleshed out in~\Cref{sec:binaryphaseconstruction} --- yields a different monogamy game with $\mathcal{X} = \bit^n$ and maximal pairwise overlap $\leq 2^{-n/2 + o(n)}$. This only requires the existence of sub-exponentially \emph{classically} secure PRFs.
\end{itemize}

\subsection{Limitations of Bounding the Operator Norm Directly}\label{sec:opnormcounterexample}

Recall that, in~\Cref{sec:binaryphaseconstruction}, we proved a multi-copy cloning bound with the assistance of a careful analysis (in~\Cref{sec:freevariablesymbolcombi}) of the structure of the shared state $\rho$ defined in Equation~\eqref{eq:multicopychoistate}, namely the structure specified by Lemma~\ref{lemma:choipartialtrace}. In this section, we argue that this more careful approach is likely \emph{necessary}:

\begin{lemma}
    Consider any $t > 1$ and projectors $\left\{\matA_x^\theta = \bar{U}_\theta \proj{x} \bar{U}_\theta^\dag\right\}$ such that $U_\theta$ has real entries for all $\theta$ i.e. $U_\theta = \bar{U}_\theta$. Then there exist projectors $\vec{P}_{i, x}^\theta$ such that, if $\Pi^\theta$ is defined as in Equation~\eqref{eq:pitheta}, then we have:
    $$\norm{\underset{\theta \sim \Theta}{\mathbb{E}} \left[\Pi^\theta\right]}_\infty \geq |\mathcal{X}|^{\lfloor t/2 \rfloor - 1}.$$
    Moreover, this bound can be attained even if these projectors need not depend at all on $\theta$. Note that for any $t \geq 2$, the RHS is $\geq 1$ so this implies that Equation~\eqref{eqn:cloningboundviaopnorm} does not give any non-trivial bound on the value of the corresponding cloning game.
\end{lemma}
\begin{proof}
    Fix some element $y \in \mathcal{X}$. For all $i, \theta$, we will define $\vec{P}_{i, x}^\theta$ to be $\id$ if $x = y$ and 0 otherwise. In other words, each player is simply going to guess $x = y$ deterministically. Then Equation~\eqref{eq:pitheta} tells us that:
    \begin{align*}
        \Pi^\theta &= |\mathcal{X}|^{t-1} \cdot \left(\bigotimes_{i = 1}^{t+1} \id \otimes \left(\matA_y^\theta\right)^{\otimes t}\right) \\
        \Rightarrow \norm{\underset{\theta \sim \Theta}{\mathbb{E}} \left[\Pi^\theta\right]}_\infty &= |\mathcal{X}|^{t-1} \cdot \norm{\underset{\theta \sim \Theta}{\mathbb{E}} \left[\left(\matA_y^\theta\right)^{\otimes t}\right]}_\infty.
    \end{align*}
    To finish, it suffices to exhibit some mixed state $\rho$ such that:
    \begin{equation}\label{eq:opnormcounterexamplegoal}
        \Tr{\underset{\theta \sim \Theta}{\mathbb{E}} \left[\left(\matA_y^\theta\right)^{\otimes t}\right] \rho} \geq |\mathcal{X}|^{\lfloor t/2 \rfloor - t}.
    \end{equation}
    Define the bit $b = t \bmod{2}$, and let $$\rho := \left(\bigotimes_{i = 1}^{\lfloor t/2 \rfloor} \proj{\mathsf{EPR}_{\mathcal{X}}}\right) \otimes \left(\frac{1}{|\mathcal{X}|}\id_{\mathcal{X}}\right)^{\otimes b}.$$In other words, we let $\rho$ comprise $\lfloor t/2 \rfloor$ EPR qudits supported on $\mathcal{X}$, and in the event that $t$ is odd we also include one qudit that is maximally mixed on $\mathcal{X}$. Intuitively, we will use the entanglement between each EPR pair to increase the probability that the $t$ measurements specified by $\left\{\matA_x^\theta\right\}$ measurements all output $y$. Indeed, for any particular $\theta$ we have:
    \begin{align*}
        \Tr{\left(\matA_y^\theta\right)^{\otimes t} \rho} &= \left(\Tr{\left(\matA_y^\theta\right)^{\otimes 2} \proj{\mathsf{EPR}_{\mathcal{X}}}}\right)^{\lfloor t/2 \rfloor} \cdot \frac{1}{|\mathcal{X}|^b} \\
        &= \left(\frac{1}{|\mathcal{X}|} \Tr{\matA_y^\theta \overline{{\matA}_y^\theta}}\right)^{\lfloor t/2 \rfloor} \cdot \frac{1}{|\mathcal{X}|^b} \text{ (Corollary~\ref{cor:CIprojector})} \\
        &= \left(\frac{1}{|\mathcal{X}|} \Tr{\matA_y^\theta \matA_y^\theta}\right)^{\lfloor t/2 \rfloor} \cdot \frac{1}{|\mathcal{X}|^b} \text{ ($U_\theta$ has real entries)} \\
        &= |\mathcal{X}|^{-\lfloor t/2 \rfloor - b} \cdot \left(\Tr{U_\theta \proj{y} U_\theta^\dag U_\theta \proj{y} U_\theta^\dag}\right)^{\lfloor t/2 \rfloor} \\
        &= |\mathcal{X}|^{-\lfloor t/2 \rfloor - b} \\
        &= |\mathcal{X}|^{\lfloor t/2 \rfloor - t}.
    \end{align*}
    In the above application of Corollary~\ref{cor:CIprojector}, we are taking $\vec{P} = \matA_y^\theta, \vec{Q} = \overline{\matA_y^\theta}$, and $\Phi$ to be the identity map.
    Now taking expectation over $\theta$ immediately yields Equation~\eqref{eq:opnormcounterexamplegoal}, as desired.
\end{proof}

\subsection{Limitations of Bounding Pairwise Overlaps}\label{sec:tfkwendoftheline}

Here, we show that the framework laid out by~\cite{Tomamichel_2013} of using Theorem \ref{thm:tfkwmain} to bound monogamy games cannot prove a better bound than $1/\sqrt{|\mathcal{X}|}$ (see the statement of~\Cref{thm:tfkwmain}).

\begin{lemma}\label{lemma:tfkwendoftheline}
    If $|\Theta| \geq 2$, then we have $$\underset{\substack{\theta, \theta' \in \Theta\\\theta \neq \theta'}}{\max}\text{ } \underset{x, x' \in \mathcal{X}}{\max} \norm{\sqrt{\matA_x^{\theta}} \sqrt{\matA_{x'}^{\theta'}}}_\infty = \underset{\substack{\theta, \theta' \in \Theta\\\theta \neq \theta'}}{\max}\text{ } \underset{x, x' \in \mathcal{X}}{\max} \norm{\matA_x^{\theta} \matA_{x'}^{\theta'}}_\infty \geq \frac{1}{\sqrt{|\mathcal{X}|}}.$$
\end{lemma}
\begin{proof}
    The first equality follows since the measurements are projective, so $\sqrt{\matA_x^\theta} = \matA_x^\theta$. Now for any distinct $\theta, \theta' \in \Theta$, we will show that $$\underset{x, x' \in \mathcal{X}}{\max} \norm{\matA_x^{\theta} \matA_{x'}^{\theta'}}_\infty \geq \frac{1}{\sqrt{|\mathcal{X}|}}.$$
    Indeed, fix any $x' \in \mathcal{X}$ such that $\matA_{x'}^{\theta'}$ is nonzero (such $x'$ exists since $\sum_{x'} \matA_{x'}^{\theta'} = \id$) and consider an arbitrary state $\ket{\psi}$ in the image of $\matA_{x'}^{\theta'}$. Then we have:
    \begin{align*}
        \ket{\psi} &= \sum_{x \in \mathcal{X}} \matA_x^\theta \ket{\psi} \\
        &= \sum_{x \in \mathcal{X}} \matA_x^\theta \matA_{x'}^{\theta'} \ket{\psi} \\
        \Rightarrow 1 &= \norm{\sum_{x \in \mathcal{X}} \matA_x^\theta \matA_{x'}^{\theta'} \ket{\psi}}_2^2.
    \end{align*}
    For each $x \in \mathcal{X}$, let $\ket{\psi_x} = \matA_x^\theta \matA_{x'}^{\theta'} \ket{\psi}$, where $\ket{\psi_x}$ may not be normalized. Note for any $x \neq y$ that $\braket{\psi_x|\psi_y} = \braket{\psi|\matA_{x'}^{\theta'} \matA_x^\theta \matA_y^\theta \matA_{x'}^{\theta'}|\psi} = 0$, since $\matA_x^\theta \matA_y^\theta = 0$. Hence the $\ket{\psi_x}$'s are mutually orthogonal, implying that:
    \begin{align*}
        1 &= \norm{\sum_{x \in \mathcal{X}} \ket{\psi_x}}_2^2 \\
        &= \sum_{x \in \mathcal{X}} \norm{\ket{\psi_x}}_2^2.
    \end{align*}
    Hence there exists $x \in \mathcal{X}$ such that $\norm{\ket{\psi_x}}_2^2 \geq 1/|\mathcal{X}| \Rightarrow \norm{\ket{\psi_x}}_2 \geq 1/\sqrt{|\mathcal{X}|}$. For this $x$, we have:
    \begin{align*}
        \frac{1}{\sqrt{|\mathcal{X}|}} &\leq \norm{\ket{\psi_x}}_2 \\
        &= \norm{\matA_x^\theta \matA_{x'}^{\theta'} \ket{\psi}}_2 \\
        &\leq \norm{\matA_x^\theta \matA_{x'}^{\theta'}}_\infty,
    \end{align*}
    as desired.
\end{proof}

\subsection{Monogamy Bounds from Binary Phase States}\label{sec:salting}

We just showed in~\Cref{sec:tfkwendoftheline} that using~\Cref{thm:tfkwmain} to bound a monogamy game with $\mathcal{X} = \bit^n$, we can only hope to prove a bound of $2^{-n/2}$. On the other hand, no existing analyses saturate this bound: the best known upper bound for the coset monogamy game is $O(2^{-n/4})$, due to~\cite{schleppy2025winning}\footnote{Previous work~\cite{schleppy2024optimal} proved an upper bound of $O(2^{-n/2})$ in a setting where the cloner $\Phi$ is restricted to splitting the state as is into two equal-sized halves, sending one to Bob and the other to Charlie. We cite $O(2^{-n/4})$ as the state of the art, as we are interested in games where the cloner $\Phi$ is unrestricted.}, and the BB84 monogamy game attains value $\approx 2^{-0.228n}$.

In this section, we close this gap. Specifically, we will show that a variant of our binary phase construction used in~\Cref{sec:binaryphaseconstruction} satisfies $\omega(\mathsf{G}) \leq 2^{-n/2 + o(n)}$, and that this can be shown following the~\cite{Tomamichel_2013} methodology and using~\Cref{thm:tfkwmain}. For details on how exactly our results in this section differ from those in~\Cref{sec:binaryphaseconstruction}, we refer the reader to~\Cref{remark:saltcomparison}.

\begin{lemma}\label{lemma:tfkwtoinnerproduct}
    Suppose we have $\matA_x^\theta = \mathsf{V}_\theta \proj{x} \mathsf{V}_\theta^\dag$ for some unitaries $\mathsf{V}_\theta$.
    Then for any $x, x' \in \mathcal{X}$ and $\theta, \theta' \in \Theta$, we have $$\norm{\sqrt{\matA_x^\theta} \sqrt{\matA_{x'}^{\theta'}}}_\infty = \norm{\matA_x^\theta \matA_{x'}^{\theta'}}_\infty = \left|\braket{x|\mathsf{V}_\theta^\dag \mathsf{V}_{\theta'}|x'}\right|.$$
\end{lemma}
\begin{proof}
    Note firstly that the $\ell_\infty$ norm of any rank 1 Hermitian PSD matrix is equal to its trace (see Lemma~\ref{lemma:opnormdef}). Bearing this in mind, we have:
    \begin{align*}
        \norm{\matA_x^\theta \matA_{x'}^{\theta'}}_\infty^2 &= \norm{\matA_x^\theta \matA_{x'}^{\theta'} \matA_{x'}^{\theta'\dag} \matA_x^{\theta\dag}}_\infty \\
        &= \norm{\mathsf{V}_\theta \proj{x} \mathsf{V}_{\theta}^\dag \cdot \mathsf{V}_{\theta'} \proj{x'} \mathsf{V}_{\theta'}^\dag \cdot \mathsf{V}_{\theta'} \proj{x'} \mathsf{V}_{\theta'}^\dag \cdot \mathsf{V}_\theta \proj{x} \mathsf{V}_{\theta}^\dag}_\infty \\
        &= \mathrm{Tr} \left[\mathsf{V}_\theta \proj{x} \mathsf{V}_{\theta}^\dag \cdot \mathsf{V}_{\theta'} \proj{x'} \mathsf{V}_{\theta'}^\dag \cdot \mathsf{V}_{\theta'} \proj{x'} \mathsf{V}_{\theta'}^\dag \cdot \mathsf{V}_\theta \proj{x} \mathsf{V}_{\theta}^\dag\right] \\
        &= \left|\braket{x|\mathsf{V}_\theta^\dag \mathsf{V}_{\theta'}|x'}\right|^2 \cdot \mathrm{Tr}\left[\mathsf{V}_\theta \ket{x} \bra{x'} \mathsf{V}_{\theta'}^\dag \cdot \mathsf{V}_{\theta'} \ket{x'}\bra{x} \mathsf{V}_{\theta}^\dag\right] \\
        &= \left|\braket{x|\mathsf{V}_\theta^\dag \mathsf{V}_{\theta'}|x'}\right|^2,
    \end{align*}
    which implies the conclusion.
\end{proof}

\noindent
We now describe our PRF-based construction. We stress that we only require the PRF to be secure against classical adversaries, and we do not assume that Bob and Charlie are computationally bounded. The reasons for this will become clear later, and are summarized in~\Cref{remark:classicalsecurityofprf}. Some other helpful remarks on this construction and its analysis can be found in~\Cref{remark:saltcomparison}.

\begin{construction}\label{construction:PRS-monogamy}
    Let $\mathfrak{F} = \left\{F_k: \left\{0, 1\right\}^{m+n} \rightarrow \bit\right\}_{k \in \bit^\secp}$ be a PRF family. (Here, $m := m(\secp)$ and $n := n(\secp)$ should be thought of as small polynomials in the security parameter $\secp$ e.g. $\secp^{0.1}$.) For any $k \in \bit^\secp$ and $\theta \in \bit^m$, define the ``salted'' function $f_{k, s}: \bit^n \rightarrow \bit$ by $f_{k, \theta}(x) = F_k(\theta || x)$.
    
    Then for any $k \in \bit^\secp$, we define the monogamy game $\mathsf{G}_{\mathfrak{F}, k}$ as follows:
    \begin{itemize}
        \item The Hilbert space $\mathcal{H}_{\reg A}$ is $\mathcal{C}^{2^n}$.
        \item The set of questions $\Theta$ is $\bit^m$.
        \item The set of answers $\mathcal{X}$ is $\bit^n$.
        \item For any $\theta \in \bit^m$ and $x \in \bit^n$, we define the following:
        \begin{align*}
            \mathsf{V}_\theta &= \mathsf{U}_{f_{k, \theta}} \mathsf{H}^{\otimes n},\text{ and} \\
            \vec{A}_x^\theta &= \mathsf{V}_\theta \proj{x} \mathsf{V}_\theta^\dag,
        \end{align*}
        where $\mathsf{U}_{f_{k, \theta}}$ is the phase unitary defined in Section \ref{sec:phaseunitarydefinition}.
    \end{itemize}
\end{construction}
\begin{remark}
    Our use of PRF security is already quite unconventional; the PRF key $k$ should be thought of here as a public parameter that is known to all parties (including Bob and Charlie) before the monogamy game commences.
\end{remark}

\begin{remark}
    At first glance, this ``salting'' construction appears unnatural; it would be much more natural to consider a PRF family $\left\{F_k: \bit^{n} \rightarrow \bit\right\}_{k \in \bit^\secp}$, set $\Theta = \bit^\lambda$, and set $\mathsf{V}_\theta = \mathsf{U}_{F_\theta}\mathsf{H}^{\otimes n}$ (where $\theta \in \bit^\lambda$ is the PRF key).

    The problem with this is that we will be analyzing this construction using~\Cref{thm:tfkwmain}, which considers the \emph{worst-case} overlap across different $\theta, \theta' \in \bit^\lambda$, while PRF security would only give us a ``with high probability'' guarantee with respect to $\theta$, which is insufficient for us.
    
    To remedy this, we salt the PRF so that the ``with high probability'' guarantee is absorbed into the setup phase of the construction. In other words, this can be thought of as a ``randomized monogamy game'' (where the new randomization occurs during the setup phase). We can now obtain the desired worst-case overlap bounds using simple concentration bounds, as we will see next. %
\end{remark}
\noindent
Our starting point to analyze Construction \ref{construction:PRS-monogamy} is the following lemma:

\begin{lemma}[Essentially \cite{booleanfunctionnotes}, Exercise 5.8]\label{lemma:randomfunctionoverlap}
    Let $F: \left\{0, 1\right\}^{m+n} \rightarrow \left\{-1, 1\right\}$ be a random function. For any $s \in \left\{0, 1\right\}^m$, define $f_s: \left\{0, 1\right\}^n \rightarrow \left\{-1, 1\right\}$ by $f_s(u) = F(s, u)$. Then with probability $1 - O(2^{-n})$ over the randomness of $F$, we have $$\max_{r \neq s} \max_{w \in \left\{0, 1\right\}^n} \left|\EE_u \left[(-1)^{\langle w, u \rangle} f_r(u) f_s(u) \right]\right| \leq 2 \cdot 2^{-n/2}\sqrt{m+n}.$$
\end{lemma}
\begin{proof}
    We will first argue for any fixed $r, s, w$ then take a union bound at the end. If we let $G(u) = f_r(u)f_s(u) = F(r, u)F(s, u)$, it is clear that $G$ is itself a random function from $\left\{0, 1\right\}^n \rightarrow \left\{-1, 1\right\}$ (noting that we get independence because $r \neq s$). Hence we just want to bound $$\left|\EE_u \left[(-1)^{\langle w, u \rangle} G(u)\right]\right|.$$
    For each $u$, $(-1)^{\langle w, u \rangle}G(u)$ is an independent and uniformly random sample from $\left\{-1, 1\right\}$, so this quantity can be bounded with a straightforward Chernoff bound. Indeed, Hoeffding's inequality tells us that:
    \begin{align*}
        \Pr\left[\left|\EE_u \left[(-1)^{\langle w, u \rangle} G(u)\right]\right| > 2 \cdot 2^{-n/2}\sqrt{m+n}\right] &\leq 2 \exp\left(\frac{-4 \cdot 2^n \cdot (m+n)}{2^{n+1}}\right) \\
        &= 2\exp(-2(m+n)).
    \end{align*}
    Taking a union bound over $2^m$ choices of $r$, $2^m$ choices of $s$, and $2^n$ choices of $w$ implies that:
    \begin{align*}
        \Pr\left[\max_{r \neq s} \max_{w \in \left\{0, 1\right\}^n} \left|\EE_u \left[(-1)^{\langle w, u \rangle} G(u)\right]\right| > 2 \cdot 2^{-n/2}\sqrt{m+n}\right] &\leq 2^{2m+n} \cdot 2\exp(-2(m+n)) \\ 
        &= O(2^{-n}).
    \end{align*}
\end{proof}

\begin{corollary}\label{cor:saltedphasestate}
    Assume that the PRF family $\mathfrak{F}$ is $(2^{m+n}, \epsilon(\secp))$-classically secure i.e. a classical distinguisher that runs in time $\poly(2^{m+n})$ can only distinguish a function sampled from $\mathfrak{F}$ from a truly random function with advantage $\leq \epsilon(\secp)$. Then with probability $1 - O(2^{-n}) - \epsilon(\secp)$ over the randomness of $k \gets \bit^\secp$, we have $$\max_{r \neq s} \max_{w \in \left\{0, 1\right\}^n} \left|\EE_u \left[(-1)^{\langle w, u \rangle + f_{k, r}(u) + f_{k, s}(u)} \right]\right| \leq 2 \cdot 2^{-n/2}\sqrt{m+n}.$$
\end{corollary}
\begin{proof}
    Consider the following PRF distinguisher given oracle access to some function $F$: it simply iterates over all $r, s, w, u$ and computes $$\max_{r \neq s} \max_{w \in \left\{0, 1\right\}^n} \left|\EE_u \left[(-1)^{\langle w, u \rangle + F(r, u) + F(s, u)} \right]\right|,$$and outputs 1 if the result is $> 2 \cdot 2^{-n/2}\sqrt{m+n}$. This distinguisher runs in time $\poly(2^{m+n})$. By Lemma~\ref{lemma:randomfunctionoverlap}, it outputs 1 given a random function with probability at most $O(2^{-n})$ (noting that the outputs of $F$ are in $\bit$, so the outputs of $(-1)^{F(\cdot)}$ are in $\left\{-1, 1\right\}$ as in Lemma~\ref{lemma:randomfunctionoverlap}). Hence by PRF security, it outputs 1 given a function sampled from $\mathfrak{F}$ with probability at most $O(2^{-n}) + \epsilon(\secp)$. The conclusion follows.
\end{proof}

Note that a PRF with this security guarantee can be instantiated assuming sub-exponentially secure PRFs since $\secp$ is a large polynomial in $m+n$. With this corollary, we can prove an upper bound on the value of our monogamy game:

\begin{theorem}\label{thm:saltedphasestatemain}
    Assume (as in Corollary \ref{cor:saltedphasestate}) that the PRF family $\mathfrak{F}$ is $(2^{m+n}, \epsilon(\secp))$-classically secure i.e. a classical distinguisher that runs in time $\poly(2^{m+n})$ can only distinguish a function sampled from $\mathfrak{F}$ from a truly random function with advantage $\leq \epsilon(\secp)$. Then with probability $1 - O(2^{-n}) - \epsilon(\secp)$ over the randomness of $k \gets \bit^\secp$, we have:
    $$\omega(\mathsf{G}_{\mathfrak{F}, k}) \leq O(2^{-m} + 2^{-n/2} \sqrt{m+n}).$$
\end{theorem}
\begin{proof}
    By Theorem \ref{thm:tfkwmain} and the analysis preceding it, we have:
    \begin{align*}
        \omega(\mathsf{G}_{\mathfrak{F}, k}) &\leq 2^{-m} + (1 - 2^{-m}) \cdot \underset{\substack{\theta, \theta' \in \bit^m\\\theta \neq \theta'}}{\max}\text{ } \underset{x, x' \in \bit^n}{\max} \norm{\sqrt{\vec A_x^{\theta}} \sqrt{\vec A_{x'}^{\theta'}}}_\infty.
    \end{align*}
    Hence it suffices to show that $$\underset{\substack{\theta, \theta' \in \bit^m\\\theta \neq \theta'}}{\max}\text{ } \underset{x, x' \in \bit^n}{\max} \norm{\sqrt{\vec A_x^{\theta}} \sqrt{\vec A_{x'}^{\theta'}}}_\infty \leq O(2^{-n/2}\sqrt{m+n}).$$
    Indeed, we have:
    \begin{align*}
        \norm{\sqrt{\vec A_x^{\theta}} \sqrt{\vec A_{x'}^{\theta'}}}_\infty &= \left|\braket{x|\mathsf{V}_\theta^\dag \mathsf{V}_{\theta'}|x'}\right| \text{ (Lemma \ref{lemma:tfkwtoinnerproduct})} \\
        &= \left|\braket{x|\mathsf{H}^{\otimes n} \mathsf{U}_{f_{k, \theta}} \mathsf{U}_{f_{k, \theta'}} \mathsf{H}^{\otimes n}|x'}\right| \\
        &= \frac{1}{2^n} \left|\sum_{y, y' \in \bit^n} (-1)^{\langle x, y \rangle + \langle x', y' \rangle} \braket{y|\mathsf{U}_{f_{k, \theta}} \mathsf{U}_{f_{k, \theta'}}|y'} \right| \\
        &= \frac{1}{2^n} \left|\sum_{y, y' \in \bit^n} (-1)^{\langle x, y \rangle + \langle x', y' \rangle + f_{k, \theta}(y) + f_{k, \theta'}(y')} \braket{y|y'} \right| \\
        &= \left|\underset{y \gets \bit^n}{\EE} \left[(-1)^{\langle x+x', y \rangle + f_{k, \theta}(y) + f_{k, \theta'}(y)}\right]\right| \\
        \Rightarrow \underset{\substack{\theta, \theta' \in \bit^m\\\theta \neq \theta'}}{\max}\text{ } \underset{x, x' \in \bit^n}{\max} \norm{\sqrt{\vec A_x^{\theta}} \sqrt{\vec A_{x'}^{\theta'}}}_\infty &= \underset{\substack{\theta, \theta' \in \bit^m\\\theta \neq \theta'}}{\max}\text{ } \underset{x, x' \in \bit^n}{\max}\left|\underset{y \gets \bit^n}{\EE} \left[(-1)^{\langle x+x', y \rangle + f_{k, \theta}(y) + f_{k, \theta'}(y)}\right]\right| \\
        &\leq O(2^{-n/2} \sqrt{m+n}),
    \end{align*}
    with probability at least $1 - O(2^{-n}) - \epsilon(\secp)$ over the randomness of $k$ by Corollary \ref{cor:saltedphasestate} (noting that we can consolidate the $\max$ over $x$ and $x'$ into a single $\max$ over $w := x + x'$).
\end{proof}

\begin{remark}\label{remark:classicalsecurityofprf}
    Our use of PRF security is only to prove the concentration bound in Corollary \ref{cor:saltedphasestate}, and hence we only need security against classical adversaries. Once we have this bound, we are applying Theorem \ref{thm:saltedphasestatemain} which holds against computationally unbounded Bob and Charlie. Therefore, although our construction is based on a cryptographic assumption, it is secure even against computationally unbounded adversaries Bob and Charlie.
\end{remark}

\begin{remark}\label{remark:saltcomparison}
    We make some other comparisons between this construction and other constructions:
    \begin{itemize}
        \item Compared with the very similar construction (Construction~\ref{const:PRS-enc}) based on binary phase states in~\Cref{sec:binaryphaseconstruction}:
        \begin{itemize}
            \item Theorem~\ref{thm:saltedphasestatemain} only establishes a bound of $\tilde{O}(2^{-n/2})$ in the single-copy setting. In comparison, in~\Cref{sec:binaryphaseconstruction} we show a bound of $O_t(2^{-n})$ in the $t$-copy setting.

            \item On the other hand, Theorem~\ref{thm:saltedphasestatemain} holds in a much stronger attack model; the players Bob and Charlie are given the basis $\theta$ in the clear and are computationally unbounded. On the other hand, in~\Cref{sec:binaryphaseconstruction} we restrict the players to each make a single query to $\mathsf{U}_f$.
        \end{itemize}

        \item Compared with the BB84~\cite{Tomamichel_2013, broadbent_et_al:LIPIcs.TQC.2020.4} and coset state~\cite{10.1007/978-3-030-84242-0_20, Culf_2022, schleppy2025winning} constructions, we show a better bound in the same attack model. However, we need to use a cryptographic assumption (sub-exponentially classically secure PRFs).
    \end{itemize}
\end{remark}

%% file: worsttoav.tex
\section{Worst-Case to Average-Case Reduction}\label{sec:worsttoav}

In this section, we show that $t \mapsto t+1$ oracular cloning games admit a worst-case to average-case reduction: even the hardest games which are specified by some worst-case unitary $U_w$
can be won by a strategy for the average-case version of the game that involves a Haar-like unitary $U_a$ from an appropriate unitary design, or alternatively from a pseudorandom unitary ensemble. 

We begin with a technical background section on mixed unitary designs; equipped with this machinery, we can then complete the proof of our worst-case to average-case reduction.

\subsection{Preliminary: Mixed Unitary Designs}\label{sec:mixedunitarydesigns}

We first review some relevant background on the vectorization technique.

\paragraph{Vectorization Formalism.} For a linear operator $\boldsymbol{\Lambda} \in \mathrm{L}(\mathbb{C}^d)$, we consider the corresponding vectorization map $\mathrm{vec}: \mathrm{L}(\mathbb{C}^d) \rightarrow (\mathbb{C}^d)^{\ot 2}$ which is defined as follows:

$$
\boldsymbol{\Lambda}  = \sum_{i,j \in [d]} \Lambda_{(i,j)} \ket{i}\hspace{-1mm}\bra{j} \quad \mapsto \quad \mathrm{vec}(\boldsymbol{\Lambda} ) :=\kett{\boldsymbol{\Lambda} } = \sum_{i,j \in [d]} \Lambda_{(i,j)} \ket{i} \otimes \ket{j}.
$$
We are also going to use the so-called ABC-rule~\cite{Mele2024introductiontohaar}: for any linear operators $\vec{A,B,C} \in \mathrm{L}(\mathbb{C}^d)$,
$$
\kett{\vec{ABC}} = (\vec A \otimes \vec C^\intercal) \kett{\vec B}.
$$

We now introduce a 
"mixed variant" of the regular vectorized moment operator~\cite{Mele2024introductiontohaar}.

\begin{definition}[Mixed-Adjoint Moment Operator]
Let $\nu$ be an ensemble of unitary operators over $\mathbb{C}^d$. Then, we define the mixed-adjoint $(p,q)$-moment operator $\algo M_{\nu, \mathsf{adj}}^{(p,q)} : \mathrm{L}(\mathbb{C}^d) \rightarrow \mathrm{L}(\mathbb{C}^d)$ by
$$
\algo M_{\nu, \mathsf{adj}}^{(p,q)}(\vec O) := \underset{U \sim \nu}{\mathbb{E}}\left[(U^{\otimes p} \otimes (U^\dag)^{\otimes q}) \vec O (U^{\otimes p} \otimes  (U^\dag)^{\otimes q})^\dag \right] 
$$
for a linear operator $\vec O \in \linear((\CC^d)^{\otimes (p+q)})$. Similarly, we let $\algo M_{\mathrm{U}(d), \mathsf{adj}}^{(p,q)}$ denote the mixed-adjoint $(p,q)$-moment operator with respect to the Haar measure over the unitary group $\algo{U}_d$.
\end{definition}

It turns out that the mixed-adjoint moment operator allows for a particularly neat characterization of mixed unitary designs introduced in~\Cref{sec:unitary-designs}.
In particular, we show the following equivalence.

\begin{lemma}\label{lem:equiv-mixed-design}
A unitary ensemble $\nu$ over $\mathbb{C}^d$ is a (non-adaptive) unitary $(p,q)$-design if and only if
$$
\underset{U \sim \nu}{\mathbb{E}}\left[U^{\otimes p} \otimes (U^\dag)^{\otimes q}  \otimes  \bar{U}^{\otimes p} \otimes  (U^\intercal)^{\otimes q} \right] = \underset{U \sim \mathrm{U}(d)}{\mathbb{E}}\left[U^{\otimes p} \otimes (U^\dag)^{\otimes q}  \otimes  \bar{U}^{\otimes p} \otimes  (U^\intercal)^{\otimes q}\right].
$$
\end{lemma}
\begin{proof} Suppose that $\nu$ is a unitary $(p,q)$-design. Then, for all $\vec O \in \linear((\CC^d)^{\otimes (p+q)})$, it holds that
$$
\algo M_{\nu,\mathsf{adj}}^{(p,q)}(\vec O) =\algo M_{\mathrm{U}(d),\mathsf{adj}}^{(p,q)}(\vec O).
$$
By applying the vectorization $\mathrm{vec}: \mathrm{L}((\CC^d)^{\otimes (p+q)}) \rightarrow ((\CC^d)^{\otimes (p+q)})^{\ot 2}$  on both sides, we get
$$
\kett{\algo M_{\nu,\mathsf{adj}}^{(p,q)}(\vec O)} = \kett{\algo M_{\mathrm{U}(d),\mathsf{adj}}^{(p,q)}(\vec O)}.
$$
By linearity and the $ABC$-rule for $\mathrm{vec}(\cdot)$, this is equivalent to
$$
\underset{U \sim \nu}{\mathbb{E}}\left[U^{\otimes p} \otimes (U^\dag)^{\otimes q}  \otimes  \bar{U}^{\otimes p} \otimes  (U^\intercal)^{\otimes q} \right]\kett{\vec O} = \underset{U \sim \mathrm{U}(d)}{\mathbb{E}}\left[U^{\otimes p} \otimes (U^\dag)^{\otimes q}  \otimes  \bar{U}^{\otimes p} \otimes  (U^\intercal)^{\otimes q} \right] \kett{\vec O}.
$$
Because $\mathrm{vec}(\cdot)$ is a bijection between $\mathrm{L}((\CC^d)^{\otimes (p+q)}) $ and $((\CC^d)^{\otimes (p+q)})^{\ot 2}$, the operators above must be identical on the entire vector space $((\CC^d)^{\otimes (p+q)})^{\ot 2}$. The converse statement can be shown analogously.
\end{proof}

\begin{lemma}\label{lem:t-design-to-mixed}
A unitary $t$-design $\nu$ is a mixed unitary $(p,q)$-design for any $p,q$ with $t=p+q$.
\end{lemma}
\begin{proof}
Let $t=p+q$. According to \Cref{lem:equiv-mixed-design}, it suffices to show that $\nu$ satisfies
$$
\underset{U \sim \nu}{\mathbb{E}}\left[U^{\otimes p} \otimes (U^\dag)^{\otimes q}  \otimes  \bar{U}^{\otimes p} \otimes  (U^\intercal)^{\otimes q}\right] = \underset{U \sim \mathrm{U}(d)}{\mathbb{E}}\left[U^{\otimes p} \otimes (U^\dag)^{\otimes q}  \otimes  \bar{U}^{\otimes p} \otimes  (U^\intercal)^{\otimes q} \right].
$$
By inserting the partial transpose with respect to the 2nd and 4th system, this is equivalent to
$$
\underset{U \sim \nu}{\mathbb{E}}\left[U^{\otimes p} \otimes \bar{U}^{\otimes q}  \otimes  \bar{U}^{\otimes q} 
\otimes U^{\otimes p} \right]^{\mathsf{T}_{2,4}} = \underset{U \sim \mathrm{U}(d)}{\mathbb{E}}\left[U^{\otimes p} \otimes \bar{U}^{\otimes q}  \otimes  \bar{U}^{\otimes q} 
\otimes U^{\otimes p}\right]^{\mathsf{T}_{2,4}}.
$$
After inserting a $\mathsf{SWAP}$ between the 2nd and 4th system via $\mathbb{F}_{2,4}$, it is also equivalent to showing that
$$
\left[\mathbb{F}_{2,4}^\dag\underset{U \sim \nu}{\mathbb{E}}\left[U^{\otimes p}  \otimes  U^{\otimes q} \otimes  \bar{U}^{\otimes p} \otimes \bar{U}^{\otimes q}  \right]\mathbb{F}_{2,4}\right]^{\mathsf{T}_{2,4}}= \left[\mathbb{F}_{2,4}^\dag\underset{U \sim \mathrm{U}(d)}{\mathbb{E}}\left[U^{\otimes p}  \otimes  U^{\otimes q} \otimes  \bar{U}^{\otimes p} \otimes \bar{U}^{\otimes q}  \right]\mathbb{F}_{2,4}\right]^{\mathsf{T}_{2,4}}.
$$
By assumption, $\nu$ is a unitary $t$-design for $t=p+q$, and hence it holds that 
$$
\underset{U \sim \nu}{\mathbb{E}}\left[U^{\otimes p}  \otimes  U^{\otimes q} \otimes  \bar{U}^{\otimes p} \otimes \bar{U}^{\otimes q}  \right] = \underset{U \sim \mathrm{U}(d)}{\mathbb{E}}\left[U^{\otimes p}  \otimes  U^{\otimes q} \otimes  \bar{U}^{\otimes p} \otimes \bar{U}^{\otimes q}  \right]
$$
which yields the desired equality from before.
\end{proof}

Finally, we show that an (exact) non-adaptive mixed unitary $t$-design is automatically also an (exact) adaptive mixed unitary $t$-design. In the approximate case, this conversion incurs an exponential blow-up.

\begin{theorem}\label{thm:exact-to-adative}
Any exact non-adaptive unitary $t$-design is also an exact adaptive mixed unitary $t$-design.    
\end{theorem}
\begin{proof} Let $\nu$ be a non-adaptive unitary $t$-design over $\mathbb{C}^d$. Suppose that $\algo A^{U,U^\dag}(1^{\lceil \log d \rceil})$ is an adaptive $t$-query quantum algorithm that makes $p$ many queries to $U$ and $q$ queries to $U^\dag$, for $U \in \mathrm{U}(d)$ and $t=p+q$.

The idea is to use a standard gate teleportation approach, similar to~\cite{cryptoeprint:2019/1204,DBLP:conf/tqc/Kretschmer21}.
Concretely, we can argue that, for any $U \in \mathrm{U}(d)$, there exists a non-adaptive algorithm $\algo B^{U,U^\dag}(1^{\lceil \log d \rceil})$ that makes $p$ many parallel queries to $U$ and $q$ many parallel queries to $U^\dag$ such that
$$
\Pr\left[ 1 \leftarrow \algo A^{U,U^\dag }(1^{\lceil \log d \rceil})  \right] = d^{2(p+q)} \Pr\left[ 1 \leftarrow \algo B^{U,U^\dag }(1^{\lceil \log d \rceil}) \right].
$$
This essentially follows from~\cite[Lemma 23]{DBLP:conf/tqc/Kretschmer21}, since the non-adaptive query algorithm has access to both $U$ and $U^\dag$.
Because $\nu$ is a non-adaptive unitary $t$-design, we know from \Cref{lem:t-design-to-mixed} that $\nu$ is also mixed unitary $(p,q)$-design. Putting everything together, we get that
\begin{align*}
\Pr\left[ 1 \leftarrow \algo A^{U,U^\dag }(1^{\lceil \log d \rceil}) \, : \, U \sim \nu \right] &= d^{2(p+q)}\Pr\left[ 1 \leftarrow \algo B^{U,U^\dag }(1^{\lceil \log d \rceil}) \, : \, U \sim \nu \right]\\
&= d^{2(p+q)}\Pr\left[ 1 \leftarrow \algo B^{U,U^\dag }(1^{\lceil \log d \rceil}) \, : \, U \sim \mathrm{U}(d) \right]\\
&=\Pr\left[ 1 \leftarrow \algo A^{U,U^\dag }(1^{\lceil \log d \rceil}) \, : \, U \sim \mathrm{U}(d) \right].    
\end{align*}
This proves the claim.
\end{proof}


\subsection{Proof of the Reduction}

Using the technical machinery on mixed unitary designs from the previous section, we can finally prove our worst-case to average-case reduction for cloning games.

\begin{theorem}[Worst-Case to Average-Case Reduction]\label{thm:worst-case-to-average-case}
Let $n,t \in \N$ and let $\nu =\{U_a\}_{a \in \Theta}$ be an ensemble of $n$-qubit unitaries to be specified later. Suppose there exists a quantum strategy 
$$
\mathsf{S}^{\mathsf{avg}} = (\algo H_{\reg{B^{t+1}}}, \Phi_{\reg{A^t \rightarrow B^{t+1}}},\{\vec{P}_{1,x}^{U_a,U_{a}^\dag}\}_{a \in \Theta, \, x \in \bit^n}\,,\,\dots,\{\vec{P}_{t+1,x}^{U_a,U_{a}^\dag}\}_{a \in \Theta, \,x \in \bit^n})$$
for the average-case $t\mapsto t+1$ oracular cloning game $\mathsf{G}_{t \mapsto t+1}^{\mathsf{avg}} = (t,\algo H_{\reg{A^t}},\Theta,\bit^n, \{U_a\}_{a \in \Theta})$, where the $t+1$ players make no more than a total of $q$ many oracle queries to either $U_a$ or $U_a^\dag$ combined, and
$$
\omega_{\mathsf{S}^{\mathsf{avg}}}(\mathsf{G}_{t \mapsto t+1}^{\mathsf{avg}}) = \epsilon.
$$
Then, there exists a quantum strategy (in which the $t+1$ many players make the same number of queries) 
$$
\mathsf{S}^{\mathsf{wst}} = (\algo H_{\reg{\tilde{B}^{t+1}}}, \tilde{\Phi}_{\reg{A^t \rightarrow \tilde{B}^{t+1}}},\{\tilde{\vec{P}}_{1,x}^{V_w,V_{w}^{\dag}}\}_{x \in \bit^n}\,,\,\dots,\{\tilde{\vec{P}}_{t+1,x}^{V_w,V_{w}^{\dag}}\}_{x \in \bit^n})$$
for any $t\mapsto t+1$ oracular cloning game $\mathsf{G}_{t \mapsto t+1}^{\mathsf{wst}} = (t,\algo H_{\reg{A^t}},\Theta',\bit^n, \{V_w\}_{w \in \Theta'})$ in the worst-case (i.e., for any adversarially chosen ensemble of $n$-qubit unitaries $\left\{V_w\right\}_{w \in \Theta'}$), such that:%
\begin{itemize}
    \item If $\nu$ is an exact unitary $r$-design, for $r=t+q$ and $q\in \N$: we will have $$\omega_{\mathsf{S}^{\mathsf{wst}}}(\mathsf{G}_{t \mapsto t+1}^{\mathsf{wst}}) = \epsilon.$$
    \item If $\nu$ is a pseudorandom unitary family with security parameter $\secp$, and the adversaries $(\Phi, \algo P_1, \ldots, \algo P_{t+1})$ are computationally bounded: we will have
    $$\omega_{\mathsf{S}^{\mathsf{wst}}}(\mathsf{G}_{t \mapsto t+1}^{\mathsf{wst}}) \geq \epsilon - \negl(\secp).$$
\end{itemize}

\end{theorem}
\begin{proof}
Let $\mathsf{G}_{t \mapsto t+1}^{\mathsf{wst}} = (t,\algo H_{\reg{A^t}},\Theta',\bit^n, \{V_w\}_{w \in \Theta'})$ be a worst-case $t \mapsto t+1$ oracular cloning game for ensemble of unitaries $\left\{V_w\right\}$. Consider the quantum strategy $\mathsf{S}^{\mathsf{wst}}$ for the game $\mathsf{G}_{t \mapsto t+1}^{\mathsf{wst}}$ which internally uses $\mathsf{S}^{\mathsf{avg}}$ and proceeds as follows:
\begin{itemize}
    \item (Cloning Channel:) on input $(V_w\ket{x})^{\otimes t}$ in register $\reg{A}^t$, the channel $\tilde{\Phi}_{\reg{A^t \rightarrow \tilde{B}^{t+1}}}$ proceeds as follows:
    \begin{enumerate}
        \item Sample a uniformly random unitary $U_a \sim \nu$ from the unitary ensemble $\nu$. 

        \item Apply $U_a$ to each copy of $V_w\ket{x}$, resulting in a state $(U_aV_w\ket{x})^{\otimes t}$ in register $\reg{A}^t$.

        \item Run ${\Phi}_{\reg{A^t \rightarrow {B}^{t+1}}}$ on $(U_aV_w\ket{x})^{\otimes t}$, and let $\reg{B^{t+1}} = \reg{B_1} \cdots \reg{B_{t+1}}$ denote the resulting registers.

        \item Output $\reg{\tilde{B}^{t+1}} := \reg{\tilde{B}_1} \cdots \reg{\tilde{B}_{t+1}}$, where $\reg{\tilde{B}_i}$ consists of $\reg{B_i}$ together with $\reg{B_i'}=\proj{a}$, for $i \in [t+1]$. 
    \end{enumerate}

    \item ($i$-th Player:) On input $\reg{\tilde{B}_i}$, the measurement $\{\tilde{\vec{P}}_{i,x}^{V_w,V_{w}^\dag}\}_{x \in \bit^n}$ proceeds as follows:   
    \begin{enumerate}
        \item Parse $\reg{\tilde{B}_i}$ as $\reg{B_i B_i'}$. Measure $\reg{B_i'}$ to obtain the string $a$.

        \item Run the oracle-aided measurement $\{\vec{P}_{i,x}^{U,U^\dag}\}_{\,x \in \bit^n}$ with respect to $U:=U_a V_w$ in such a way that, whenever one of the $q$-many oracle queries to $U$ or $U^\dag$ is submitted:
        \begin{itemize}
            \item If the query is to $U$: the query is first submitted to the available oracle $V_w$, then the unitary $U_a$ is applied to the resulting outcome.

            \item If the query is to $U^\dag$: the unitary $U_a^\dag$ is applied, and the resulting outcome submitted to the available oracle $V_w^\dag$.
        \end{itemize}
    \end{enumerate}
\end{itemize}
Let us now analyze the success probability $\omega_{\mathsf{S}^{\mathsf{wst}}}(\mathsf{G}_{t \mapsto t+1}^{\mathsf{wst}})$ of the strategy $\mathsf{S}^{\mathsf{wst}}$. We first address the case where $\nu$ is an exact unitary $r$-design. Recall from \Cref{thm:exact-to-adative} that any exact (non-adaptive) unitary $r$-design is also an exact adaptive $r$-design. Therefore, using that $\nu$ is an adaptive unitary $r$-design, for $r=t+q$, as well as the right-invariance of the Haar measure over the unitary group $\mathrm{U}(2^n)$, we get:
\begin{align*}
&\underset{\substack{x \sim \bit^n \\ w \sim \Theta'}}{\mathbb{E}}  \, \Tr{\left( \tilde{\vec{P}}_{1,x}^{V_w,V_{w}^\dag}\ot \dots \ot \tilde{\vec{P}}_{t+1,x}^{V_w,V_{w}^\dag}\right) \tilde{\Phi}_{\reg{A^t \rightarrow \tilde{B}^{t+1}}}\left( (V_w \proj{x}V_w^\dag)_{\reg{A^t}}^{\otimes t} \right)}\\
&=\underset{U_a \sim \nu}{\mathbb{E}} \,\underset{\substack{x \sim \bit^n \\ w \sim \Theta'}}{\mathbb{E}}\Tr{\left( \vec{P}_{1,x}^{U_aV_w,(U_aV_{w})^\dag}\ot \dots \ot \vec{P}_{t+1,x}^{U_aV_w,(U_aV_{w})^\dag}\right) {\Phi}_{\reg{A^t \rightarrow {B}^{t+1}}}\left( (U_aV_w \proj{x}(U_aV_w)^\dag)_{\reg{A^t}}^{\otimes t} \right)}\\
&=\underset{W \sim \mathrm{U}(2^n)}{\mathbb{E}} \,\underset{\substack{x \sim \bit^n \\ w \sim \Theta'}}{\mathbb{E}}\Tr{\left( \vec{P}_{1,x}^{WU_w,(WU_{w})^\dag}\ot \dots \ot \vec{P}_{t+1,x}^{WU_w,(WU_{w})^\dag}\right) {\Phi}_{\reg{A^t \rightarrow {B}^{t+1}}}\left( (WV_w \proj{x}(WV_w)^\dag)_{\reg{A^t}}^{\otimes t} \right)}\\
&=\underset{U \sim \mathrm{U}(2^n)}{\mathbb{E}} \,\underset{x \sim \bit^n}{\mathbb{E}}\Tr{\left( \vec{P}_{1,x}^{U,U^\dag}\ot \dots \ot \vec{P}_{t+1,x}^{U,U^\dag}\right) {\Phi}_{\reg{A^t \rightarrow {B}^{t+1}}}\left( (U \proj{x}U^\dag)_{\reg{A^t}}^{\otimes t} \right)}\\
&=\underset{U_a \sim \nu}{\mathbb{E}} \,\underset{x \sim \bit^n}{\mathbb{E}}\Tr{\left( \vec{P}_{1,x}^{U_a,U_a^\dag}\ot \dots \ot \vec{P}_{t+1,x}^{U_a,U_a^\dag}\right) {\Phi}_{\reg{A^t \rightarrow {B}^{t+1}}}\left( (U_a \proj{x}U_a^\dag)_{\reg{A^t}}^{\otimes t} \right)}.
\end{align*}
Therefore, we get that
$\omega_{\mathsf{S}^{\mathsf{wst}}}(\mathsf{G}_{t \mapsto t+1}^{\mathsf{wst}}) = \omega_{\mathsf{S}^{\mathsf{avg}}}(\mathsf{G}_{t \mapsto t+1}^{\mathsf{avg}})=\epsilon$, which proves the claim. In the case that $\nu$ is a PRU family and the adversary is computationally bounded, we can apply essentially the same calculation, but we will incur a potential additive loss of $\negl(\secp)$ when passing from $U_a$ to the Haar measure. (Note that we can appeal to PRU security because the strategy $\mathsf{S}^{\mathsf{wst}}$ can be simulated using only oracle access to $U_a, U_a^\dag$.)
\end{proof}

%% file: black-hole-games.tex
\section{Black Hole Cloning Games}\label{sec:black-hole-games}


Hayden and Preskill~\cite{Hayden_Preskill_2007} put forward the idea that the dynamics of a black hole are well-described by a random unitary time-evolution operator, e.g., via a \emph{unitary design}. Does such a scrambling process limit the extent to which black holes clone information? In this section, we seek to give a quantitative answer to this question. We formally state and discuss our model for this problem in terms of \emph{black hole cloning games} in~\Cref{sec:blackholedefs}, and then turn to proving a bound on the value of black hole cloning games in~\Cref{sec:blackholebounds}.

\subsection{Definition}\label{sec:blackholedefs}

Inspired by the monogamy game of Tomamichel,
Fehr, Kaniewski and Wehner~\cite{Tomamichel_2013}, we formalize the notion of a \emph{black hole cloning game} as follows:

\begin{definition}[Black Hole Cloning Game]\label{def:black-hole-game} A black hole cloning game is specified by a tuple of the form $\mathsf{G}_{\text{\tiny{BH}}} = (\algo H_{\reg{I}},\algo H_{\reg{B}},\algo H_{\reg{B'}},\algo H_{\reg{H}},\algo H_{\reg{R}},\Theta, \{U_{\theta}\}_{\theta \in \Theta},\Phi_{\reg{IB' \rightarrow HR}})$ and consists of the following elements:
\begin{itemize}
    \item A finite dimensional Hilbert space $\algo H_{\reg{I}}$ associated with the internal degrees of the freedom of the black hole; in particular, where $\algo H_{\reg{I}}$ contains the $(n-k)$-qubit initial state of the black hole;

    \item A pair of isomorphic finite dimensional Hilbert spaces $\algo H_{\reg{B}}$ and $\algo H_{\reg{B}'}$ which are associated with $k$-qubit $\mathsf{EPR}$ pairs that emerge near the boundary of the black hole;

    \item A finite dimensional Hilbert space $\algo H_{\reg{H}}$ associated with the final state of the black hole that comprises all of the qubits at its event horizon;

    \item A finite dimensional Hilbert space $\algo H_{\reg{R}}$ associated with the emitted Hawking radiation;

  \item A finite set of indices $\Theta$ over the set of all possible scrambling unitaries;

\item A finite ensemble of scrambling unitaries
$\{U_{\theta}^\dag\}_{\theta \in \Theta}$ indexed by $\Theta$ which is associated with the internal time-evolution of the black hole within its event horizon;

    \item A completely positive and trace-preserving channel $\Phi_{\reg{IB' \rightarrow HR}}$ associated with the physical process that maps the internal registers $\reg{IB'}$ of the black hole into a final internal register $\reg{H}$ and a register $\reg{R}$ associated with the emitted Hawking radiation.
\end{itemize}
\end{definition}

We remark that the modeling assumptions behind our \emph{black hole cloning game} (see \Cref{fig:black-hole-game}) appear consistent with the postulates of black hole complementarity~\cite{PhysRevD.48.3743,THOOFT1985727}; all of the components of our game are modeled according to the existing understanding of physics and Hawking radiation:
\begin{itemize}
    \item we model the entire process of black hole evolution as a (possibly unitary) quantum channel which takes as input the set of qubits belonging to the interior---together with Alice's infalling qubits---and converts them into a global (possibly pure) state of which a subsystem constitutes outgoing radiation;

    \item we assume that Hawking radiation is a valid phenomenon---it enables Bob to intercept outgoing radiation in the form of qubits that he can process on his quantum computer; meanwhile, Charlie is simply another observer that is anchored at the event horizon and has access to the remaining qubits; and

    \item we assume that Bob and Charlie have knowledge of the internal dynamics of the black hole, say as the result of statistical mechanical and thermodynamic considerations---in analogy to how deciphering the contents of a burning book is possible, at least in principle, by observing its smoke and ashes.\footnote{This analogy was also used in the work of Hayden and Preskill~\cite{Hayden_Preskill_2007}.} 
\end{itemize}
Therefore, we believe that black hole cloning games offer a reasonable characterization of quantum cloning in the context of evaporating black holes. In \Cref{sec:openquestions}, we discuss further improvements to our modeling assumptions which could potentially make our game even more realistic from a physical perspective.

\begin{definition}[Quantum Strategy]
A quantum strategy $\mathsf{S} = (\{\vec{H}_{x}^{U_\theta,U_{\theta}^\dag}\}_{\theta \in \Theta, x \in \bit^k},\{\vec{R}_{x}^{U_\theta,U_{\theta}^\dag}\}_{\theta \in \Theta, x \in \bit^k})$ for a black hole cloning game $\mathsf{G}_{\text{\tiny{BH}}} = (\algo H_{\reg{I}},\algo H_{\reg{B}},\algo H_{\reg{B'}},\algo H_{\reg{H}},\algo H_{\reg{R}},\Theta, \{U_{\theta}\}_{\theta \in \Theta},\Phi_{\reg{IB' \rightarrow HR}})$ consists of
\begin{itemize}
 \item An ensemble of oracle-aided positive operator-valued measurements $$\left\{\vec{H}_{x}^{U_\theta,U_{\theta}^\dag}\right\}_{\theta \in \Theta, x \in \bit^k}$$ which are to be performed on Charlie's system $\algo H_{\reg{H}}$.

  \item An ensemble of oracle-aided positive operator-valued measurements $$\left\{\vec{R}_{x}^{U_\theta,U_{\theta}^\dag}\right\}_{\theta \in \Theta, x \in \bit^k}$$ which are to be performed on Bob's system $\algo H_{\reg{R}}$.
\end{itemize}
\end{definition}

Next, we define the value of a black hole cloning game, which can be thought of as the maximal winning probability over all admissible strategies.

\begin{definition}[Value of a Black Hole Cloning Game]\label{def:bh-value}
Consider a black hole cloning game of the form
$\mathsf{G}_{\text{\tiny{BH}}} = (\algo H_{\reg{I}},\algo H_{\reg{B}},\algo H_{\reg{B'}},\algo H_{\reg{H}},\algo H_{\reg{R}},\Theta, \{U_{\theta}\}_{\theta \in \Theta},\Phi_{\reg{IB' \rightarrow HR}})$. 
Then, the winning probability of a quantum strategy $\mathsf{S} = (\{\vec{H}_{x}^{U_\theta,U_{\theta}^\dag}\}_{\theta \in \Theta, x \in \bit^k},\{\vec{R}_{x}^{U_\theta,U_{\theta}^\dag}\}_{\theta \in \Theta, x \in \bit^k})$ for $\mathsf{G}_{\text{\tiny{BH}}}$ is defined by the quantity
\begin{align*}
\omega_{\mathsf{S}}(\mathsf{G}_{\text{\tiny{BH}}}) := 
&\underset{\theta\sim \Theta}{\mathbb{E}} \Bigg\{ \sum_{x \in \bit^k}\mathrm{Tr}\Bigg[ \left(\vec{H}_x^{U_{\theta},U_{\theta}^\dag} \otimes \vec{R}_x^{U_{\theta},U_{\theta}^\dag} \otimes\proj{x}_{\reg{B}} \right) \Big(\Phi_{\reg{IB' \rightarrow HR}} \otimes \id_{\reg{B}} \Big)  \\
& \quad\quad\,\,\left(\big(U_{\theta} \fullcirc U_{\theta}^\dag\big)_{\reg{IB' \rightarrow IB'}} \otimes \id_{\reg{B}} \right)\Big( \proj{0^{n-k}}_{\reg{I}} \otimes \proj{\mathsf{EPR}^k}_{\reg{B'B}} \Big) \Bigg] \Bigg\}.
\end{align*}
Moreover, we define the value of the monogamy game $\mathsf{G}_{\text{\tiny{BH}}}$ as the optimal winning probability
$$
\omega(\mathsf{G}_{\text{\tiny{BH}}}) := \underset{\mathsf{S} = \big(\big\{\vec{H}_{x}^{U_\theta,U_{\theta}^\dag}\},\{\vec{R}_{x}^{U_\theta,U_{\theta}^\dag}\big\}\big)}{\sup} \,\,\omega_{\mathsf{S}}(\mathsf{G}_{\text{\tiny{BH}}}).
$$
\end{definition}

\begin{figure}
    \centering
 \begin{protocol}[Black Hole Cloning Game] \label{prot:OCG}\ \\
A black hole cloning game $\mathsf{G}_{\text{\tiny{BH}}} = (\algo H_{\reg{I}},\algo H_{\reg{B}},\algo H_{\reg{B'}},\algo H_{\reg{H}},\algo H_{\reg{R}},\Theta, \{U_{\theta}\}_{\theta \in \Theta},\Phi_{\reg{IB' \rightarrow HR}})$ for a quantum strategy $\mathsf{S} = (\{\vec{H}_{x}^{U_\theta,U_{\theta}^\dag}\}_{\theta \in \Theta, x \in \bit^k},\{\vec{R}_{x}^{U_\theta,U_{\theta}^\dag}\}_{\theta \in \Theta, x \in \bit^k})$ is the following game between a trusted referee called Alice and two colluding and adversarial parties Bob and Charlie.
\begin{enumerate}
  \item (\textbf{Setup phase}) A tripartite quantum state $\rho \in \algo D(\algo H_{\reg{I}} \otimes \algo H_{\reg{B'}} \otimes \algo H_{\reg{B}})$ is prepared, where
  $$
  \rho = \Big( \proj{0^{n-k}}_{\reg{I}} \otimes \proj{\mathsf{EPR}^k}_{\reg{B'B}}  \Big).
  $$
  Here, $k$ denotes the number of qubits in the registers $\reg{B}$ and $\reg{B'}$. Next, Alice receives register $\reg{B}$.

  \item (\textbf{Time-evolution phase})
  A random scrambling unitary $U_\theta$ is selected, where $\theta \sim \Theta$ is chosen uniformly at random, and the internal registers of the black hole evolve according to the unitary channel $\big(U_\theta \fullcirc U_\theta^\dag\big)_{\reg{IB' \rightarrow IB'}}$ which is applied to registers $\reg{IB'}$ of the state $\rho$.
  
Afterwards, the channel $\Phi_{\reg{IB' \rightarrow HR}}$ is applied to registers $\reg{IB'}$ and produces registers $\reg{HR}$.

  \item (\textbf{Guessing phase})
  Charlie and Bob receive the registers $\reg{H}$ and $\reg{R}$, respectively. They also receive oracles for both $U_\theta$ and $U_\theta^\dag$, but may no longer communicate. They independently perform the measurements $\{\vec{H}_{x}^{U_\theta,U_{\theta}^\dag}\}_{x \in \algo X}$ and $\{\vec{R}_{x}^{U_\theta,U_{\theta}^\dag}\}_{x \in \algo X}$ and output a $k$-bit string.

\item (\textbf{Outcome phase}) 
Alice measures $\reg{B}$ is measured in the computational basis, resulting in an outcome $x \in \bit^k$. Charlie and Bob win if they both guessed $x$ correctly.
\end{enumerate}
\end{protocol}   
    \caption{A black hole cloning game.}
    \label{fig:black-hole-game}
\end{figure}

Let us
remark that the initial quantum state is not adversarially prepared by Bob and Charlie (unlike in typical monogamy games~\cite{Tomamichel_2013}); rather, it is generated by an external process (say, nature) over which the players have no control. While this is also true of the $\Phi_{\reg{IB' \rightarrow HR}}$ cloning channel in practice, our bounds will hold even if $\Phi$ is chosen \emph{adversarially} by Bob and Charlie.

\subsection{Bounds On the Value of a Black Hole Cloning Game}\label{sec:blackholebounds}

In this section, we bound the maximal value $\omega(\mathsf{G}_{\text{\tiny{BH}}}) = {\sup}_{\mathsf{S}} \,\omega_{\mathsf{S}}(\mathsf{G}_{\text{\tiny{BH}}})$ of a particular black hole cloning game $\mathsf{G}_{\text{\tiny{BH}}}$ for a unitary $3$-design $\{U_{\theta}\}_{\theta \in \Theta}$ and where we restrict the set of oracle-aided strategies $$\mathsf{S} = \big(\{\vec{H}_{x}^{U_\theta,U_{\theta}^\dag}\}_{\theta \in \Theta, x \in \bit^k},\{\vec{R}_{x}^{U_\theta,U_{\theta}^\dag}\}_{\theta \in \Theta, x \in \bit^k}\big)$$ such that Charlie and Bob only make a single oracle query (to either $U_\theta$ or $U_\theta^\dag$), for any given $\theta \in \Theta$.

\input{fig-circuits}

Let us first give a brief overview of the idea behind our proof. We refer the reader to our technical overview (\Cref{sec:techoverview}) and the associated technical sections for details on each of these steps.

\paragraph{Overview of the proof.} To obtain a bound, we consider a sequence of \emph{hybrid games}:
\begin{itemize}
\item $\mathsf{G}_{\text{\tiny{BH}}}$: This is a black hole cloning game of the form
$$
\mathsf{G}_{\text{\tiny{BH}}} = (\algo H_{\reg{I}},\algo H_{\reg{B}},\algo H_{\reg{B'}},\algo H_{\reg{H}},\algo H_{\reg{R}},\Theta, \{U_{\theta}\}_{\theta \in \Theta},\Phi_{\reg{IB' \rightarrow HR}})
$$
where $\nu=\{U_{\theta}\}_{\theta \in \Theta}$ is an $n$-qubit unitary $3$-design and $\Phi_{\reg{IB' \rightarrow HR}}$ is an arbitrary CPTP map.

\item $\mathsf{G}_{\text{\tiny{MOE}}}$: This is a (regular) monogamy of entanglement game (as in \Cref{sec:MOE}), where
$$\mathsf{G}_{\text{\tiny{MOE}}} =  (\algo H_{\reg{A}},\Theta,\bit^n, \{\vec{A}_y^{\theta}\}_{\theta \in \Theta, x \in \bit^n})
$$
where Alice performs a set of projective measurements $\{\vec A_{y}^\theta\}_{\theta \in \Theta,\,y \in \bit^n}$ acting on the Hilbert space $\algo H_{\reg{A}}= (\mathbb{C}^{2})^{\otimes n}$, for some rank-$1$ projectors $\vec A_y^\theta = \bar{U}_\theta \proj{y} \bar{U}_\theta^\dag$.

\item $\mathsf{G}_{1 \mapsto 2}$: This is a $1 \mapsto 2$ oracular cloning game (as in \Cref{sec:oracularcloning}), where
$$
\mathsf{G}_{1 \mapsto 2} = (1,\algo H_{\reg{A}},\Theta, \bit^n, \{U_\theta\}_{\theta \in \Theta}).
$$

\item $\mathsf{G}_{\mathfrak{F},1}$: This is a different $1 \mapsto 2$ oracular cloning game (as in \Cref{sec:oracularcloning}), where
$$
\mathsf{G}_{\mathfrak{F},1} = (1,\algo H_{\reg{A}},\bit^\lambda, \bit^n, \{\mathsf{U}_{f_\theta} \mathsf{H}^{\otimes n}\}_{\theta\in \bit^\lambda}).
$$
and $\mathfrak{F} = \{f_\theta : \bit^n \rightarrow  \bit\}_{\theta \in \Theta}$ is a family of $6$-wise uniform functions.
\end{itemize}

First, we show that the game $\mathsf{G}_{\text{\tiny{BH}}}$ emerges as a special case of $\mathsf{G}_{\text{\tiny{MOE}}}$ in which we post-select on the event that Alice measures $\{\vec A_{y}^\theta\}_{\theta \in \Theta,\,y \in \bit^n}$ and obtains the outcome $y =x||0^{n-k}$, for some $x \in \bit^k$. Informally, because this event occurs with probability $2^{-n+k}$, we can deduce that that: $$\mathrm{sup}_{\hat{\mathsf{S}}} \,\,\omega_{\hat{\mathsf{S}}}(\mathsf{G}_{\text{\tiny{MOE}}}) \, \geq \, 2^{-n+k} \, \cdot \, {\sup}_{\mathsf{S}} \,\omega_{\mathsf{S}}(\mathsf{G}_{\text{\tiny{BH}}}),$$
where we maximize over the choice of strategies $\hat{\mathsf{S}}$ selected by Bob and Charlie which consist of a tripartite state $\rho$, where $\rho$ is the normalized \emph{Choi state} of the quantum channel $\Phi$, and where Bob and Charlie perform oracle aided measurements with single-query access to $U_\theta$ and $U_\theta^\dag$) on an enlarged Hilbert space. Therefore, in order to obtain an asymptotically optimal bound of the form $\omega(\mathsf{G}_{\text{\tiny{BH}}}) = O(2^{-k})$, it suffices to show that the related monogamy game $\mathsf{G}_{\text{\tiny{MOE}}}$ has a maximal value of $\mathrm{sup}_{\hat{\mathsf{S}}}\,\omega_{\hat{\mathsf{S}}}(\mathsf{G}_{\text{\tiny{MOE}}}) = O(2^{-n})$. \textbf{Note that analyses of monogamy games preceding this work~\cite{Tomamichel_2013, broadbent_et_al:LIPIcs.TQC.2020.4, 10.1007/978-3-030-84242-0_20, Culf_2022, schleppy2025winning} would not suffice here, since they only prove bounds of the form $O(2^{-cn})$ for $c < 1$.} (We will discuss this in some more detail in~\Cref{remark:bhblcomparison}.)

Second, we relate the game $\mathsf{G}_{\text{\tiny{MOE}}}$ to the $1 \mapsto 2$ cloning game $\mathsf{G}_{1 \mapsto 2}$. Here, we use the general result in \Cref{lem:equiv} which allows us to relate this particular class of monogamy games to cloning games. As a result, we find that $\mathrm{sup}_{\hat{\mathsf{S}}} \,\,\omega_{\hat{\mathsf{S}}}(\mathsf{G}_{\text{\tiny{MOE}}}) = \mathrm{sup}_{{\mathsf{S}'}} \,\,\omega_{{\mathsf{S}'}}(\mathsf{G}_{1 \mapsto 2})$, where $\mathsf{S}'$ ranges over the set of analogous oracular cloning strategies, but which involve $\Phi$ as a cloning channel.

Third, we use the insight from our worst-case to average-case reduction in \Cref{thm:worst-case-to-average-case} in order to argue that the $\mathsf{G}_{1 \mapsto 2}$ is at least as hard as the cloning game $\mathsf{G}_{\mathfrak{F},1}$. In particular, we observe that the winning probabilities satisfy $\mathrm{sup}_{\hat{\mathsf{S}'}} \,\,\omega_{\hat{\mathsf{S}'}}(\mathsf{G}_{1 \mapsto 2}) \leq  \mathrm{sup}_{{\mathsf{S}'}} \,\,\omega_{{\mathsf{S}'}}(\mathsf{G}_{\mathfrak{F},1})$, where the set of strategies $\hat{\mathsf{S}'}$ remains the same.

Finally, we invoke Theorem~\ref{thm:tcopycloning} which gives an explicit bound on the game $\mathsf{G}_{\mathfrak{F},1}$. Specifically, we prove that $\mathrm{sup}_{{\mathsf{S}'}} \,\,\omega_{{\mathsf{S}'}}(\mathsf{G}_{\mathfrak{F},1}) \leq O(2^{-n})$, if $\mathfrak{F}$ is a family of $6$-wise uniform functions. \textbf{As remarked above, proving a bound this strong requires the new techniques that we introduced earlier in this work.}

Putting everything together, we then  obtain the aforementioned asymptotically optimal bound of the form $\omega(\mathsf{G}_{\text{\tiny{BH}}}) = O(2^{-k})$ on the black hole cloning game $\mathsf{G}_{\text{\tiny{BH}}}$.
Let us now state our main theorem.

\begin{theorem}\label{thm:black-hole-one-query} Let $n,k \in \N$ be integers such that $n \geq k$ and let $\nu=\{U_\theta\}_{\theta \in \Theta}$ be a unitary $3$-design on $n$-qubits. Then, for any quantum channel $\Phi_{\reg{IB' \rightarrow HR}}$, the maximal single-query value of the black hole cloning game $\mathsf{G}_{\text{\tiny{BH}}} = (\algo H_{\reg{I}},\algo H_{\reg{B}},\algo H_{\reg{B'}},\algo H_{\reg{H}},\algo H_{\reg{R}},\Theta,\{U_\theta\}_{\theta \in \Theta},\Phi_{\reg{IB' \rightarrow HR}})$ is at most 
\begin{align*}
\mathrm{sup}_{\mathsf{S}} \,\,\omega_{\mathsf{S}}(\mathsf{G}_{\text{\tiny{BH}}}) 
 \, = \, O(2^{-k}) \, ,
\end{align*}
where the supremum ranges over all oracle-aided strategies $$\mathsf{S} = \big(\{\vec{H}_{x}^{U_\theta,U_{\theta}^\dag}\}_{\theta \in \Theta, x \in \bit^k},\{\vec{R}_{x}^{U_\theta,U_{\theta}^\dag}\}_{\theta \in \Theta, x \in \bit^k}\big)$$ that only make a single oracle query (to either $U_\theta$ or $U_\theta^\dag$), for any given $\theta \in \Theta$.
\end{theorem}
\begin{proof}
Let $\mathsf{S}$ be any single-query strategy. For convenience, we also assume there exists an auxiliary register $\reg{E}$ (say, the exterior of the black hole) which is initialized to $\ket{0^{n-k}}_{\reg{E}}$ and not touched by any of the processes in the black hole cloning game. This is without loss of generality, since it can always be absorbed into $\Phi$ by re-defining the quantum channel appropriately. Then, it follows that:
\begin{align*}
\omega_{\mathsf{S}}(\mathsf{G}_{\text{\tiny{BH}}})=&\underset{U \sim \nu}{\mathbb{E}} \Bigg\{ \sum_{x \in \bit^k}\mathrm{Tr}\Bigg[ \left(\vec{H}_x^{U,U^\dag} \otimes \vec{R}_x^{U,U^\dag} \otimes\proj{x0^{n-k}}_{\reg{BE}} \right) \Big(\Phi_{\reg{IB' \rightarrow HR}} \otimes \id_{\reg{BE}} \Big)  \\
& \quad\quad\,\,\left(\big(U \fullcirc U^\dag\big)_{\reg{IB' \rightarrow IB'}} \otimes \id_{\reg{BE}} \right)\Big( \proj{0^{n-k}}_{\reg{I}} \otimes \proj{\mathsf{EPR}^k}_{\reg{B'B}} \otimes \proj{0^{n-k}}_{\reg{E}} \Big) \Bigg] \Bigg\}\\
&=2^{n-k}\underset{U \sim \nu}{\mathbb{E}} \Bigg\{ \sum_{x \in \bit^k}\mathrm{Tr}\Bigg[ \left(\vec{H}_x^{U,U^\dag} \otimes \vec{R}_x^{U,U^\dag} \otimes\proj{x0^{n-k}}_{\reg{BE}} \right) \Big(\Phi_{\reg{IB' \rightarrow HR}} \otimes \id_{\reg{BE}} \Big) 
\\
&\quad\quad\,\,\,\quad\quad\quad\left(\big(U \fullcirc U^\dag\big)_{\reg{IB' \rightarrow IB'}} \otimes \id_{\reg{BE}} \right) \Big( \proj{\mathsf{EPR}^n}_{\reg{IB'BE}} \Big) \Bigg] \Bigg\}.
\end{align*}
The above step holds because in the second line the projector $\proj{0^{n-k}}_{\reg{E}}$ will act on one half of the EPR pair $\proj{\mathsf{EPR}^{n-k}}_{\reg{IE}}$, thus collapsing it to $\proj{0^{n-k}}_{\reg I} \otimes \proj{0^{n-k}}_{\reg E}$ and pulling out a factor of $2^{k-n}$.

We continue by using the ricochet property of EPR pairs (formally, Corollary~\ref{cor:CIprojector}) to pull $U, U^\dag$ ``out'' of the cloning channel to obtain something that will look more like a monogamy game. First, consider the channel $\Psi_{\reg{IB' \rightarrow HR}}$ defined as $\Phi \circ (U \fullcirc U^\dag)$ (here, $\circ$ denotes composition). Applying Corollary~\ref{cor:CIprojector} to the channel $\Psi$ yields the following:
\begin{align*}
&2^{n-k}\underset{U \sim \nu}{\mathbb{E}} \Bigg\{ \sum_{x \in \bit^k}\mathrm{Tr}\Bigg[ \left(\vec{H}_x^{U,U^\dag} \otimes \vec{R}_x^{U,U^\dag} \otimes\proj{x0^{n-k}}_{\reg{BE}} \right) \Big(\Phi_{\reg{IB' \rightarrow HR}} \otimes \id_{\reg{BE}} \Big) 
\\
&\quad\quad\,\,\,\quad\quad\quad\left(\big(U \fullcirc U^\dag\big)_{\reg{IB' \rightarrow IB'}} \otimes \id_{\reg{BE}} \right) \Big( \proj{\mathsf{EPR}^n}_{\reg{IB'BE}} \Big) \Bigg] \Bigg\}\\
= &2^{-k}\underset{U \sim \nu}{\mathbb{E}} \Bigg\{ \sum_{x \in \bit^k}\mathrm{Tr}\Bigg[ \left(\vec{H}_x^{U,U^\dag} \otimes \vec{R}_x^{U,U^\dag} \otimes\proj{x0^{n-k}}_{\reg{BE}} \right) J(\Psi)\Bigg]\Bigg\} \\
= &2^{-k}\underset{U \sim \nu}{\mathbb{E}} \Bigg\{ \sum_{x \in \bit^k}\mathrm{Tr}\Bigg[ \left(\vec{H}_x^{U,U^\dag} \otimes \vec{R}_x^{U,U^\dag}\right) \Psi(\proj{x0^{n-k}})\Bigg]\Bigg\} \\
= &2^{-k}\underset{U \sim \nu}{\mathbb{E}} \Bigg\{ \sum_{x \in \bit^k}\mathrm{Tr}\Bigg[ \left(\vec{H}_x^{U,U^\dag} \otimes \vec{R}_x^{U,U^\dag}\right) \Phi(U\proj{x0^{n-k}}U^\dag)\Bigg]\Bigg\}.
\end{align*}
Next, we apply Corollary~\ref{cor:CIprojector} once more, this time to the channel $\Phi$:
\begin{align*}
    &2^{-k}\underset{U \sim \nu}{\mathbb{E}} \Bigg\{ \sum_{x \in \bit^k}\mathrm{Tr}\Bigg[ \left(\vec{H}_x^{U,U^\dag} \otimes \vec{R}_x^{U,U^\dag}\right) \Phi(U\proj{x0^{n-k}}U^\dag)\Bigg]\Bigg\} \\
    = &2^{-k}\underset{U \sim \nu}{\mathbb{E}} \Bigg\{ \sum_{x \in \bit^k}\mathrm{Tr}\Bigg[ \left(\vec{H}_x^{U,U^\dag} \otimes \vec{R}_x^{U,U^\dag} \otimes \bar{U} \proj{x0^{n-k}}_{\reg{BE}} \bar{U}^\dag \right) J(\Phi)\Bigg]\Bigg\} \\
    = &2^{n-k}\underset{U \sim \nu}{\mathbb{E}} \Bigg\{ \sum_{x \in \bit^k}\mathrm{Tr}\Bigg[ \left(\vec{H}_x^{U,U^\dag} \otimes \vec{R}_x^{U,U^\dag} \otimes \bar{U}\proj{x0^{n-k}}_{\reg{BE}} \bar{U}^\dag \right) 
\\
&\quad\quad\,\,\,\,\,\,\quad\quad\quad\Big(\Phi_{\reg{IB' \rightarrow HR}} \otimes \id_{\reg{BE}} \Big)\Big( \proj{\mathsf{EPR}^n}_{\reg{IB'BE}} \Big) \Bigg] \Bigg\}.
\end{align*}
\noindent
For the remainder of the proof, we will bound the final quantity in the expression above; specifically, by relating it to the value of a related monogamy of entanglement game. To this end, we now observe that
\begin{align}
&2^{n-k} \underset{U \sim \nu}{\mathbb{E}}\Bigg\{ \sum_{x \in \bit^k}\mathrm{Tr}\Bigg[ \left(\vec{H}_x^{U,U^\dag} \otimes \vec{R}_x^{U,U^\dag} \otimes \bar{U}\proj{x0^{n-k}}_{\reg{BE}} \bar{U}^\dag \right) 
\nonumber\\
&\quad\quad\,\,\,\,\,\,\quad\quad\quad\Big(\Phi_{\reg{IB' \rightarrow HR}} \otimes \id_{\reg{BE}} \Big)\Big( \proj{\mathsf{EPR}^n}_{\reg{IB'BE}} \Big) \Bigg] \Bigg\} \nonumber\\
&= 2^{n-k}\underset{U \sim \nu}{\mathbb{E}}\Bigg\{ \sum_{x \in \bit^k}\mathrm{Tr}\Bigg[ \left(\tilde{\vec{H}}_{x||0^{n-k}}^{U,U^\dag} \otimes \tilde{\vec{R}}_{x||0^{n-k}}^{U,U^\dag} \otimes \bar{U}\proj{x0^{n-k}}_{\reg{BE}} \bar{U}^\dag \right) \nonumber
\\
&\quad\quad\,\,\,\,\,\,\quad\quad\quad\Big(\Phi_{\reg{IB' \rightarrow HR}} \otimes \id_{\reg{BE}} \Big)\Big( \proj{\mathsf{EPR}^n}_{\reg{IB'BE}} \Big) \Bigg] \Bigg\}\nonumber\\
&\leq 2^{n-k} \underset{\tilde{\mathsf{S}} = \big(\big\{\hat{\vec{H}}_{y}^{U,U^\dag}\},\{\hat{\vec{R}}_{y}^{U,U^\dag}\big\}\big)}{\sup}  \,\, \underset{U \sim \nu}{\mathbb{E}}\Bigg\{ \sum_{y \in \bit^n}\mathrm{Tr}\Bigg[ \left(\hat{\vec{H}}_{y}^{U,U^\dag} \otimes \hat{\vec{R}}_{y}^{U,U^\dag} \otimes \bar{U}\proj{y}_{\reg{BE}} \bar{U}^\dag \right) \nonumber
\\
&\quad\quad\,\,\,\,\,\,\quad\quad\quad\Big(\Phi_{\reg{IB' \rightarrow HR}} \otimes \id_{\reg{BE}} \Big)\Big( \proj{\mathsf{EPR}^n}_{\reg{IB'BE}} \Big) \Bigg] \Bigg\}. \label{eq:bound-bh-game}
\end{align}

We have now transitioned successfully to $\mathsf{G}_{\text{\tiny{MOE}}}$. Because the bound in \Cref{eq:bound-bh-game} applies to any single-query strategy $\mathsf{S}$, we can therefore complete the proof by bounding the black hole cloning game as follows:
\begin{align*}
\omega(\mathsf{G}_{\text{\tiny{BH}}}) &= \underset{\mathsf{S} = \big(\{\vec{H}_{x}^{U,U^\dag}\},\{\vec{R}_{x}^{U,U^\dag}\}\big)}{\mathrm{sup}} \,\,\omega_{{\mathsf{S}}}(\mathsf{G}_{\text{\tiny{BH}}}) \\
&\leq \,  2^{n-k} 
 \,\, \underset{\hat{\mathsf{S}} = \big(\algo H_{\reg{H}},\algo H_{\reg{R}}, \, \rho_{\reg{AHR}} \, ,\{\hat{\vec{H}}_{y}^{U,U^\dag}\},\{\hat{\vec{R}}_{y}^{U,U^\dag}\}\big)}{\mathrm{sup}} \,\,\omega_{\hat{\mathsf{S}}}(\mathsf{G}_{\text{\tiny{MOE}}}) && \text{(by Equation~\eqref{eq:bound-bh-game})}\\
 &=  \,  2^{n-k} 
 \,\, \underset{{\mathsf{S}'} = \big(\algo H_{\reg{H}} \otimes \algo H_{\reg{R}}, \Phi_{\reg{IB' \rightarrow HR}} ,\{\hat{\vec{H}}_{y}^{U,U^\dag}\},\{\hat{\vec{R}}_{y}^{U,U^\dag}\}\big)}{\mathrm{sup}} \,\,\omega_{{\mathsf{S}'}}(\mathsf{G}_{1 \mapsto 2}) && \text{(Lemma~\ref{lem:equiv})}\\
 &\leq \, 2^{n-k} \underset{{\mathsf{S}'} = \big(\algo H_{\reg{H}} \otimes \algo H_{\reg{R}}, \Phi_{\reg{IB' \rightarrow HR}} ,\{\hat{\vec{H}}_{y}^{U,U^\dag}\},\{\hat{\vec{R}}_{y}^{U,U^\dag}\}\big)}{\mathrm{sup}} \,\,\omega_{{\mathsf{S}'}}(\mathsf{G}_{\mathfrak{F},1}) && (\text{\Cref{thm:worst-case-to-average-case}})\\
 &\leq 2^{n-k} \cdot O(2^{-n}) \,\,=  \,\, O(2^{-k}). && (\text{\Cref{thm:tcopycloning}})
 \end{align*}
\end{proof}

We conclude this section with the following remark.

\begin{remark}[Comparison with the Cloning Game Bound by~\cite{broadbent_et_al:LIPIcs.TQC.2020.4}]\label{remark:bhblcomparison}
    Earlier, we made the claim that we require a construction of a cloning game with value $\leq O(2^{-n})$, and that our work is the first to achieve this. At the surface, it might appear that the construction of Broadbent and Lord~\cite{broadbent_et_al:LIPIcs.TQC.2020.4} also achieves this (even if their unclonable encryption construction does not strictly conform to the syntax in Definition~\ref{def:OCG}): they provide a construction with $\mathcal{X} = \bit^n$ and $\Theta = \bit^\secp$ with value $\leq \left(\cos^2(\pi/8)\right)^\secp + O(2^{-n})$. Thus, if $\secp \geq cn$ for some $c > 1$, this value will be $O(2^{-n})$.

    However, the result by~\cite{broadbent_et_al:LIPIcs.TQC.2020.4} still does not suffice for our black hole application. The reason is that their ciphertext states comprise $\secp > n$ qubits and are thus not succinct. Looking through our proof of~\Cref{thm:black-hole-one-query}, what we really need is a cloning game such that messages $x \in \bit^n$ are ``encrypted'' using $m$ qubits and where we can prove a bound of $O(2^{-m})$, which the~\cite{broadbent_et_al:LIPIcs.TQC.2020.4} construction does not satisfy as $m = \secp$. The reason for this is that an adaptation of our worst-case to average-case reduction in~\Cref{sec:worsttoav} would have proceed by applying a Haar random (or pseudorandom) unitary to the \emph{entire} ciphertext state. This forces the Haar random unitary to act on $m$ qubits, so our goal would be to prove a bound of $O(2^{-m})$.
\end{remark}

%% file: fig-circuits.tex
\begin{figure}
\begin{center}
\begin{tikzpicture}
    \node (qc) {
        \begin{quantikz}[row sep=0.1cm, column sep=0.7cm]
            \qw & \gate[7]{V_1} & \gate[4]{U\text{ or }U^\dag} & \gate[7]{V_2} & \qw \\
            \qw & \ghost{V_1} & \ghost{U} & \ghost{V_2} & \qw \\
            \qw & \ghost{V_1} & \ghost{U} & \ghost{V_2} & \qw \\
            \lstick{} & \ghost{V_1} & \qw      & \ghost{V_2} & \qw \\
            \lstick{$\ket{0}$} & \ghost{V_1} & \qw      & \ghost{V_2} & \qw \\
            \lstick{$\ket{0}$} & \ghost{V_1} & \qw      & \ghost{V_2} & \qw \\
            \lstick{$\ket{0}$} & \ghost{V_1} & \qw      & \ghost{V_2} & \meter{}
        \end{quantikz}
    };

    \coordinate (left_center) at ($(qc.north west)!0.5!(qc.south west)$);

    \draw [decorate, decoration={brace, mirror, amplitude=10pt}]
    (qc.north west) -- 
    (left_center) 
    node [midway, left=0.5cm] {\parbox{2cm}{Hawking \\radiation $\reg{R}$}};

    \coordinate (right_bottom) at ($(qc.north east)!0.75!(qc.south east)$);

    \draw [decorate, decoration={brace, amplitude=10pt}]
    (right_bottom) -- 
    (qc.south east) 
    node [midway, right=0.5cm] {\parbox{2cm}{Output\\$x \in \bit^k$}};
\end{tikzpicture}
\end{center}
    \caption{Visualization of Bob's quantum computation in our black hole cloning game. He takes the intercepted Hawking radiation in register $\reg{R}$ as input, adds any number of ancilla qubits (in the $\ket{0}$ state) of his choosing, and applies an initial unitary $V_1$ to the entire system. He then makes one oracle query to either $U$ or $U^\dag$, where $U$ is the black hole's scrambling unitary. Finally, he applies an additional unitary $V_2$ then measures the last $k$ qubits to produce his guess $x \in \bit^k$. The diagram for Charlie's strategy would be similar, except the input would consist of the black hole's interior qubits in register $\reg{H}$.}\label{fig:bobcharlieBH}
\end{figure}

%% file: unclonable-enc.tex
\section{Search-Secure Unclonable Encryption}\label{sec:unclonableenc}

In this section, we formally define (succinct, deterministic, search-secure) unclonable encryption, loosely following the terminology introduced by Broadbent and Lord~\cite{broadbent_et_al:LIPIcs.TQC.2020.4}. We will use $\theta \in \bit^\secp$ (rather than $k$, which we are already using in~\Cref{sec:black-hole-games} to denote the number of EPR pairs in a black hole cloning game) to denote the secret key.

\subsection{Definitions}

\begin{definition}[Unclonable Encryption]
Let $\secp \in \N$ be the security parameter and let $n := n(\secp)$ be some polynomial in $\secp$. A \textbf{succinct and deterministic} unclonable encryption scheme $(\sUE)$ is a tuple $(\KeyGen,\Enc, \Dec)$ consisting of the following QPT algorithms:
\begin{itemize}
    \item $\KeyGen(1^\secp, 1^n):$ takes as input $1^\secp, 1^n$ and outputs $\theta \in \bit^\secp$.

    \item $\Enc(\theta \in \bit^\secp, x \in \bit^n):$ on input $(\theta, x)$, it outputs a pure ciphertext state $\ket{\psi_x^\theta}$. \textbf{We require $\Enc(\theta, x)$ to deterministically output $U_\theta\ket{x}$, for some unitary $U_{\theta, n} \in \algo{U}(2^n)$.} Thus the ciphertext state must also comprise $n$ qubits.

    \item $\Dec(1^n, \theta \in \bit^\secp,\rho)$: on input $\theta$ and a quantum state $\rho$, it outputs $x' \in \bit^n$.
\end{itemize}
We require the following correctness property: for any $\secp, n$, it holds that
$$
\Pr\left[\Dec\big(1^n, \theta,\proj{\psi_x^\theta} \big) = x \, : \, \substack{
\theta \leftarrow \KeyGen(1^\secp, 1^n)\\
\ket{\psi_x^\theta} \leftarrow \Enc(\theta, x)\\
} \right] =1.
$$
\textbf{Succinctness} is implicit in our requirement that the key length only depends on $\secp$ rather than $n$.
\end{definition}

\begin{definition}[$t \mapsto t+1$ $\sUE$ security]\label{defn:uesecurity}

Let $(\KeyGen,\Enc, \Dec)$ be a $\sUE$ scheme, and $t \in \N$ a positive integer. Consider the following experiment between a challenger and an adversary $(\Phi,\algo P_1,\ldots,\algo P_{t+1})$ consisting of a cloner $\Phi$ and $t+1$ players $\algo P_1,\ldots \algo P_{t+1}$
who are not allowed to communicate:

\begin{enumerate}
    \item The challenger runs $\theta \leftarrow \KeyGen(1^\secp, 1^n)$. Next, the challenger samples $x \gets \bit^n$ and runs $\Enc(\theta, x)$ to obtain the ciphertext $\ket{\psi_x^\theta}$, and sends $t$ copies $\ket{\psi_x^\theta}^{\otimes t}$ of the state to the cloner $\Phi$.

    \item The cloner $\Phi$ applies any quantum channel to $\ket{\psi_x^\theta}^{\otimes t}$ in registers $\reg{A_1 \dots A_t}$ and then splits the resulting state into $t+1$ registers $\reg{B_1},\ldots, \reg{B_{t+1}}$. Finally, $\Phi$ sends $\reg{B_i}$ to player $\algo P_i$.

    \item The players $\algo P_1,\ldots, \algo P_{t+1}$ receive $\theta$ and output their guesses for $x$, and win if they all guess correctly.
\end{enumerate}
We say that $(\KeyGen,\Enc,\Dec)$ satisfies statistical (respectively, computational) $t \mapsto t+1$ and $\epsilon(t, \secp, n)$-$\sUE$ security if, for any computationally unbounded (respectively, computationally bounded) adversary $(\Phi,\algo P_1,\ldots, \algo P_{t+1})$, where each $\algo P_i$ is an ensemble of positive-operator valued measurements $\{\vec P_{i,x}^{\theta}\}_{x,\theta}$,
$$
\underset{(x,\theta) \sim \KeyGen(1^n)}{\mathbb{E}} \, \Tr{\left(\vec P_{1,x}^{\theta} \ot \ldots \ot \vec P_{t+1, x}^\theta \right) \Phi_{\reg{A_1 \dots A_t \rightarrow B_1 \ldots B_{t+1}}}\left( \proj{\psi_x^\theta}^{\ot t}_{\reg{A_1 \dots A_t}} \right)} \leq O\left(\epsilon(t, \secp, n)\right).
$$

\end{definition}

The $1 \mapsto 2$ $\sUE$ security experiment is visualized in Figure~\ref{fig:nocloning}. In the following definition, we also define an \emph{oracular} version of this security experiment.

\begin{figure}[!htb]
\begin{center}
{\small
\begin{tikzpicture}
\draw (2,-0.5) rectangle (3.3,0.5) node [pos=.5]{${\bf Ch}$}; 
  \draw (5,-0.5) rectangle (6.3,0.5) node [pos=.5]{${\Phi}$}; 
  \draw[->] (3.3,0) -- node[above]{{\small $\ket{\psi_x^\theta}$}} (5,0);
  \draw[->,dashed] (5.65,-0.5) -- node[left]{} (5,-1.5);
  \draw[->,dashed] (5.65,-0.5) -- node[right]{} (6.3,-1.5);
  \draw (4.5,-1.5) rectangle (5.65,-2.5) node [pos=.5]{$\algo P_1$};
  \draw[->] (4,-2) node[left]{$\theta$} -- (4.5,-2);
  \draw[->] (5.05,-2.5) -- (5.05,-3) node[below]{$x_{1}$};
  \draw[->] (7.65,-2) node[right]{$\theta$} -- (7.15,-2);
  \draw[->] (6.55,-2.5) -- (6.55,-3) node[below]{$x_{2}$};
  \draw[->,dashed] (5.65,-0.5) -- (6.3,-1.5);
  \draw (6,-1.5) rectangle (7.15,-2.5) node [pos=.5]{${\algo P_2}$};

 \end{tikzpicture}
}
\end{center}
\caption{The $1 \mapsto 2$ $\sUE$ experiment. A cloner $\Phi$ splits a state $\ket{\psi_x^\theta}$ prepared by the challenger {\bf Ch} into two parts, one is sent to $\mathcal{P}_1$ and one is sent to $\mathcal{P}_2$. Given $\theta$, $\mathcal{P}_1$ and $\mathcal{P}_2$ then output their guesses $x_{1}$ and $x_{2}$ for $x$.}\label{fig:nocloning}
\end{figure}

\begin{definition}[$t \mapsto t+1$ oracular $\sUE$ security]\label{def:ueoraclesecurity}
    We say that $(\KeyGen, \Enc, \Dec)$ satisfies statistical (respectively, computational) $t \mapsto t+1$ $\epsilon$-$\sUE$ oracular security under the same conditions as Definition~\ref{defn:uesecurity}, with the following modification: in the final phase, the players $\mathcal{P}_1, \ldots \mathcal{P}_{t+1}$ do not receive $\theta$ directly. Instead, they receive query access to the unitary $U_{\theta, n}$ computing $\Enc(\theta, \cdot)$ as well as its inverse $U_{\theta, n}^\dag$.

    We say that $(\KeyGen, \Enc, \Dec)$ satisfies the weaker notion of $(\epsilon, q)$-$\sUE$ oracular security if each of the players may only make a total of $\leq q$ queries to $U_{\theta, n}$ and $U_{\theta, n}^\dag$.

\end{definition}

\begin{remark}
    We emphasize that our work is the first to consider $t \mapsto t+1$ cloning games for $t > 1$: not only is prior work limited to $1 \mapsto 2$ cloning games, all existing unclonable cryptography (based on BB84 states or coset states) becomes completely insecure if $t$ is allowed to grow polynomially (see ~\Cref{sec:prevtechniques} and~\cite{ananth2024revocableencryptionprogramsmore} for more details).

    While we are only able to prove security for $t = o(n/\log n)$ (see Theorem~\ref{thm:uemulticopy}), we reiterate that our construction could very well be secure for $t$ that is an arbitrary polynomial in $n$ (unlike previous constructions based on BB84 states~\cite{broadbent_et_al:LIPIcs.TQC.2020.4} and coset states~\cite{10.1007/978-3-030-84242-0_20}).
\end{remark}

\begin{remark}[Comparison with~\cite{broadbent_et_al:LIPIcs.TQC.2020.4}]
    At first glance, in the $t = 1$ case this notion may already appear to have been achieved by the construction by Broadbent and Lord~\cite{broadbent_et_al:LIPIcs.TQC.2020.4}, which achieves a security bound of $\epsilon(\secp, n) = \frac{9}{2^n} + \left(\cos^2 \frac{\pi}{8}\right)^\secp$. At a high level, they compose a $\secp$-bit BB84 cloning game as in~\cite{Tomamichel_2013} with a PRF-based one-time pad (see Remark~\ref{remark:justifyingdet} below for details). However, their construction has two aspects which we would like to improve on:
    \begin{itemize}
        \item Their construction assumes the existence of post-quantum pseudorandom functions. We would like to instantiate a $\sUE$ scheme assuming the milder notions of pseudorandom quantum states or unitaries.

        \item Their encryption is randomized. The natural deterministic analogue of this would be the BB84-based encryption scheme without the PRF, which has security $\left(\cos^2 \frac{\pi}{8}\right)^n$ but is no longer succinct as this would use keys of length $n$.
    \end{itemize}
    \noindent
    We note that a common shortcoming of both the work by~\cite{broadbent_et_al:LIPIcs.TQC.2020.4} and our work is the reliance on oracles for proving security.
\end{remark}

\begin{remark}[Justifying Deterministic Encryption]\label{remark:justifyingdet}
    Our reason for focusing on $\mathsf{UE}$ schemes with \emph{deterministic unitary} encryption is in order to be able to naturally instantiate the oracular security setting. Regardless, any such scheme comes with the obvious shortcoming that it is deterministic, and hence does not satisfy the ideal notions of indistinguishable or unclonable-indistinguishable security (as defined in~\cite{broadbent_et_al:LIPIcs.TQC.2020.4}). Nevertheless, we argue that a deterministic scheme satisfying unclonable search security can plausibly be bootstrapped to a scheme satisfying unclonable-indistinguishable security.

    One such bootstrapping transformation was proposed by Broadbent and Lord~\cite{broadbent_et_al:LIPIcs.TQC.2020.4}: suppose we have an unclonable search-secure scheme $\mathsf{SearchEnc}$. Then to encrypt $x \in \bit^n$ under secret keys $(k, \theta) \in \bit^\secp \times \bit^\secp$, sample a random PRF seed $r \in \bit^\secp$ and output the classical string $x \oplus \mathsf{PRF}(k, r)$ together with the quantum state $\ket{\mathsf{SearchEnc}(\theta, r)}$. Broadbent and Lord~\cite{broadbent_et_al:LIPIcs.TQC.2020.4} also provided some mild evidence that this may be secure if the PRF is instantiated with a random oracle.
    
    We emphasize that \emph{proving} the unclonable-indistinguishable security of this transformation (or a similar one) is a notoriously difficult open problem~\cite{kundu2023deviceindependentuncloneableencryption, 10.1007/978-3-031-38554-4_3, aky24}. Our point is just that studying unclonable encryption in the weaker search-secure setting is still an interesting and relevant cryptographic problem, and hence that restricting attention to the case of deterministic unclonable encryption is also interesting.
\end{remark}

\noindent
Finally, we recall from Remark~\ref{remark:veryrestricted} that the oracular $\sUE$ setting is still quite expressive, even if we only allow each player \emph{one} query. In particular, it would be sufficient for recovering $x$ from $\ket{\psi_x^\theta}$, and thus there is still a trivial strategy that succeeds with probability $2^{-n}$: the cloner forwards their copies to the first $t$ players, and nothing to player $t+1$. The first $t$ players can decrypt and output $x$, and player $t+1$ will simply guess randomly.

\subsection{Constructions}

\begin{construction}\label{const:PRUUE}
    Let $\mathfrak{U} = \left\{\mathfrak{U}_n\right\}_{n \in \N}$ be some ensemble of unitaries (we will specify what $\mathfrak{U}$ should be later). Recall that $\mathfrak{U}_n = \left\{U_{\theta, n}\right\}_{\theta \in \bit^\secp}$. Our construction proceeds as follows:
    \begin{itemize}
        \item $\KeyGen(1^\secp, 1^n)$: sample and output uniformly random $\theta \in \bit^\secp$.
        \item $\Enc(\theta, x)$: output $U_{\theta, n} \ket{x}$.

        \item $\Dec(1^n, \theta, \rho)$. First apply the unitary channel $U_{\theta, n}^\dag \fullcirc U_{\theta, n}$ to obtain the state $U_{\theta, n}^\dag \rho U_{\theta, n}$. Now measure in the standard basis and output the result.
    \end{itemize}
\end{construction}
Correctness is clear, so we now prove security in two different settings. First, we show $1 \mapsto 2$ security assuming the existence of pseudorandom unitaries, thus placing unclonable encryption in MicroCrypt:
\begin{theorem}\label{thm:UEfromPRU}
    If $\mathfrak{U}$ is a pseudorandom unitary (as defined in Definition~\ref{def:PRU}), then the $\sUE$ scheme specified in Construction~\ref{const:PRUUE} satisfies computational $1 \mapsto 2$ $\epsilon$-$\sUE$ oracular security, where $\epsilon = \negl(\secp) + \left(\cos^2 \frac{\pi}{8}\right)^n$.
\end{theorem}
\begin{proof}
    We consider a series of hybrid games:
    \begin{itemize}
        \item $\mathsf{Hyb}_0$: This is the $1 \mapsto 2$ oracular $\sUE$ security game, as defined in Definition~\ref{def:ueoraclesecurity}.

        \item $\mathsf{Hyb}_1$: In step 1 of the $\sUE$ security game, the challenger also only has oracle access to $U_{\theta, n}, U_{\theta, n}^\dag$. To generate the ciphertext state $\ket{\psi_x^\theta}$, they query the oracle for $U_{\theta, n}$ on input $\ket{x}$.


        \item $\mathsf{Hyb}_2$: Now, the unitary $U$ is sampled as follows: sample a string $b \gets \bit^n$ uniformly at random, and output $\mathsf{H}^b$ (which applies a Hadamard at every position where the corresponding entry of $b$ is 1).
    \end{itemize}
    For $i = 0, 1, 2$, let $\omega(\mathsf{Hyb}_i)$ denote the probability of the players winning the security game in $\mathsf{Hyb}_i$. Then we observe the following:
    \begin{itemize}
        \item $\omega(\mathsf{Hyb}_0) = \omega(\mathsf{Hyb}_1)$, as these two are functionally equivalent.

        \item $|\omega(\mathsf{Hyb}_1) - \omega(\mathsf{Hyb}_2)| \leq \negl(\secp)$ by the worst-case to average-case reduction in Section~\ref{sec:worsttoav}.

        \item $\omega(\mathsf{Hyb}_2) \leq \cos^2\left(\frac{\pi}{8}\right)^n$: this is exactly the BB84 security game. This security bound essentially follows from analysis by~\cite{Tomamichel_2013}, and was formally shown in~\cite[Corollary 2]{broadbent_et_al:LIPIcs.TQC.2020.4}.
    \end{itemize}
    The conclusion follows.
\end{proof}
\begin{remark}\label{remark:uefromcoset}
    We note that we could just as easily have used the subspace coset monogamy game~\cite{10.1007/978-3-030-84242-0_20, Culf_2022} and its analysis by~\cite{schleppy2025winning} in the place of the BB84 cloning game in $\mathsf{Hyb}_2$ to obtain a slightly stronger upper bound of $\negl(\secp) + O(2^{-n/4})$.

    In some more detail: the unitary will be indexed by a linear subspace $A$ of $\mathbb{F}_2^n$ with $\dim A = n/2$. The unitary will be that which takes an $n$-bit string $x$ as input, interprets this as a pair of cosets $s+A, s'+A^\perp$, and outputs $$\ket{A_{s, s'}} = \frac{1}{2^{n/4}}\sum_{a \in A} (-1)^{\langle a, s' \rangle} \ket{a+s}.$$
\end{remark}

\noindent
Secondly, we show assuming the existence of post-quantum one-way functions that Construction~\ref{const:PRUUE} can be instantiated to satisfy multi-copy security against query-bounded adversaries:
\begin{theorem}\label{thm:uemulticopy}
    For any $\secp, n$, let $n' = n - \omega(\log \secp)$ and consider $t$ such that \begin{align*}
        t &\leq O\left(\frac{n'}{\log n'}\right) \Leftrightarrow \exp(O(t \log t)) \leq 2^{n'} = 2^n \cdot \negl(\secp).
    \end{align*}
    Let $\left\{f_{\theta, n}: \theta \in \bit^\secp\right\}$ be a post-quantum pseudorandom function family from $\bit^n \rightarrow \bit$. Then we define $\mathfrak{U}$ by $$U_{\theta, n} = \mathsf{U}_{f_{\theta, n}} \mathsf{H}^{\otimes n}.$$
    Then the $\sUE$ scheme specified in Construction~\ref{const:PRUUE} satisfies computational $t \mapsto t+1$ $(\epsilon, 1)$-$\sUE$ oracular security, where $$\epsilon = \exp(O(t \log t)) \cdot 2^{-n} + \negl(\secp) = \negl(\secp).$$
\end{theorem}
\begin{proof}
    We first pass from a pseudorandom function $f_{\theta, n}$ to a truly random function $f$ at the expense of an additive $\negl(\secp)$ security loss. A random function is $(4t+2)$-wise uniform, so we can then finish using Lemma~\ref{lemma:cloningtoue} and Theorem~\ref{thm:tcopycloning}.
\end{proof}